\documentclass[11pt]{article}

\usepackage{geometry,showlabels, fullpage, authblk}
\usepackage{xr-hyper}
\usepackage{times,enumitem}
\usepackage{hyperref}[]
\hypersetup{
	colorlinks=true,
	linkcolor=blue,
	filecolor=magenta,      
	urlcolor=cyan,
	citecolor=blue,
}
\usepackage{ulem}

\usepackage{url}
\usepackage{subcaption}
\usepackage{bigints}
\usepackage{amsfonts,amsmath,amssymb,amsthm, bbm}
\usepackage{wrapfig}
\usepackage{verbatim,float,url}
\usepackage{graphicx}
\usepackage{subcaption}
\usepackage[round]{natbib}   
\usepackage[english]{babel}
\usepackage{cancel}
\usepackage{color}
\usepackage[dvipsnames]{xcolor}
\usepackage{cleveref}

\usepackage{mathtools}

\newtheorem{theorem}{Theorem}
\newtheorem{lemma}[theorem]{Lemma}
\newtheorem{corollary}[theorem]{Corollary}
\usepackage{mdwlist}	

\newtheorem{remark}{Remark}
\theoremstyle{definition}
\newtheorem{definition}{Definition}

\usepackage{tikz}  
\usetikzlibrary{trees}
\usetikzlibrary{arrows.meta}
\usetikzlibrary{decorations.pathreplacing}

\usepackage{comment}

\def\R{\mathbb{R}}
\def\Z{\mathbb{Z}}

\graphicspath{{figs/}}

\usepackage{algorithm}
\usepackage{algpseudocode}

\algnewcommand\algorithmicinput{\textbf{INPUT:}}
\algnewcommand\INPUT{\item[\algorithmicinput]}
\algnewcommand\algorithmicoutput{\textbf{OUTPUT:}}
\algnewcommand\OUTPUT{\item[\algorithmicoutput]}

\usepackage{mathtools}
\DeclarePairedDelimiter\ceil{\lceil}{\rceil}
\DeclarePairedDelimiter\floor{\lfloor}{\rfloor}

\DeclareMathOperator*{\argmin}{arg\,min}

\theoremstyle{plain}
\newtheorem{assumption}{Assumption}
\allowdisplaybreaks
\usepackage{bm}

\title{Quantile Regression by Dyadic CART}

\author[1]{Oscar Hernan Madrid Padilla}
\author[2]{Sabyasachi Chatterjee}

\affil[1]{\small Department of Statistics, University of California, Los Angeles}
\affil[2]{\small Department of Statistics, University of Illinois at Urbana-Champaign}

\begin{document}
	\maketitle
	
	\begin{abstract}
In this paper we propose and study a version of the Dyadic Classification and Regression Trees (DCART) estimator from \cite{donoho1997cart} for (fixed design) quantile regression in general dimensions. We refer to this proposed estimator as the QDCART estimator. Just like the mean regression version, we show that a) a fast dynamic programming based algorithm with computational complexity $O(N \log N)$ exists for computing the QDCART estimator and b) an oracle risk bound (trading off squared error and a complexity parameter of the true signal) holds for the QDCART estimator. This oracle risk bound then allows us to demonstrate that the QDCART estimator enjoys adaptively rate optimal estimation guarantees for piecewise constant and bounded variation function classes. In contrast to existing results for the DCART estimator which requires subgaussianity of the error distribution, for our estimation guarantees to hold we do not need any restrictive tail decay assumptions on the error distribution. For instance, our results hold even when the error distribution has no first moment such as the Cauchy distribution. Apart from the Dyadic CART method, we also consider other variant methods such as the Optimal Regression Tree (ORT) estimator introduced in~\cite{chatterjee2019adaptive}. In particular, we also extend the ORT estimator to the quantile setting and establish that it enjoys analogous guarantees.  Thus, this paper extends the scope of these globally optimal regression tree based methodologies to be applicable for heavy tailed data. We then perform extensive numerical experiments on both simulated and real data which illustrate the usefulness of the proposed methods. 


	\end{abstract}

\section{Introduction}

We consider the problem of nonparametric quantile regression in general dimensions and specifically consider the setting of fixed/lattice design regression or array denoising. In this setting, we are given an array of independent random variables $y \in \mathbb{R}^{L_{d,n}}$ where $L_{d,n}$ is the $d$-dimensional square lattice or grid graph with nodes indexed by $\{1,\ldots,n\}^d$.  Then the goal is to estimate the \textit{true $\tau$ quantile array} $\theta^*$ where
\[
\theta^*_i \,=\,  \underset{a \in \mathbb{R}}{\arg \min}\,\mathbb{E}\left(\rho_{\tau}\left(y_i -a \right)\right) 
\]
for all $i \in  L_{d,n}$, $\tau \in (0,1)$ is a fixed quantile level and where $\rho_{\tau}(x) = \max\{\tau x, (1-\tau)x \}$ is the usual piecewise linear convex quantile loss function. For example, when $\tau  = 0.5$, our setting here amounts to estimate the true median array of the noisy array $y.$ The model here can be called \textit{the quantile sequence model}. This generalizes the usual Gaussian sequence model where the quantile $\tau$ is taken to be $0.5$ and the distribution of $y$ is taken to be multivariate normal with the covariance matrix a multiple of identity.

Assuming lattice design is common practice for studying non parametric regression estimators and clearly our setting is relevant for image denoising and computer vision when $d = 2$ or $3$. The problem of estimating the true signal $\theta^*$ becomes meaningful when the true array satisfies some additional structure so that the effective parameter size is much less even though the actual number of unknown parameters is the same as the sample size. Structured signal denoising is a standard problem and arises in several scientific disciplines, e.g see applications in 
computer vision \citep[e.g.][]{bian2017gms, wirges2018object}, medical imaging \citep[e.g.][]{lang2014adaptive}, and neuroscience \citep[e.g.][]{tansey2018false}.

In this paper we are interested in scenarios where $\theta^*$ is (or is close to) a piecewise constant array on a rectangular partition of $L_{d,n}$. For mean regression, Dyadic classification and regression trees (DCART) method introduced in 
\cite{donoho1997cart} is known to be computationally efficient while achieving adaptively minimax rate optimal rates of convergence for classes of signals which are piecewise constant in a rectangular partition of $L_{d,n}$; see~\cite{chatterjee2019adaptive} for a thorough study of statistical adaptivity of DCART. However, since we are interested in quantile regression, we would like to propose a quantile version of Dyadic CART.

The most natural way to define the quantile version of Dyadic  CART estimator is as follows:
	\begin{equation}\label{eqn:initestimator}
		\hat{\theta}_{rdp} = \argmin_{\theta \in \mathbb{R}^{L_{d, n}}}\left\{ \sum_{i=1}^n \rho_{\tau}(y_i - \theta_i)  + \lambda k_{rdp}(\theta)\right\}, 
	\end{equation}
where we define $k_{rdp}(\theta)$ as the smallest natural number for which there exists  a dyadic partition $\Pi$  of $L_{d,n}$ such that $\theta $  is constant in each element of $\Pi$  and    $\vert  \Pi\vert = k_{rdp}(\theta)$. The estimator we propose and study in this article is a slightly modified version of the above estimator in~\eqref{eqn:initestimator}. We refer to this proposed estimator as the QDCART estimator. The precise definition of our estimator and the meaning of a dyadic rectangular partition of $L_{d,n}$ and the complexity parameter $k_{rdp}(\theta)$ will be given in Section \ref{sec:description}.

The usual mean regression version of Dyadic CART estimator is a computationally feasible decision tree method proposed first in~\cite{donoho1997cart} in the context of regression on a two-dimensional grid design. This estimator optimizes the same criterion as in~\eqref{eqn:initestimator} except that the quantile loss is replaced by the usual squared loss. Subsequently after~\cite{donoho1997cart}, several papers have used ideas related to dyadic partitioning for regression, classification and density estimation; e.g see~\cite{nowak2004estimating},~\cite{scott2006minimax},~\cite{blanchard2007optimal},~\cite{willett2007multiscale}. Recently, the paper~\cite{chatterjee2019adaptive} generalized the Dyadic CART estimator to general dimensions and to higher orders and studied the ability of Dyadic CART to estimate piecewise constant signals of various types. Dyadic CART has also been recently used for recovering level sets of piecewise constant signals; see~\cite{padilla2021lattice}. It is fair to say that the two most important facts about the usual mean regression version of Dyadic CART are:

\begin{itemize}
	\item The Dyadic CART estimator attains an oracle risk bound; e.g see Theorem $2.1$ in~\cite{chatterjee2019adaptive}. This oracle risk bound can then be used to show that the Dyadic CART estimator is nearly minimax rate optimal for several function classes of interest.
	
	\item The Dyadic CART estimator can be computed by a bottom up dynamic program with computational complexity linear in the sample size, see Lemma $1.1$ in~\cite{chatterjee2019adaptive}.
\end{itemize}

These two properties of the Dyadic CART make it a very attractive signal denoising method.
However, the oracle risk bound satisfied by Dyadic CART is known to hold only under sub-Gaussian errors. A natural question is whether it is possible to define a version of Dyadic CART which satisfies a result like Theorem $2.1$ in~\cite{chatterjee2019adaptive} without any tail decay assumptions on the error distribution and still retains essentially linear time computational complexity? This is the main question that motivated the research in this article and naturally led us to study a quantile regression version of Dyadic CART. The results in this paper answer our question as affirmative. We now summarize our results.

\begin{itemize}
\item Theorem~\ref{thm0} gives an oracle risk bound for the QDCART estimator proposed in this paper. 
The advantage of our risk bound is that it holds under an extremely mild assumption (see Assumption~\ref{as1} in Section~\ref{sec:results}) on the distribution of the error or noise variables. For example, our risk bound holds when the error distribution is heavy tailed like the Cauchy distribution for which even the first moment does not exist. In contrast, Theorem $2.1$ in~\cite{chatterjee2019adaptive} heavily relies on the subgaussian nature of the errors. Therefore, our main contribution here is to establish the robustness of the quantile version of Dyadic CART to heavy tailed errors. The result in Theorem~\ref{thm0} can be thought of as generalizing Theorem $2.1$ in~\cite{chatterjee2019adaptive} to the heavy tailed setting.

\item Once the oracle risk bound in Theorem~\ref{thm0} has been established, it has been shown in~\cite{chatterjee2019adaptive} how this automatically implies that the QDCART estimator would be minimax rate optimal for several function/signal classes of interest. In particular, this opens the door for us to establish minimax rate optimality of our QDCART estimator over the space of piecewise constant and/or bounded variation arrays. We provide these results in Section~\ref{sec:bddvar}. At the risk of reiterating, the state of the art mean regression estimators for estimating piecewise constant and/or bounded variation arrays typically require subgaussianity of the errors while the QDCART estimator is robust to heavy tailed error distributions. A natural competing quantile regression estimator to QDCART is the Quantile Total Variation Denoising estimator studied in~\cite{padilla2020risk}. 
Just like for the corresponding mean regression counterparts, we argue in Section~\ref{sec:bddvar} that the QDCART estimator has certain advantages over the Quantile Total Variation Denoising estimator, not least the fact that QDCART is computable in essentially linear time in any dimension whereas Quantile Total Variation Denoising is not known to have linear time computational complexity in multivariate settings ($d > 1$).

\item We explain in Section~\ref{sec:ort} that our proof technique for Theorem~\ref{thm0} can also be used to derive similar risk bounds for other variants of the QDCART estimator. For example, in~\cite{chatterjee2019adaptive} the Optimal Regression Tree (ORT) estimator was introduced and studied for mean regression. This ORT estimator is similar to the Dyadic CART estimator with the same optimization objective function except that the optimization is done over all decision trees or hierarchical partitions (not necessarily dyadic). It was then shown in~\cite{chatterjee2019adaptive} that this estimator attains a better risk bound than Dyadic CART in general. However, its computational complexity is slower and scales like $O(N^{2 + 1/d})$ in $d$ dimensions in contrast to the $O(N)$ computational complexity of Dyadic CART. The proof techniques of this paper actually also imply that a quantile version of the ORT estimator can be defined which will enjoy the corresponding risk guarantee. We prefer to present our main results only for QDCART to make the exposition short and because of its significantly better computational complexity.

\item We give a bottom up dynamic programming algorithm which can exactly compute the QDCART estimator. This algorithm is similar to the original one proposed for the DCART estimator in~\cite{donoho1997cart}, suitably adapted to our setting. The computational complexity of our algorithm is $O(N (\log N)^d)$ (see Theorem~\ref{thm:compu}) which is slightly slower than the $O(N)$ computational complexity of the DCART estimator. This extra log factor in the computation seems unavoidable to us because of the need to compute and propagate quantiles of various dyadic rectangles. Our algorithm is described in detail in Section~\ref{sec:compu}.

\end{itemize}

\subsection{Outline}

The rest of the paper is organized as follows.  Section  \ref{sec:description}  presents the precise definition of the QDCART estimator. The main theoretical result (Theorem~\ref{thm0}) of this paper is then presented in Section~\ref{sec:results}. We then provide implications of our main result (Theorem~\ref{thm0}) to the class of bounded variation signals in Section \ref{sec:bddvar}, and to the class of piecewise constant signals in Section \ref{sec:pc}.  Section \ref{sec:ort} discusses the quantile optimal tree regression (QORT) estimator. Section~\ref{sec:discuss} is a discussion section. Section \ref{sec:proof} presents a overview of the proof of our main theorem. In Section \ref{sec:comp}, we compare our theoretical guarantees for QDCART with what is known for a natural competitor estimator, the quantile total variation denoising estimator.
Section \ref{sec:compu} provides the details of our algorithm for implementing QDCART. Section \ref{sec:experiments} contains extensive numerical results in both simulated and real data examples. Finally, Section~\ref{sec:proofs} contains the full proofs of all our theoretical results.




\section{Description of QDCART Estimator}
\label{sec:description}
In this section, we precisely describe the QDCART estimator we propose to study. Let's first introduce some notation which we will use throughout this article.
For any fixed dimension $d \geq 1$, we denote our \textit{sample size} by $N = n^d$ which is the size of the lattice $L_{d, n}$. Let us denote the discrete interval of positive integers as $[a,b] := \{i \in \Z_{+}: a \leq i \leq b\}$ where $\Z_{+}$ denotes 
the set of positive integers. For a positive integer $n$ we 
also denote the set $[1,n]$ by just $[n].$  For squences  $a_n$ and $b_n$ we write $a_n =  O(b_n)$ if there exists a positive  constant $c>0$ such that  $a_n \leq   c b_n$. If instead $a_n \leq b_n (\log n)^l$ for a positive constant $l $ then we write $a_n  = \tilde{O}(b_n)$.  A subset $R 
\subset L_{d,n}$ is called an \textit{axis aligned 
	rectangle} if $R$ is a product of 
discrete intervals, i.e. $R = \prod_{i = 1}^{d} [a_i,b_i].$ 
Henceforth, we will just use the word rectangle to denote an 
axis aligned rectangle. The size of a rectangle $R = \prod_{i = 1}^{d} [a_i,b_i]$ is denoted by  $\vert R \vert$ and defined as 
\[
 \vert R \vert   =  \prod_{i = 1}^{d} (b_i-a_i+1).
\]

Let us define a \textit{rectangular 
	partition} of $L_{d,n}$ to be a set of rectangles 
$\mathcal{R}$ such that (a) the rectangles in $\mathcal{R}$ 
are pairwise disjoint and (b) $\cup_{R \in \mathcal{R}} R = 
L_{d,n}.$

Let us consider a generic discrete interval $[a,b].$ We define a \textit{dyadic split} of the interval to be a split of the interval $[a,b]$ into two  intervals fo equal size. We assume that the interval has even size for ease of exposition. If not, 
then one can set forth a convention for defining the middle point and then follow it throughout.  A dyadic partition of $L_{d,n}$ is constructed iteratively as follows. Starting from the trivial partition which is just $L_{d,n}$ itself, we can create a refined partition by dyadically splitting $L_{d,n}.$  This will result in a partition of $L_{d,n}$ into two rectangles. We can now keep on dividing recursively, generating new partitions. In general, if at some stage we have the partition $\Pi = (R_1,\dots,R_k)$, we can choose any of the rectangles $R_i$ and dyadically split it to get a refinement of $\Pi$ with $k + 1$ nonempty rectangles. \textit{A recursive dyadic partition} (RDP) is any partition reachable by such successive dyadic splitting. Let us denote the set of all recursive dyadic partitions of $L_{d,n}$ as $\mathcal{P}_{rdp}(L_{d,n}).$ Figure \ref{fig5} shows a depiction of  a dyadic partition.

\begin{figure}[htbp!]
	\begin{center}
		\includegraphics[width=1.98in,height=2.08in]{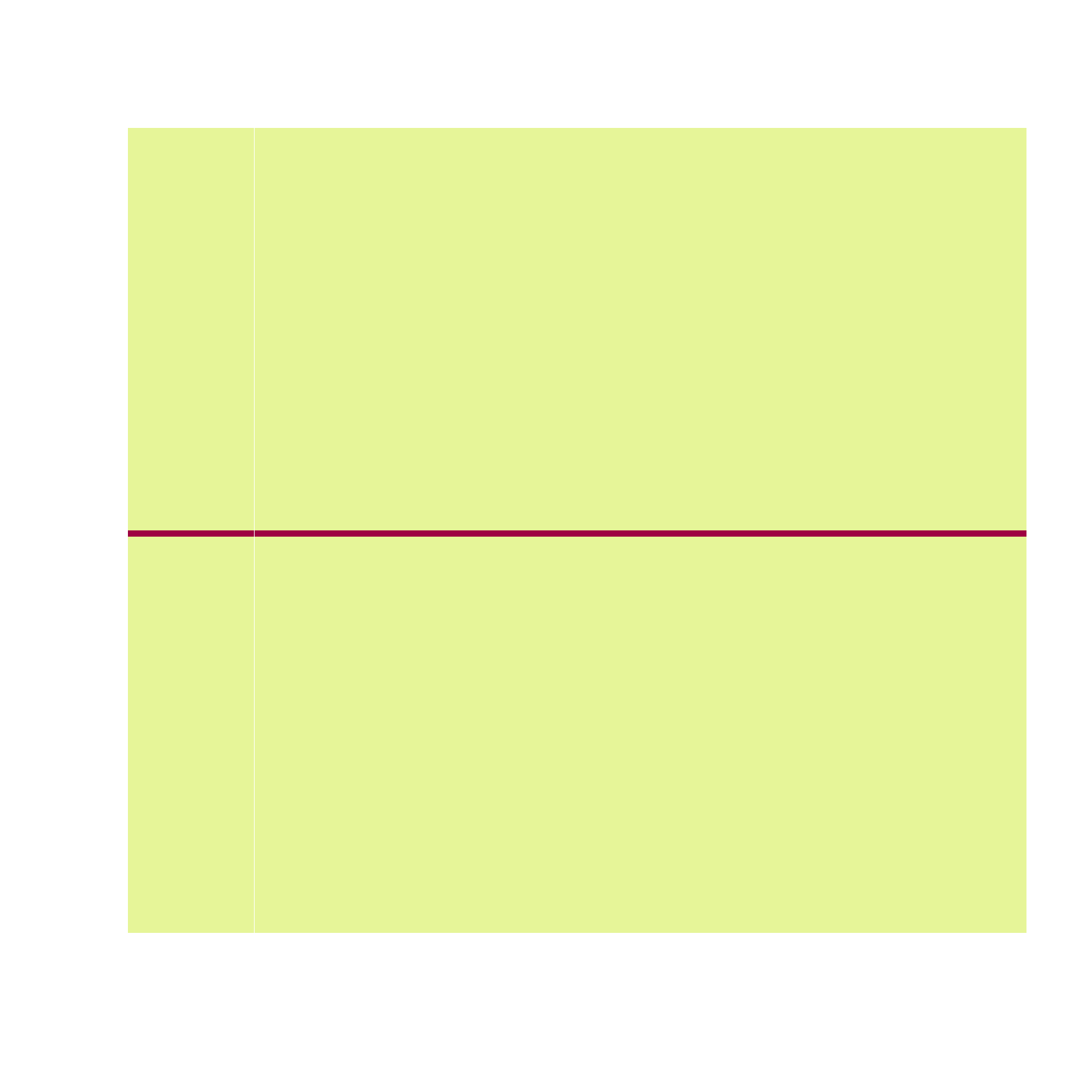}
		\includegraphics[width=1.98in,height=2.08in]{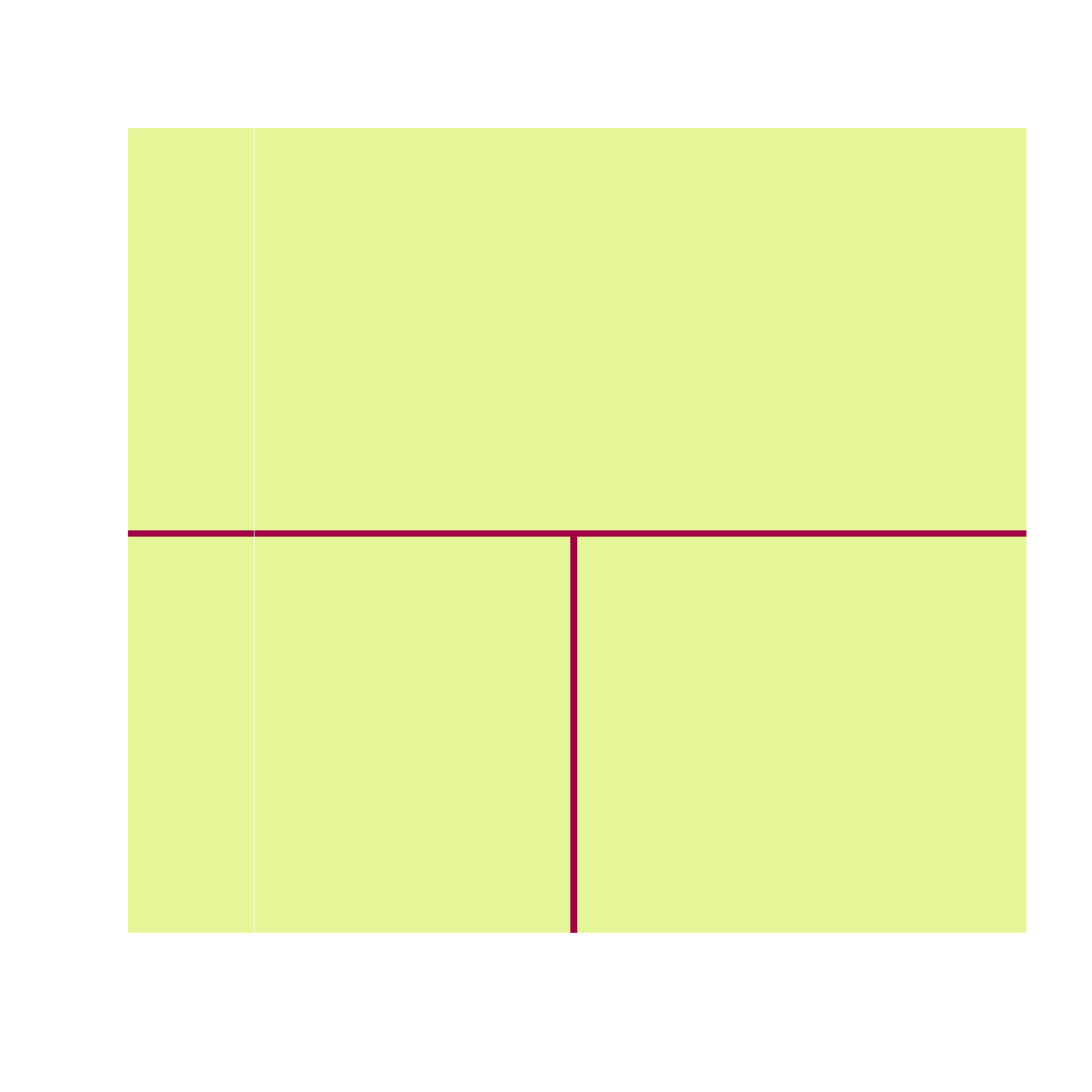}
		\includegraphics[width=1.98in,height=2.08in]{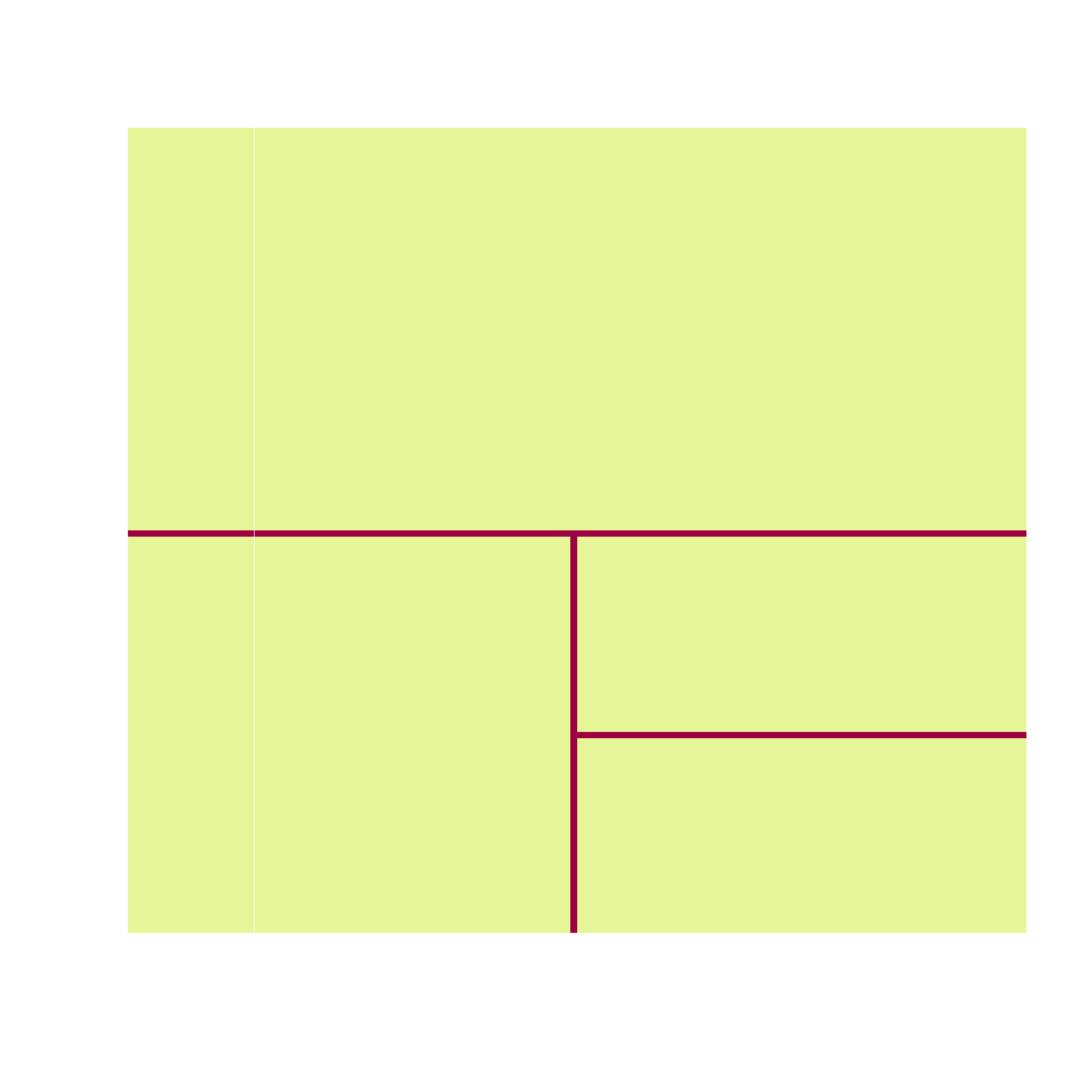}
		\caption{ 		\label{fig5} From left to right the panels show an   example of a sequence of three dyadic splits that lead to a dyadic parition.}
	\end{center}
\end{figure}

\begin{figure}[h!]
	\begin{center}
		\includegraphics[width=2.58in,height=2.08in]{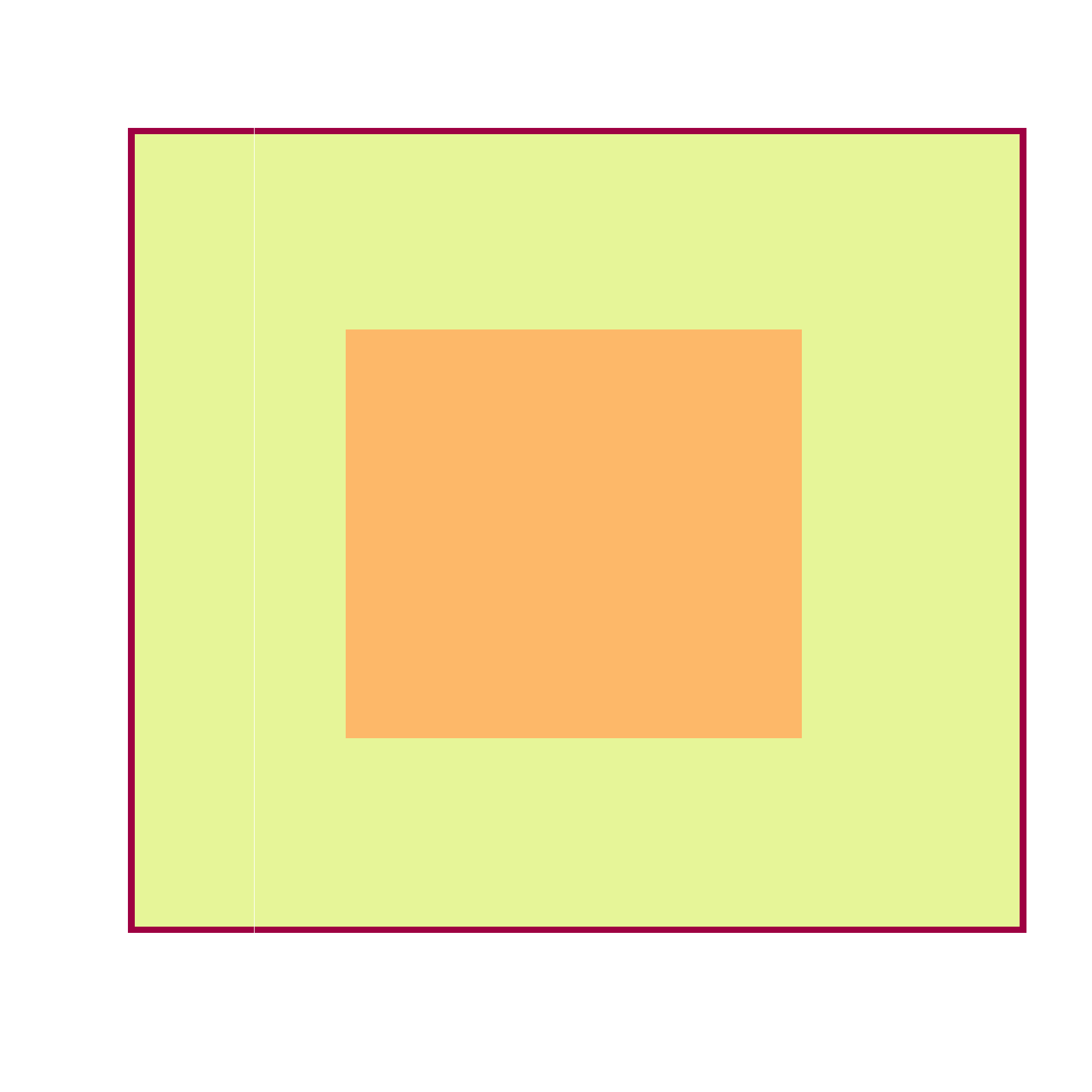}
		\includegraphics[width=2.58in,height=2.08in]{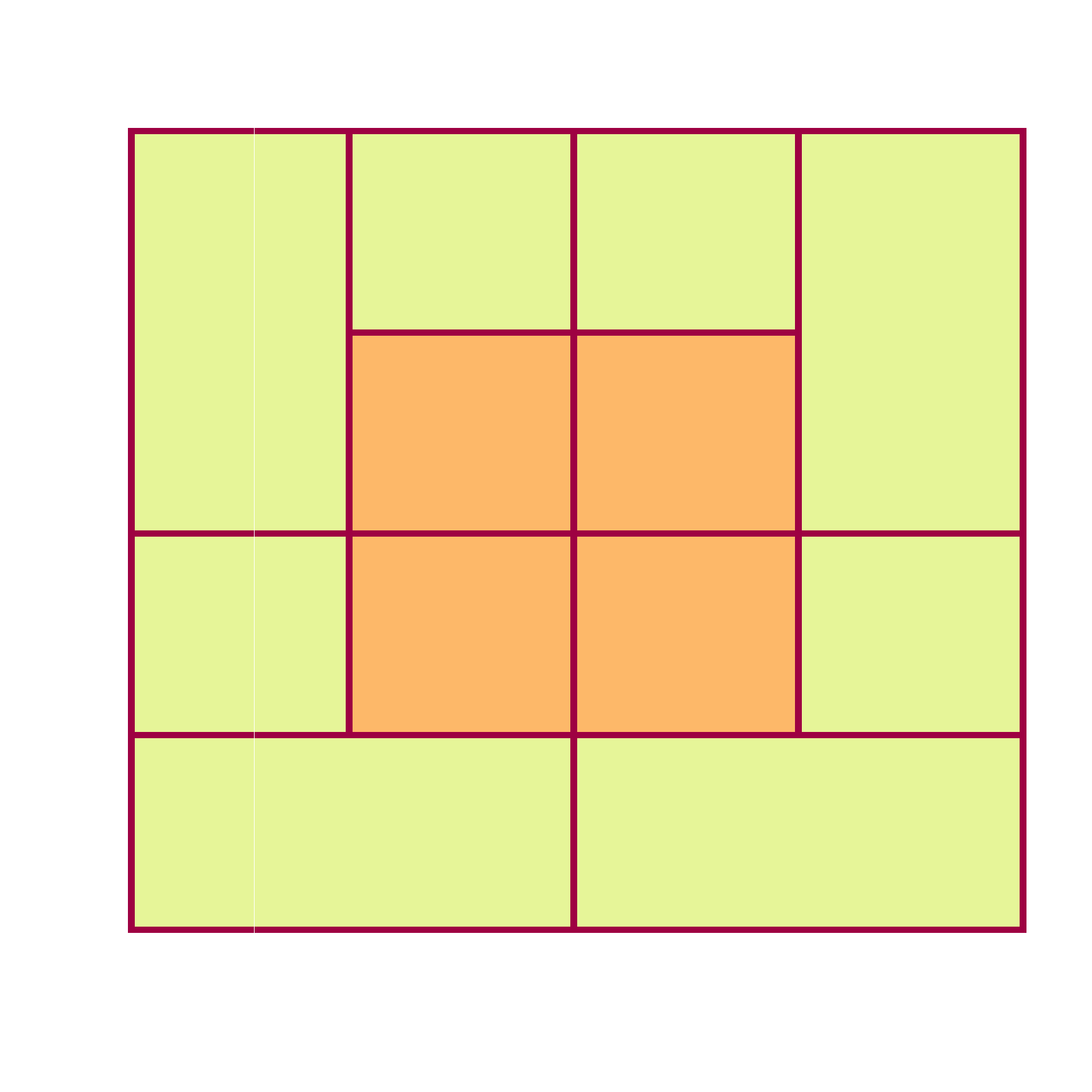}
		\caption{ 		\label{fig6}  The left panel shows the representation of a $\theta \in \mathbb{R}^{L_{2,n}}$ that takes on two values. The right panel shows a dyadic partition with a minimal number of elements where  $\theta$  is piecewise constant. In this example  $k_{rdp}(\theta)=  12$, the number of rectangles in the dyadic partition in the right panel. }
	\end{center}
\end{figure}

For a given array $\theta \in \R^{L_{d,n}}$, let \textit{$k_{rdp}(\theta)$ denote the smallest positive integer $k$ such that a set of $k$ rectangles $R_1,\dots,R_k$ form a recursive dyadic partition of $L_{d,n}$ and the restricted array $\theta_{R_i}$ is a constant array for all $1 \leq i \leq k.$} 
In other words, $k_{rdp}(\theta)$ is the cardinality of the minimal recursive dyadic partition of $L_{d,n}$ such that $\theta$ is piecewise constant on the partition. A visual representation of $k_{rdp}(\theta)$ is given in Figure \ref{fig6} for a signal $\theta \in \mathbb{R}^{_{2,n}}$.

To define our estimator, we will need a few more notations. If  $\Pi$ is any rectangular partition of $L_{d, n}$ we  let $S(\Pi)$ be the linear subspace of $\mathbb{R}^{L_{d,n}}$ consisting of vectors with constant values on each rectangle of $\Pi$. We also write $R \in \Pi$ to mean that the rectangle $R$ is one of the constituent rectangles of the partition $\Pi.$ We now define $O_{\Pi,\tau}(\cdot)$ be a function from $\mathbb{R}^{L_{d,n}}$ to  $S(\Pi)$ such that 
\[
(O_{\Pi,\tau}(y))_{i} = q_{\tau}(y_R)
\]
for $i \in R$, $R \in \Pi$, and where $q_{\tau}(y_R)$ is the empirical $\tau$-quantile of the set of values $y_R:= (y_i)_{i\in R}$.

Armed with the above notation we can reformulate the optimization problem in~\eqref{eqn:initestimator} by noting that $\hat{\theta}$ defined in~\eqref{eqn:initestimator} is the same as $O_{\tilde{\Pi},\tau}(y)$ where the partition $\tilde{\Pi}$ is an optimal solution to the following discrete optimization problem:
\begin{equation}
\label{eqn:problem_version2}
\underset{\Pi \in \mathcal{P}_{rdp}( L_{d,n})  }{\min}\,\,\left\{ \sum_{i\in L_{d,n}} \rho_{\tau}(y_i -  (O_{\Pi,\tau}(y))_i  ) + \lambda\vert \Pi\vert   \right\}.
\end{equation}

However, the estimator defined in~\eqref{eqn:initestimator} is not quite the estimator we propose and study in this paper as we need to modify the estimator slightly. To describe our QDCART estimator, which is the main object of study in this paper, we now define for any fixed quantile level $0 < \tau < 1$,
\begin{equation}\label{eqn:new_version}
 \hat{\theta}_{rdp}= O_{\hat{\Pi},\tau}(y)
\end{equation}
where 
\begin{equation}
\label{eqn:qdcart2}
\hat{\Pi}\,:= \argmin_{\Pi \in \mathcal{P}_{rdp}(L_{d,n})\,:\,   \vert R\vert \geq \gamma\,\,\forall R \in \Pi} \left\{ \sum_{i\in L_{d,n}} \rho_{\tau}(y_i -  (O_{\Pi,\tau}(y))_i  ) + \lambda\vert \Pi\vert   \right\}
\end{equation}
for tuning parameters $\lambda,\gamma > 0$.


Note that in view of~\eqref{eqn:problem_version2}, the above QDCART estimator is basically the same as the estimator in~\eqref{eqn:initestimator} with a slight modification. We restrict the optimization space to all partitions in  $\mathcal{P}_{rdp}(L_{d,n})$ with the constraint that the size of each of its constituent rectangles is larger than $\gamma>0$. This restriction is needed to avoid the estimator from being affected by large outliers.  We say more on this point in  Remark \ref{rem:outliers}. 

\section{Main Results}\label{sec:results}

We first state an assumption on the distribution of the coordinates of the data vector $y.$

\begin{assumption}\label{as1}
	There exist positive constants $L$, $\underline{f}$  and  $\overline{f}$  such that for any $\delta \in \mathbb{R}^{L_{d,n}}$  satisfying  $\|\delta\|_{\infty} \leq L$  we have that  for $i \in L_{d,n}$,
	\begin{equation}
	\label{eqn:lower}
 \overline{f}\,  \vert \delta_i \vert \geq 	\vert   F_{y_i}(\theta_i^* + \delta_i)  -F_{y_i}(\theta_i^*) \vert\,  \geq \,  \underline{f}\,  \vert \delta_i \vert.
	\end{equation}
\end{assumption}

Assumptions like the above are standard and commonly made in the quantile regression literature. For instance, without the upper bound part, Assumption \ref{as1} (Equation (\ref{eqn:lower})) also appeared in \cite{padilla2020risk} and is a  weaker version of Condition 2 from \cite{he1994convergence}, and is closely related to Condition D.1 in \cite{belloni2011}. 
The lower bound in Assumption~\ref{as1} is needed to ensure uniqueness of the $\tau$ quantiles of the marginal distributions of the coordinates of $y.$ We believe that Assumption~\ref{as1} is very mild. For example, sequences of distributions which are stochastically dominated by a distribution with continuous density (w.r.t. Lebesgue measure) which is bounded away from $0$ on any compact interval satisfy Assumption~\ref{as1}. If it is assumed that the errors are i.i.d. (which is a commonly made assumption) then if the error distribution itself has continous density (w.r.t Lebesgue measure) which is bounded away from $0$ on any compact interval then Assumption~\ref{as1} is satisfied. In particular, the error distributions could be i.i.d. Cauchy with no moments existing.



Before stating our main result we will need to make the following definition.
\begin{definition}\label{as2}
	Let  $\theta^{\prime}$  and $\theta^{\prime \prime}$  be arrays  of the true $\tau/2$-quantiles and  $(1-\tau)/2$-quantiles of $y$ respectively,  so that 
	\[
	\theta_i^{\prime}   =  \arg \min_{a \in  \mathbb{R}} \mathbb{E}\:\{\rho_{\tau/2}(y_i - a) \},
	\]
	and 
	\[
	\theta_i^{\prime \prime}  =  \arg \min_{a \in   \mathbb{R} } \mathbb{E}\:\{\rho_{(1-\tau)/2}(y_i - a) \}.
	\]
	Then  we denote
	\[
	U\,:=\,\max\{ \| \theta^{\prime}\|_{\infty},  \|\theta^{\prime \prime}\|_{\infty}    \}.
	\]
\end{definition}

Definition \ref{as2} simply quantifies the supremum norm   of  the $\tau/2$-quantiles and  $(1-\tau)/2$-quantiles  of  $y$.  We are now ready to state our main result for the QDCART estimator.


\begin{theorem}\label{thm0}
Suppose that   Assumption~\ref{as1} holds. There exists universal constants $c_1,C_1, C_2,C_3>0$ such that for any $0 < \epsilon < 1$, if we set $\gamma = c_1 \log N $ and
	\[
	\lambda \,=\,     C_1 \frac{\max\{ 1,U   \} \log(N)\log(N U)}{\epsilon}
	\]
	then with probability at least  $1-C_2\epsilon$,	
	\begin{equation}
	\label{eqn:upper0.3}
	\frac{\| \hat{\theta}_{rdp}-\theta^* \|^2}{N}\,\leq \, \frac{C_3  Q_{rdp}(\theta^*) }{\epsilon^2}
	\end{equation}
where
\[
Q_{rdp}(\theta^*)  := \underset{\theta  \in \mathbb{R}^N }{\inf}  \left\{\frac{  \| \theta   - \theta^* \|^2 }{N}\,+\, \frac{ k_{rdp}(\theta)( \max\{1,U^2\}\log^2 \left(\max\{N,U\}\right)  +  \|\theta\|^2_{\infty} \log N)}{N}  \right\}.
\]
\end{theorem}


Theorem \ref{thm0} provides the generalization of the oracle risk bound in Theorem 2.1 in \cite{chatterjee2019adaptive} to the quantile setting. We now list the differences of Theorem~\ref{thm0} with the oracle risk bound (Theorem 2.1 in \cite{chatterjee2019adaptive}) known for the mean regression counterpart.

\begin{enumerate}
\item
Theorem 2.1 in \cite{chatterjee2019adaptive}  requires that $Y -\theta^*$, the vector of errors, consists of i.i.d. mean zero Gaussian random variables. In contrast, Theorem \ref{thm0} holds under Assumption~\ref{as1} which does not require any tail decay assumptions for the distributions of the coordinates of the error vector (independence is still assumed). In particular, Assumption~\ref{as1} allows error distributions with no moments as well like the Cauchy distribution.


\item The result in \cite{chatterjee2019adaptive} is stronger in the sense that theirs is an upper bound in expectation,  given as
	\begin{equation}
		\label{eqn:oracle_dcart}
		\mathbb{E}\left( \frac{\|\hat{\theta} - \theta^* \|^2}{N} \right)\,\leq \,  \underset{\theta  \in \mathbb{R}^N }{\inf}\left\{\frac{(1-\delta)}{(1+\delta)}\frac{\|\theta-\theta^*\|^2}{N}  \,+\,   \frac{C \sigma^2 k_{rdp}(\theta) \log N}{\delta(1-\delta) N}\right\},
	\end{equation}
	for all $\delta\in (0,1)$, and for some constant  $C>0$. Our result in Theorem~\ref{thm0} gives a tail probability inequality which does not ensure that $\frac{\|\hat{\theta}_{rdp} - \theta^*\|^2}{N}$ has a finite first moment. It does ensure however that $$\frac{\|\hat{\theta}_{rdp} - \theta^*\|^2}{N} = \tilde{O}_{\mathbb{P}} (Q_{rdp}(\theta^{*}))$$ where $\mathbb{P}$ refers to an appropriately defined sequence of probability distributions corresponding to denoising problems of increasing size.
	

\item In effect, the the upper bound  in Theorem \ref{thm0} is only off by logarithmic factors compared to  the upper bound in Theorem 2.1 in  \cite{chatterjee2019adaptive}. Our bound in Theorem \ref{thm0} contains some extra terms which are benign. The factor $U$ should scale like $O(1)$ for any realistic error distribution sequence. The factor $\|\theta\|_{\infty}$ inside the infimum in the definition of $Q(\theta^*)$ essentially introduces another multiplicative factor of $\|\theta^*\|_{\infty} \leq U.$ 
\end{enumerate}

\begin{remark}
	\label{rem:outliers}
The choice of $\gamma$ in Theorem \ref{thm0} ensures that the QDCART estimator will  be well behaved in the sense of the $\ell_{\infty}$ norm, see Lemma \ref{cor1}. Such a restriction on the size of the rectangles in the optimal partition is actually needed. Otherwise, the QDCART estimator can be arbitrarily large in some locations under the presence of heavy tailed errors. If one considers standard subgaussian type assumptions on the errors, then this restriction on the size of the rectangles can be removed.
\end{remark}

We now turn to the issue of computation. In this article we also give an algorithm to compute the QDCART estimator based on bottom up dynamic programming. This algorithm is similar to the original algorithm given in~\cite{donoho1997cart} adapted to the quantile setting. We now state our computation result as a theorem.

\begin{theorem}\label{thm:compu}
		There exists an absolute constant $C > 0$ (not depending on $d,n$) such that the computational complexity, i.e. the number of elementary operations involved in the computation of the QDCART estimator in $d$ dimensions is bounded by $C \big(N (\log n)^d + d\:2^d N\big).$
\end{theorem}

The description of the algorithm and the proof of its computational complexity are given in Section~\ref{sec:compu}.



\subsection{Implications  for  Bounded Variation Signals}\label{sec:bddvar}

	

It was shown in~\cite{donoho1997cart} and~\cite{chatterjee2019adaptive} that an oracle risk bound of the type shown in Theorem~\ref{thm0} implies minimax rate optimality (up to log factors) for other function classes of interest as well. We now proceed to discuss consequences of  Theorem \ref{thm0} for the class $\mathcal{BV}_{d,n}(V)$ of bounded variation signals. This class of signals is  defined as 
\[
\mathcal{BV}_{d,n}(V)\,:=\, \left\{  \theta \in \mathbb{R}^{L_{d,n}}\,:\,\text{TV}(\theta) \leq V  \right\}, 
\]
where 
\[
\text{TV}(\theta) \,:=\,\sum_{(i,j)\in E_{d,n}}\vert \theta_i - \theta_{j}\vert,
\]
and $E_{d,n}$ is the edge set of the graph $L_{d,n}$.

The class of signals $\mathcal{BV}_{d,n}(V)$ is rich enough to contain signals that are  smooth in certain regions of their domain but discontinuous in other regions. The problem of estimation of a signal in the class $\mathcal{BV}_{d,n}(V)$ has attracted a lot of attention in the statistics literature, see for instance \cite{mammen1997locally,tibshirani2014adaptive,sadhanala2016total,hutter2016optimal,PadillaSST17,chatterjee2019new,ortelli2019prediction,guntuboyina2020adaptive}.

We arrive at the next corollary by combining Theorem \ref{thm0} with existing approximation theoretic results shown in~\cite{chatterjee2019adaptive} (see Proposition $8.9$ and Theorem $4.2$ there)

\begin{corollary}
	\label{cor2}
For any $\theta^* \in \mathcal{BV}_{d,n}(V)$, there exists a constant $C>0$ only depending on the dimension $d$ such that 
\begin{equation*}
\displaystyle Q_{rdp}(\theta^*) \leq 
\begin{cases}
C \left(   \frac{V^{2/3}  \max\{1,U^2\} \log^{5/3} (\max\{N,U\})}{N^{2/3}}+\frac{   \max\{1,U^2\}  \log^2 \max\{N,U\}   }{N}    \right) \:\:\:{\text{if $d = 1$}}\\
C \left(  \frac{V \max\{1,U^2\}\log^2\left(\max\{N,U\}\right)  }{N}+\frac{\max\{1,U^2\} \log^2 \max\{N,U\} }{N}\right) \:\:\:{\text{if $d > 1$}}.
\end{cases}
\end{equation*}
Therefore, under the same assumptions and the choice of $\lambda$ and $\gamma$ in Theorem~\ref{thm0}, the same probability tail bound as in~\eqref{eqn:upper0.3} holds for any $\theta^* \in \mathcal{BV}_{d,n}(V)$ with $Q_{rdp}(\theta^*)$ replaced by the bound above. 
\end{corollary}

The rates implied by Corollary \ref{cor2}  are minimax optimal, save for logarithmic  factors, in the class $\mathcal{BV}_{d,n}(V)$,  see the discussion in \cite{tibshirani2014adaptive} for the case $d=1$ and the corresponding one in~\cite{hutter2016optimal},~\cite{sadhanala2016total} for the case $d>1$. It was shown in Theorem 5.1 from ~\cite{chatterjee2019adaptive} that the mean regression version of Dyadic CART is minimax rate optimal (up to log factors) in the class $\mathcal{BV}_{d,n}(V)$. Corollary~\ref{cor2} can be seen as an extension of this result to the quantile setting which holds under much weaker tail decay conditions.

\subsection{Implications for Piecewise Constant  signals}
\label{sec:pc}


We now discuss consequences of Theorem \ref{thm0} for the class of piecewise constant signals in dimensions $d = 1$ and $d = 2.$ Towards that end, given $\theta \in \mathbb{R}^{L_{d,n}}$, we define  \textit{$k(\theta)$ as the size of the smallest rectangular partition $\Pi$ of $L_{d,n}$ such that $\theta$ is constant in each rectangle of $\Pi$}. By construction, $k(\theta) \leq  k_{rdp}(\theta)$ for all $\theta \in \mathbb{R}^{L_{d,n}}$. Furthermore, Proposition 3.9 in \cite{chatterjee2019adaptive} shows that there exists an absolute constant $C>0$ such that  for all $\theta \in \mathbb{R}^{L_{d,n}}$ it holds that 
\begin{equation}
	\label{eqn:k1}
	 k_{rdp}(\theta)   \,\leq \, C  k(\theta) \log \left(\frac{en}{k(\theta) }\right)
\end{equation}
if  $d=1$ and
\begin{equation}
	\label{eqn:k2}
	 k_{rdp}(\theta)   \,\leq \, C (\log n)^2 k(\theta) 
\end{equation}
if $d=2$.

Combining  Theorem \ref{thm0} with  (\ref{eqn:k1}) and (\ref{eqn:k2}) we immediately obtain our next corollary

\begin{corollary}
	\label{cor3}
For any $\theta^* \in \mathbb{R}^{L_{d,n}}$, there exists a constant $C>0$ only depending on the dimension $d$ such that 
\begin{equation*}
Q_{rdp}(\theta^*) \leq 
\begin{cases}
C\left(    \frac{k(\theta^*)\max\{1,U^2\}\log^2 (\max\{N,U\})  \log(N/ k(\theta^*) )   }{N}    \right) \:\:\:{\text{if $d = 1$}}\\
C \left(    \frac{k(\theta^*)\max\{1,U^2\}\log^2 (\max\{N,U\})\log^2  N}{N}    \right) \:\:\:{\text{if $d = 2$}}.
\end{cases}
\end{equation*}
Therefore, under the same assumptions and the choice of $\lambda$ and $\gamma$ in Theorem~\ref{thm0}, the same probability tail bound as in~\eqref{eqn:upper0.3} holds for any $\theta^* \in \mathbb{R}^{L_{d,n}}$ with $Q_{rdp}(\theta^*)$ replaced by the bound above. 
\end{corollary}

Notice  that  in Corollary~\ref{cor3}, the resulting rate implied is  $\tilde{O}(k(\theta^*)/N)$ which is the usual parametric rate of estimation for a signal $\theta^*$ consisting of $k(\theta^*)$ pieces if one knows the locations of the end points of the constant pieces of $\theta^*.$ Here, of course the QDCART estimator does not know the true partition corresponding to the true signal. Corollary~\ref{cor3} can be seen as an extension of  Corollary 3.10 in~\cite{chatterjee2019adaptive} to the quantile setting which holds even under heavy tailed error distributions.

\begin{remark}
	The situation when $d > 2$ is more difficult as versions of~\eqref{eqn:k1} and~\eqref{eqn:k2} are not known to hold in higher dimensions than $2.$ We refer the reader to~\cite{chatterjee2019adaptive} where this issue has been thoroughly discussed. We prefer therefore to just state our results for dimensions $d \leq 2$.
\end{remark}

\subsection{Quantile ORT Estimator}\label{sec:ort}
The ORT estimator, introduced in~\cite{chatterjee2019adaptive} is a variant of the Dyadic CART estimator which enjoys better statistical risk guarantees in general but has significantly slower computational complexity; see Lemma $1$ in~\cite{chatterjee2019adaptive}. Just as we have proposed QDCART, it is natural to extend the optimal regression tree (ORT) estimator to the quantile setting as well. This leads us to define the quantile optimal regression tree (QORT)  estimator. Before giving the definition of QORT, we need to introduce some additional notation.

 Given a rectangle  $R = \prod_{j=1}^d [a_j,b_j] \subset L_{d,n}$, a hierarchical split consists  of choosing a coordinate $j \in \{1,\ldots,d\}$ and then constructing the rectangles $R_1$ and $R_2$ with $R = R_1 \cup R_2$, $R_1\cap R_2 =\emptyset$ and
\[
R_1 \,=\, \prod_{i=1}^{j-1} [a_i,b_i] \times [a_j,l] \times \prod_{i=j+1}^d [a_i,b_i],
\]
with $a_j\leq l \leq b_j$ and $l \in \mathbb{Z}_{+}$. Thus, the difference of a hierarchical split with a dyadic split is that the former is not restricted to split an interval only at the midpoint. Starting from $L_{d,n}$, one can keep on performing hierarchical splits recursively, creating refined partitions. A \textit{hierarchical partition/decision tree} is any partition reachable by such successive hierarchical splits. Note that this is the usual definition of a decision tree except we are carrying out everything on the lattice $L_{d,n}.$  We denote by $\mathcal{P}_{tree}(L_{d,n})$ the set of hierarchical partitions of $L_{d,n}$.

Given $\theta \in \mathbb{R}^{L_{d,n}}$, we denote by \textit{$k_{tree}(\theta)$ the smallest number of elements of any hierarchical partition in which $\theta$ is piecewise constant}. It is clear that for any $\theta \in \mathbb{R}^{L_{d,n}}$ we must have $$k(\theta) \leq k_{tree}(\theta) \leq k_{rdp}(\theta).$$

Armed with the notation above, we can now define the estimator 
\begin{equation}\label{eqn:new_version2}
	\hat{\theta}_{tree}= O_{\hat{\Pi}_{tree},\tau}(y)
\end{equation}
where 
\begin{equation}
	\label{eqn:qdcart3}
	\hat{\Pi}_{tree}\,:= \argmin_{\Pi \in \mathcal{P}_{tree}(L_{d,n})\,:\,   \vert R\vert \geq \gamma\,\,\forall R \in \Pi} \left\{ \sum_{i\in L_{d,n}} \rho_{\tau}(y_i -  (O_{\Pi,\tau}(y))_i  ) + \lambda\vert \Pi\vert   \right\}
\end{equation}
for tuning parameters $\lambda,\gamma > 0$. By construction, $\hat{\theta}_{tree}$ is the quantile version of the ORT estimator proposed and studied in \cite{chatterjee2019adaptive}.




With the notation above in hand, we are now ready to present our main result for the QORT estimator.

\begin{theorem}\label{thm5}
	Define for any $\theta \in \mathbb{R}^{L_{d,n}}$, the quantity
	\begin{equation*}
	Q_{tree}(\theta^*)  := \underset{\theta  \in \mathbb{R}^N }{\inf}  \left\{\frac{  \| \theta   - \theta^* \|^2 }{N}\,+\, \frac{ k_{tree}(\theta)( \max\{1,U^2\}\log^2 \left(\max\{N,U\}\right)  +  \|\theta\|^2_{\infty} \log N)}{N}  \right\}
	\end{equation*}
	Under the same assumptions and the choice of $\lambda$ and $\gamma$ in Theorem~\ref{thm0}, the same probability tail bound as in~\eqref{eqn:upper0.3} holds for any $\theta^* \in \mathbb{R}^{L_{d,n}}$ with one difference;  $Q_{rdp}(\theta^*)$ is replaced by $Q_{tree}(\theta^*).$

\end{theorem}

\begin{remark}
	The above theorem basically says that $\frac{1}{N} \|\hat{\theta}_{tree} - \theta^*\|^2 = O_{\mathbb{P}} (Q_{tree}(\theta^*)).$ This is in general a better bound than saying $\frac{1}{N} \|\hat{\theta}_{tree} - \theta^*\|^2 = O_{\mathbb{P}} (Q_{rdp}(\theta^*))$ because $k_{tree}(\theta^*) \leq k_{rdp}(\theta^*).$  
\end{remark}


Theorem \ref{thm5} generalizes to the quantile setting  the general risk bound proven in \cite{chatterjee2019adaptive} for the ORT estimator. It is clear that the implications presented for bounded variation and piecewise constant function classes continue to hold for the QORT estimator as well. However, the QORT estimator would have significantly worse computational complexity than the QDCART estimator which is why we focus more on the QDCART estimator in this paper. It should be possible to provide an algorithm demonstrating a computational complexity result scaling like $\tilde{O}(N^{2 + 1/d})$ for the QORT estimator, analogous to Theorem~\ref{thm:compu}. We do not carry this due to space considerations.


\begin{remark}
	All our theoretical guarantees are in the regime when $d$ is held fixed and $n$ is growing. Our estimator and our results are practically useful when $d$ is small, like $1,2$ or $3$. 
\end{remark}
\section{Discussion}\label{sec:discuss}
\subsection{Proof Technique}\label{sec:proof}

Our proof follows a M-estimation approach by viewing the QDCART estimator as a penalized M estimator. This M estimation viewpoint was also used to analyze the Quantile version of Trend Filtering as in~\cite{padilla2020risk} and we borrow some of the techniques from~\cite{padilla2020risk}. The mean regression version of Dyadic CART was throroughly analyzed in~\cite{chatterjee2019adaptive}. We also use some proof techniques developed there and adapt it to our setting. We now discuss a sketch of the proof of our main result in Theorem \ref{thm0}. We have divided the proof sketch into several steps for the convenience of the reader.

From the M estimation viewpoint, the natural loss function which arises is the following population quantile loss function
$M \, :\,\mathbb{R}^{L_{d,n}}\rightarrow \mathbb{R} $
\[
M(\theta )\,:=\,  \sum_{i\in L_{d,n}} \mathbb{E}\left(  \rho_{\tau}(y_i-\theta_i) -  \rho_{\tau}(y_i-\theta_i^*)  \right).
\]

Another loss function that plays a role in our proof is the following Huber loss type function 
\[
\Delta^2(\theta)  =    \sum_{i \in  L_{d,n}}   \min\{  \vert \theta_i \vert,   \vert \theta_i\vert^2  \}.
\]

The Huber loss function $\Delta^2$ is always upper bounded by the population quantile loss; this is the content of Lemma 13 in \cite{padilla2020risk} which says that, under Assumption \ref{as1}, there exists an absolute constant $c_0>0$ such that for all $\delta \in \mathbb{R}^{L_{d,n}}$ it holds that
\begin{equation}
\label{eqn:loss}
M(\theta^* + \delta) \,\geq \,  c_0   \Delta^2(\theta).
\end{equation}

With the notation above in hand, we now proceed to sketch the different steps involved in the proof of  Theorem \ref{thm0}.

\textbf{Step 1: (Preliminary Localization)} We first show that $\|\hat{\theta}_{rdp}\|_{\infty} \leq U$ with high probability. The reader can think of this as a preliminary localization step. 

Recall  $U$ from   Definition \ref{as2}. Within this proof sketch, we will assume $U = O(1)$ which is the regime of interest. For sequences $a_n$ and $b_n$, we will also use the notation $a_n \lesssim b_n$ to denote that $a_n \leq C b_n$ for an absolute constant $C$.

This preliminary localization step is crucial in our proof because it allows us to conclude that 
\begin{equation}\label{eq:equi}
\|\hat{\theta}_{rdp} - \theta^*\|^2 \lesssim \Delta^2(\hat{\theta}_{rdp}  - \theta^*) \lesssim M(\hat{\theta}_{rdp}) \lesssim \|\hat{\theta}_{rdp}  - \theta^*\|^2.
\end{equation}
The last inequality above follows from Lemma~\ref{lem5}. The above means that the loss functions $M, \Delta^2$ are essentially equivalent (up to constants) to the squared loss. In the reminder of the proof sketch all the events are intersected with  $\|\hat{\theta}_{rdp}\|_{\infty} \leq U$ as the complement event $\|\hat{\theta}_{rdp}\|_{\infty} > U$ has negligible probability and can be handled separately.

\textbf{Step 2: (Reduction to Bounding $M$ loss)} Because of~\eqref{eq:equi}, in order to bound $\mathbb{P}(\|\hat{\theta}_{rdp} - \theta^*\|^2 > t^2)$ for any $t > 0$, it suffices for us to bound 
$$\mathbb{P}(M(\hat{\theta}_{rdp} ) > t^2).$$

\textbf{Step 3: (Peeling)} To bound the required probability, we perform the peeling step which is a standard step in empirical process theory. We write
$$\mathbb{P}(M(\hat{\theta}_{rdp}) > t^2) = \sum_{j = 1}^{J} p_j,\,\,\,\,\,p_j \,:=\, \mathbb{P}(2^{j - 1} t^2 < M(\hat{\theta}_{rdp}) \leq 2^j t^2).$$
Again, because $\|\hat{\theta}_{rdp}\|_{\infty} \leq U$ is true it follows that $M(\hat{\theta}) \lesssim \Delta^2(\hat{\theta}_{rdp}-\theta^*) \lesssim N$ as explained in (\ref{eqn:inf}). Therefore, we only need to sum up to $J = O(\log N) = \tilde{O}(1)$ in our peeling step.

\textbf{Step 4:(Basic Inequality, Suprema and Markov's Inequality)} 
For $i \in L_{d,n}$, define the sample version of the quantile population loss function as a random function $\hat{M}: \mathbb{R} \rightarrow \mathbb{R}$ such that 
\[
\hat{M}(\theta)   :=    \sum_{i \in  L_{d,n}}  	\left\{\rho_{\tau}(y_i -   \theta_i) -   \rho_{\tau}(y_i -   \theta_i^* ) \right\} .
\]
Now the so called basic inequality gives us
\begin{equation}
\label{eqn:bi}
\hat{M}(\tilde{\theta}) - \hat{M}(\hat{\theta}_{rdp}) +   \lambda k_{rdp}(\tilde{\theta})  -  \lambda k_{rdp}(\hat{\theta}_{rdp})   \geq 0
\end{equation}
for any $\tilde{\theta} \in \Theta$,  with $\Theta$  the parameter space defined in (\ref{eqn:Theta2}). 

Now we can bound $p_j$ as follows . For any reference $\tilde{\theta} \in \Theta$, notice that
\begin{align*}
&p_j = \mathbb{P}(2^{j - 1} t^2 < M(\hat{\theta}_{rdp}), M(\hat{\theta}_{rdp}) \leq 2^j t^2,) \leq \\& \mathbb{P}(2^{j - 1} t^2 < M(\hat{\theta}_{rdp}) + \underbrace{\hat{M}(\tilde{\theta}) - \hat{M}(\hat{\theta}_{rdp}) + \lambda   k_{rdp}(\tilde{\theta})  -  \lambda k_{rdp}(\hat{\theta}_{rdp}) }\underbrace{- M(\tilde{\theta}) + M(\tilde{\theta})},M(\hat{\theta}_{rdp}) \leq 2^j t^2) 
\end{align*}
Above, we used the basic inequality and added and subtracted $M(\tilde{\theta})$ because we would like to obtain an oracle risk bound with respect to $\tilde{\theta}.$ Now we will just  ``\textit{sup out}" $\hat{\theta}_{rdp}$ to obtain 
\begin{align*}
&p_j \leq \mathbb{P}\bigg(2^{j - 1} t^2 < \sup_{\theta \in \Theta: \|\theta - \theta^*\|^2 \lesssim 2^j t^2} \big(M(\theta) + \hat{M}(\tilde{\theta}) - \hat{M}(\theta) +  \lambda  k_{rdp}(\tilde{\theta})  - \lambda k_{rdp}(\theta) - M(\tilde{\theta}) + M(\tilde{\theta})\big)\bigg)
\end{align*}
where again we used the equivalence of the squared loss and the $M$ loss. 
Next, we simply use Markov's inequality to obtain $ p_j\,\leq \,      T_{1,j} + T_2  $, where 
\[
T_{1,j}   :=  \frac{1}{2^{j - 1} t^2 }\mathbb{E}\left( \sup_{\theta \in \Theta: \|\theta - \theta^*\|^2 \lesssim 2^j t^2} \big(M(\theta) - \hat{M}(\theta) - M(\tilde{\theta}) + \hat{M}(\tilde{\theta}) +   \lambda k_{rdp}(\tilde{\theta})  - \lambda k_{rdp}(\theta)\big) \right),
\]
and \[
T_{2,j}  =  \frac{M(\tilde{\theta})}{2^{j - 1} t^2}.
\]

\textbf{Step 5:(Symmetrization and Contraction)}
At this point, $T_{1,j}$ can be viewed as the expected suprema of a penalized empirical process. To simplify matters further, we use the tools of symmetrization (Lemma \ref{lem28}) and contraction  (Lemma \ref{lem29}) to convert a penalized empirical process to a penalized Rademacher process. Specifically, we show that 
\begin{equation}
\label{eqn:local}
T_{1,j} \,\lesssim\,  T_{1,j}^{\prime}  := \frac{1}{2^{j - 1} t^2 }\mathbb{E}\left(    \underset{ \theta \in \Theta \,:\,\| \theta-\theta^* \|^2 \lesssim 2^{j}t^2 }{\sup}\left\{  \xi^{\top}(\theta  - \tilde{\theta}) +\frac{  \lambda}{2}(   k_{rdp}(\tilde{\theta}) -  k_{rdp}(\theta) )    \right\}     \right)
\end{equation}
where  $\xi$ is a vector of  independent Radamacher random variables, see Lemma \ref{lem3}.

\textbf{Step 6:(Bounding Penalized Rademacher Complexity)}
At this point, we are left with the task of bounding the penalized Rademacher complexity term as in (\ref{eqn:local}). This task has essentially been done in Proposition $8.9$ in~\cite{chatterjee2019adaptive}, the only difference being that they bounded the corresponding Gaussian complexity term. It is not hard to convert the ideas there to our setting where we have Rademacher variables. In lemma~\ref{lem3} we show that
$$T_{1,j}^{\prime}    =    \tilde{O}\left(   \frac{k_{rdp}(\tilde{\theta}) + \|\tilde{\theta} - \theta^*\|^2}{2^{j-1}  t^2 }  \right),$$
provided that $\lambda$ is chosen to be not less than $O(\log^2 N)$.

\textbf{Step 7: (Putting everything together)}
Next, we bound $T_{2,j}$ by the squared error loss using Lemma~\ref{lem5}, which shows that
\[
M(\tilde{\theta})    \lesssim \|\tilde{\theta} - \theta^*\|^2.
\]
It follows from this and the previous steps that
\[
\mathbb{P}(  \| \hat{\theta}_{rdp} -\theta^* \|^2 > t^2)   \,\leq \,  \sum_{j=1}^{J} p_j  \,\lesssim\,  \sum_{j=1}^{J}   \frac{1}{2^{j-1}}\left[  \frac{k_{rdp}(\tilde{\theta})}{ t^2 }  +        \frac{1}{t^2}\|\tilde{\theta} - \theta^*\|^2 \right], 
\]
and so we can take an infimum over $\tilde{\theta} \in \Theta$ on the right hand side above.
A simple approximation lemma (see Lemma~\ref{lem6}) is now used to justify that we can actually take an infimum over all $\tilde{\theta} \in \R^{L_{d,n}}$ in the previous display which leads to 
\[
\mathbb{P}(  \| \hat{\theta}_{rdp} -\theta^* \|^2 > t^2)   \,\lesssim\, \frac{Q(\theta^*)}{t^2}.  
\]
where $$Q_{rdp}(\theta^*) \lesssim \inf_{\theta \in \R^{L_{d,n}}} \big(k_{rdp}(\tilde{\theta}) + \|\tilde{\theta} - \theta^*\|^2\big).$$
The above display implies that $\| \hat{\theta}_{rdp} -\theta^* \|^2 = O_{\mathbb{P}}(Q_{rdp}(\theta^*)).$ This concludes the proof.

\subsection{Comparison with Quantile Total Variation Denoising}\label{sec:comp}

We have shown that QDCART is computable in near linear time and enjoys attractive statistical properties. We do not compare QDCART with corresponding mean regression estimators simply because under heavy tails, the mean regression estimators perform poorly while our results continue to hold; see the simulations section in~\cite{padilla2020risk}. Therefore, comparison is appropriate with other quantile regression estimators. We believe that the most natural competitor to QDCART is the Quantile Total Variation Denoising (QTVD) estimator, proposed and studied in~\cite{padilla2020risk}. Actually, there are two versions of QTVD, the so called constrained and the penalized version. In the univariate case,~\cite{padilla2020risk} refers to the constrained version of this estimator as  constrained quantile  fused lasso (QCFL), and to the penalized version as penalized quantile fused lasso (PQFL). Both of these estimators were analyzed in \cite{padilla2020risk}.


In comparing  Corollaries \ref{cor2} and \ref{cor3} with the existing results known for CQFL and PQFL estimators, we make the following points. 

\begin{enumerate}
	\item The results proven in \cite{padilla2020risk} for CQFL and PQFL only need the lower bound portion in Assumption \ref{as1} while Corollaries \ref{cor2} and \ref{cor3}  require also the upper bound in Assumption \ref{as1} to hold.  However, this  is a very mild condition that leads to a guarantee for QDCART in terms of the mean squared error, a stronger (and a much cleaner) result than those for CQFL and PQFL which are based on the loss $\Delta_N^{2}(\cdot)$ defined as 
	\begin{equation}
		\label{eqn:delta}
		\Delta_N^2(v)\,:=\,\frac{1}{N} \sum_{i \in L_{d,n}} \min\{\vert v_i\vert, v_i^2 \},\,\,\forall v \in \mathbb{R}^{L_{d,n}}. 
	\end{equation}
	Results under the squared error loss are not yet known for the CQFL and PQFL estimators.  The main reason for this is that the  localization result described in Step 1 in Section \ref{sec:proof} is not available for CQFL and PQFL.
	
	\item The dependence on the total variation of the true signal $V$ in both Corollaries \ref{cor2} and \ref{cor3} are optimal in the sense that they match the right dependence known in the mean regression case; see~\cite{sadhanala2016total} for lower bounds on bounded variation signal classes. The results for the CQFL and PQFL estimators in~\cite{padilla2020risk} seem to have sub optimal dependence on $V.$ Thus, to the best of our knowledge, our results on QDCART are the first quantile regression based estimators which enjoy minimax rate optimality with respect to both the sample size $N$ and the total variation of the unknown signal $V.$
	

\item Our results for  QDCART depend on the quantity $U$  something that is not the case for  QTVD. However, in any realistic setting, we would have $U = O(1)$.

  \item Corollary \ref{cor3} gives a near parametric rate of convergence for piecewise constant signals. In the univariate case, the corresponding result is known for the CQFL and PQFL in Theorems 2 and  4 from \cite{padilla2020risk}. However, these results need the true signal to satsify a \textit{minimal spacing condition} which is not needed for QDCART. This is potentially a significant advantage of QDCART over QTVD, even in $d = 1$, as far as attaining adaptively optimal rates of convergence for piecewise constant signals is concerned.
  

 \item In the mean regression problem, it is known that when $d = 2$, the TVD estimator cannot attain near parametric rates of convergence for a rectangular piecewise constant signal; see Theorem $2.3$ in~\cite{chatterjee2019new}. Therefore, it is expected that the QTVD estimator would also not be the best tool for estimating rectangular piecewise constant signals. On the other hand, the QDCART estimator does attain the $\tilde{O}(\frac{k(\theta^*)}{N})$ rate in $2$ dimensions and seems to be the right tool for 
estimating rectangular piecewise constant signals as well.


\item The QDCART estimator is computable in $\tilde{O}(N)$ time in any dimensions. In contrast, it is unknown and unlikely that the QTVD estimator is computable in $\tilde{O}(N)$ time in dimension larger than $1.$

\end{enumerate}

\begin{remark}
	The last three points above show that the QDCART estimator is a computationally faster alternative to the QTVD estimator while also enjoying some statistical advantages. We perform numerical experiments to further compare finite sample performance of QDCART and QTVD estimators in Section~\ref{sec:experiments}.
\end{remark}

\subsection{Background and Related Literature}
\label{sec:other}

Our work in this paper falls under the scope of nonparametric quantile regression. We now briefly review some classical work on nonparametric quantile regression.  In the context of median regression some early works include \cite{utreras1981computing}, \cite{cox1983asymptotics}, and \cite{eubank1988spline}. \cite{koenker1994quantile} proposed  one dimensional quantile smoothing splines. These estimators were studied in \cite{he1994convergence} under the assumption that the quantile function is H\"{o}lder continuous.

Other related quantile nonparametric estimators include the
quantile random forest proposed by \cite{meinshausen2006quantile}. \cite{brown2008robust} developed a wavelet-based estimator for median regression. Recently, \cite{ye2021non} developed the  k-nearest neighbour quantile fused lasso approach and~\cite{padilla2020quantile} studied quantile regression with ReLU networks.

\subsection{Future Directions}
There are different research directions that we leave for future work. A natural extension is to consider piecewise polynomial structures in the estimator, similarly to 
\cite{chatterjee2019adaptive}. However, we are currently unaware of how to extend our theory to such a setting. The main bottleneck is that given a fixed sub rectangle of $L_{d,n}$ we do not know how to obtain an $\ell_{\infty} $upper bound when fitting quantile regression constrained to the class of polynomials of degree $r>0$. When $r=0$ the latter can be done as in  Corollary \ref{cor1}. This a crucial ingredient in our proof that we do not know how to handle when dealing with higher order piecewise polynomial signals.

Moreover, it would be worthwhile to mention here that all our theoretical results hold under a theoretical choice of the tuning parameters. In our experiments, following \cite{yu2001bayesian}, we choose  the tuning parameters by Bayesian information criterion   (BIC) for quantile regression. It would be interesting to provide theoretical guarantees for an estimator which chooses the tuning parameters in a data driven way, for example, by some form of cross validation.

\section{Computation of the QDCART Estimator}\label{sec:compu}

The goal of this section is to develop a computationally efficient algorithm for the QDCART estimator defined in(~\ref{eqn:new_version}). In doing so,  our construction will imply the conclusion of  Theorem  \ref{thm:compu}. This algorithm is an adaptation of the original algorithm given in~\cite{donoho1997cart} (also see Lemma $1$ in~\cite{chatterjee2019adaptive}) to the quantile setting.

We have to solve the discrete optimization problem in~\eqref{eqn:qdcart2}. Let us first see how we can solve the discrete optimization problem in~\eqref{eqn:problem_version2} where there are no constraints on the size of the rectangles.

In order to study the optimization problem~(\ref{eqn:problem_version2}), we first define a corresponding subproblem for any rectangle  $R.$ For any rectangle $R \subset L_{d,n}$ and a partition $\Pi \in \mathcal{P}_{rdp}(L_{d,n})$, we let  $\Pi(R) \,:=\, \{  A \cap R \,:\,  A \cap R  \neq \emptyset,\,\,\,A\in \Pi    \}$  be the  partition induced by $\Pi$ in  $R$. We then let $\mathcal{P}_{rdp}(R) \,:=\, \{  \Pi(R) \,:\, \Pi \in  \mathcal{P}_{rdp}(L_{d,n})\}$. In words, $\mathcal{P}_{rdp}(R)$ is the set of recursive dyadic partitions of the rectangle $R.$ 

Then we define the following subproblem and define its optimal value as
\[
\mathrm{OPT}(R)\,:=\, \underset{\Pi \in \mathcal{P}_{rdp}(R)  }{\min}\,\,\left\{ \sum_{i\in R} \rho_{\tau}(y_i -  (O_{\Pi,\tau}(y))_i  ) + \lambda\vert \Pi\vert   \right\}.
\]
Clearly,  $\mathrm{OPT}(L_{d,n})$ is the optimal value of the objective function associated with QDCART. The basic idea is to be able to solve \textit{smaller} subproblems as above and build these smaller solutions to solve the full optimization problem. The following \textit{dynamic programming} relation allows us to build up from bottom up.

\[
\mathrm{OPT}(R)\,:=\, \underset{R_1, R_2 \,\,\text{dyadic split of} R\, }{\min}\,\,\left\{ \mathrm{OPT}(R_1) + \mathrm{OPT}(R_2),  \sum_{i\in R} \rho_{\tau}(y_i -  q_{\tau}(y_R)  ) +\lambda  \right\},
\]
where by saying that ``$R_1 $ and $R_2$ dyadic  split of $R$" we mean that $R_1$ and $R_2$ were obtained after performing a dyadic split of $R.$ The above relation follows because of the separable nature of our optimization objective and the second term inside the minimum above corresponds to not splitting $R$ at all.

We now proceed visiting dyadic rectangles bottom-up according to the length of $R$. The length of $R$ is defined as the sum of the lengths of the sides of the rectangles. We will start from the minimum possible length $d$ all the way to $nd$. Our goal is to store $\mathrm{OPT}(R)$ and $\mathrm{SPLIT}(R)$ for each dyadic rectangle $R$, where $\mathrm{SPLIT}(R)$ indicates the optimal split for rectangle $R$. Note that the total number of possible splits is $d$ (one for each dimension) and thus $\mathrm{SPLIT}(R)$ can be represented by a single integer within the set $[d].$

For each dyadic rectangle $R$, let us denote 
\begin{equation}\label{eqn:quantities}
	\mathrm{SQL}(R):= \sum_{i\in R} \rho_{\tau}(y_i -  q_{\tau}(y_R)).
\end{equation}
where $\mathrm{SQL}$ stands for sum of quantile loss and $q_{\tau}(y_R)$ is an empirical $\tau$ quantile of the set of observations in $y_{R}.$ Assume that we have succesfully computed $\mathrm{SQL}(R)$ for each dyadic rectangle. Then, at each dyadic rectangle $R$, to compute $\mathrm{OPT}(R)$ we have to compute the sum $\mathrm{OPT}(R_1) + \mathrm{OPT}(R_2)$ for each possible non trivial dyadic split of $R$ into $R_1,R_2$ and then compute the minimum of $d + 1$ numbers. Once $\mathrm{OPT}(R)$ is computed, $\mathrm{SPLIT}(R)$ is also automatically calculated when we are computing the minimum of these $d + 1$ numbers.

Note that since we are visiting dyadic rectangles bottom up, we have already computed $\mathrm{OPT}(R')$ for all sub rectangles $R' \subset_{\neq} R.$ Therefore, the computation required for computing $\mathrm{OPT}(R)$ is $d + 1.$ The total number of dyadic rectangles is at most $2^{d}N.$ Therefore, the total computation required to compute $\mathrm{OPT}(R)$ for all dyadic rectangles $R$ is at most $(d + 1) 2^d N.$

Now we proceed to explain how to compute $\mathrm{SQL}(R)$ for each dyadic rectangle $R.$ We again do this by a bottom up scheme by visiting dyadic rectangles according to their length (small to large). Our aim here is to compute a sorted list of observations within each dyadic rectangle $R$. We do this bottom up. For each dyadic rectangle $R$ we consider a  particular dyadic split $R_1$ and $R_2$ which is obtained by dyadically splitting (in dictionary order) the first coordinate of $R.$ In the case the first coordinate is a singleton, we use the second coordinate to split and so on. Now on our bottom up visits, we iteratively compute sorted lists for these dyadic rectangles. For instance, for a given dyadic rectangle $R$, we take its corresponding dyadic split (in dictionary order) into $R_1,R_2.$ Now we have already created a sorted list for $R_1$ and $R_2.$ To create a sorted list for $R$ we just need to merge two sorted lists. We can do this by the standard merge sort algorithm. The computation required at this step is $O(|R_1|) + O(|R_2|) = O(|R|).$ Once we are able to construct the sorted list of observations within $R$, now $\mathrm{SQL}(R)$ can be readily computed in $O(|R|)$ time. 

Now consider a dyadic rectangle $R$ of a given size $2^{i_1} \times \ldots \times 2^{i_d}$. The total number of dyadic rectangles of this size $2^{i_1} \times \ldots \times 2^{i_d}$ is at most 
\[
\frac{n}{2^{i_1}} \times \ldots   \frac{n}{2^{i_d}}.
\]
Therefore, the total computational work needed to compute $\mathrm{SQL}(R)$ for all dyadic rectangles $R$ with this given size is simply $$O(2^{i_1} \times \ldots \times 2^{i_d} \times \frac{n}{2^{i_1}} \times \ldots   \frac{n}{2^{i_d}}) = O(n^d).$$

Now note that the total number of distinct sizes $2^{i_1} \times \ldots \times 2^{i_d}$ is at most $(\log n)^d.$ Therefore, the total computational work needed to compute $\mathrm{SQL}(R)$ for all dyadic rectangles of all sizes with our bottom up scheme is $O(N (\log n)^d).$

Finally, we see that the total computation required to compute $\mathrm{OPT}(R)$ and $\mathrm{SPLIT}(R)$ for all dyadic rectangles $R$ is $O(N (\log n)^d + d 2^d N)$. 
After $\mathrm{OPT}(R)$ and $\mathrm{SPLIT}(R)$ have been constructed, we can find the optimal partition by going top-down. This would be a lower order computation.

Based on the discussion above, it is not hard to see that to compute~\eqref{eqn:qdcart2} we just have to modify the above algorithm slightly. We do not need to compute $\mathrm{OPT}(R)$ and $\mathrm{SPLIT}(R)$ for rectangles $R$ with $|R| < \gamma.$ Thus, when we visit dyadic rectangles bottom up according to its size, we just visit the \textit{feasible} rectangles. Also, for a given rectangle $R$, to compute $\mathrm{OPT}(R)$ we now have to compute the sum $\mathrm{OPT}(R_1) + \mathrm{OPT}(R_2)$ for each possible non trivial dyadic split of $R$ into $R_1,R_2$  where $R_1,R_2$ are both \textit{feasible} and then compute the minimum of these numbers and $\sum_{i\in R} \rho_{\tau}(y_i -  q_{\tau}(y_R)).$


\section{Experiments}
\label{sec:experiments}

\subsection{Comparisons in 1d}

We now proceed to evaluate the performance of QDCART in the 1d setting. In our simulations we consider as benchmarks the  the penalized quantile fused lasso (PQFL) proposed in \cite{brantley2019baseline} and studied in \cite{padilla2020risk}, and the univariate mean regression DCART method from \cite{donoho1997cart}. For our evaluations in this subsection, for QDCART and DCART  we consider a grid of $25$ values of $\lambda$ given as  $\{2^{-2},2^{-1.75},\ldots,2^4\}$ and we set $\gamma=8$. As for PQFL, we take $\lambda$ such that $\log \lambda \in \{1+ \frac{j(7.5-1)}{99}\,:\, j \in \{0,1,\ldots,99\}  \}$. Then, for each method and choice of tuning parameter we calculate the average mean squared error, averaging over 100 data sets generated from different scenarios and with $n\in \{512,1024\}$. For each method and each scenario we then report the optimal MSE. The only remaining ingredient is to explain the different scenarios for generating data that we consider. These are described next.


For each scenario, we generate the  data $y \in \mathbb{R}^{L_{1,n}}$ as 
\[
y_i \,=\,  \theta^*_i +\epsilon_i,
\]
for $i\in L_{1,n}$ and some $\theta^*,\epsilon \in \mathbb{R}^{L_{1,n}}$. We now explain the constructions of $\theta^*$ and $\epsilon$ for the different scenarios.

\begin{figure}[t!]
	\begin{center}
		\includegraphics[width=2.32in,height=2.2in]{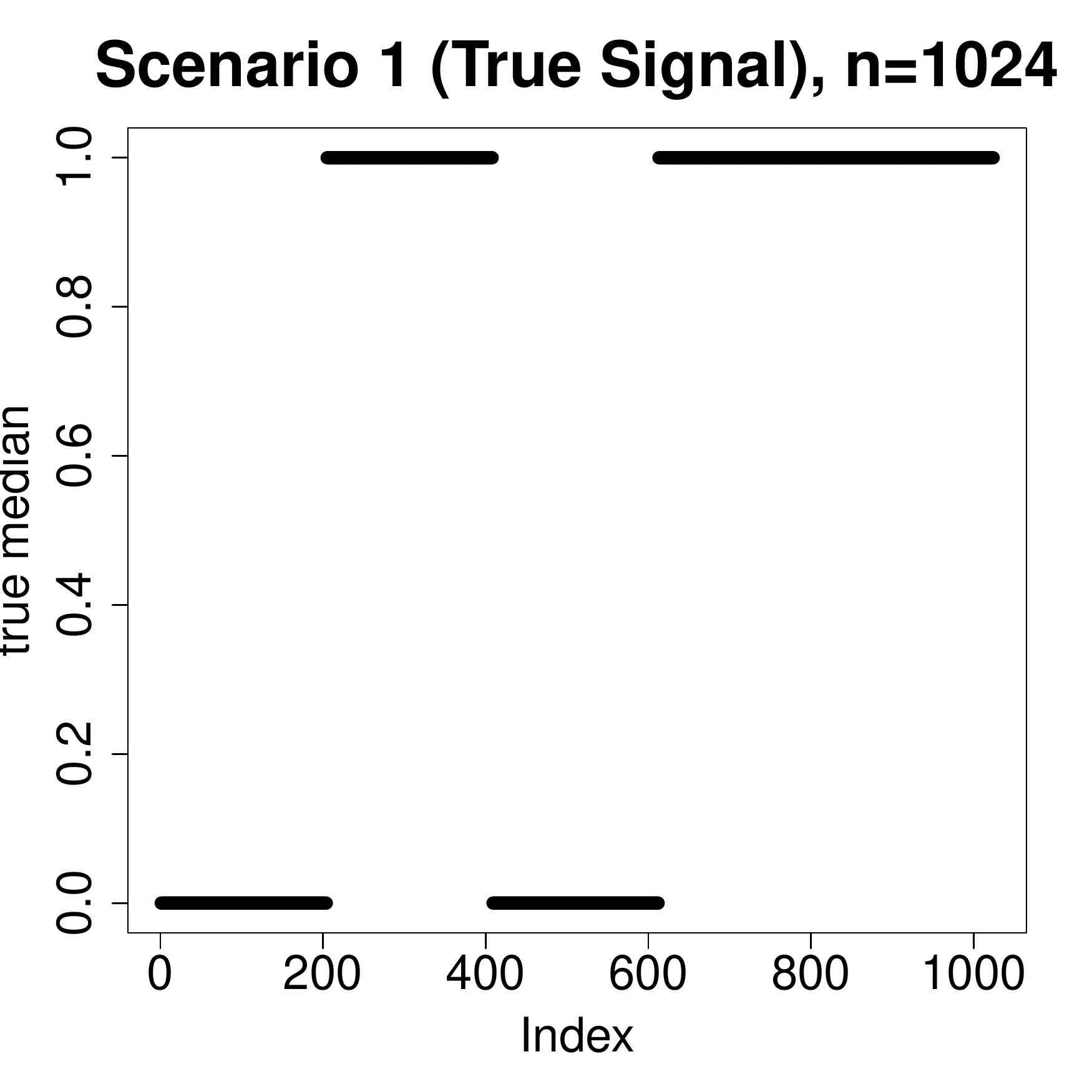}
		\includegraphics[width=2.32in,height=2.2in]{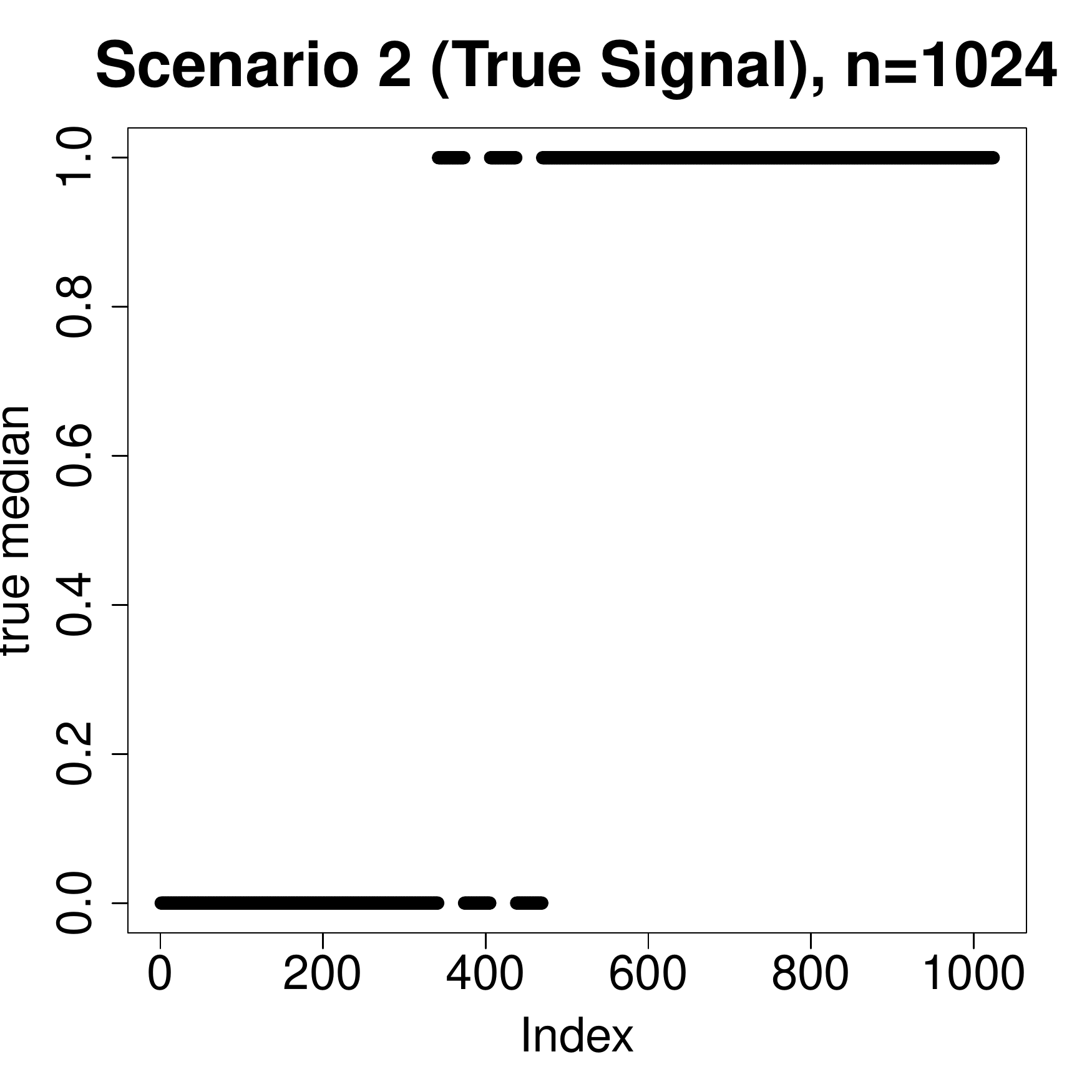}
		\includegraphics[width=2.32in,height=2.2in]{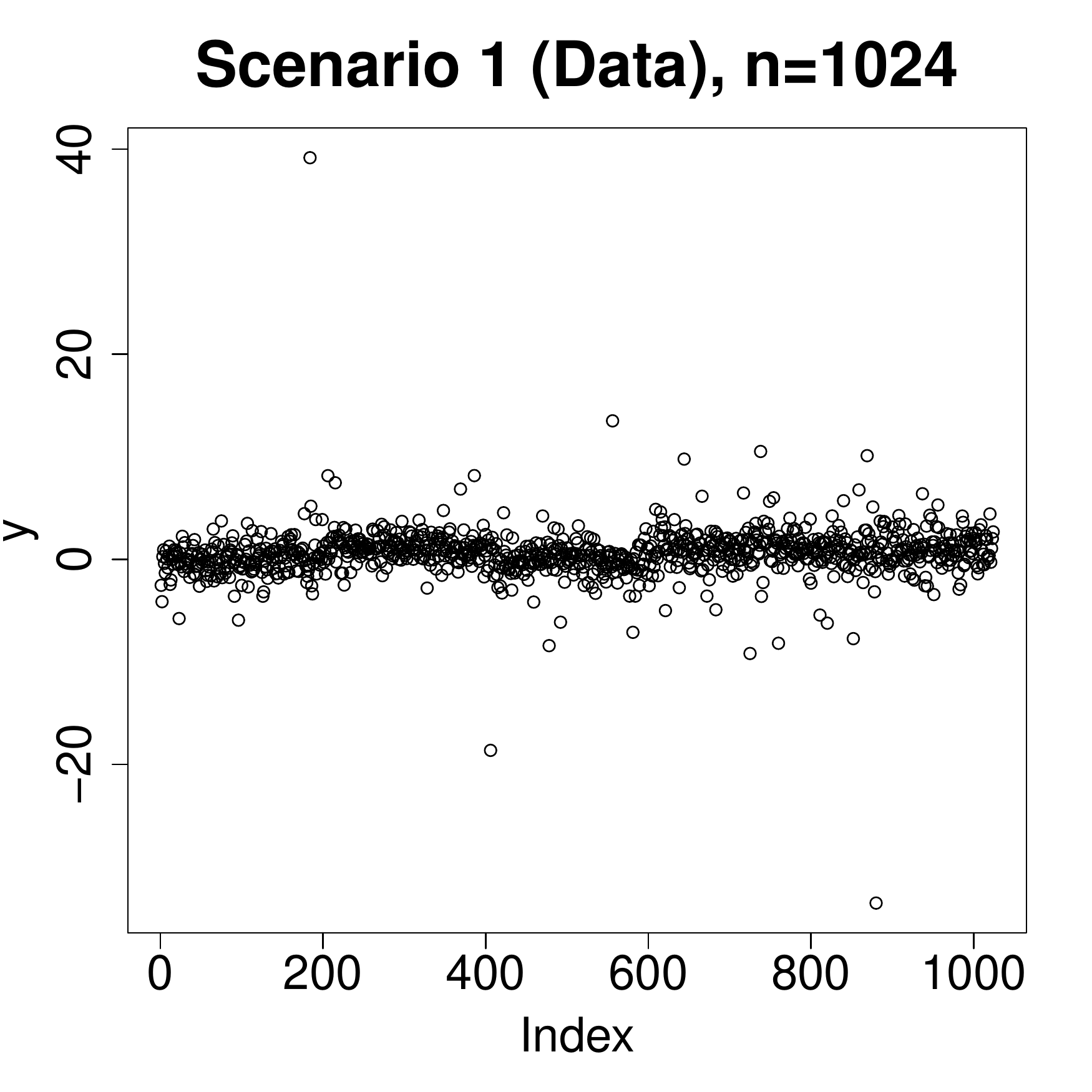} 
		\includegraphics[width=2.32in,height=2.2in]{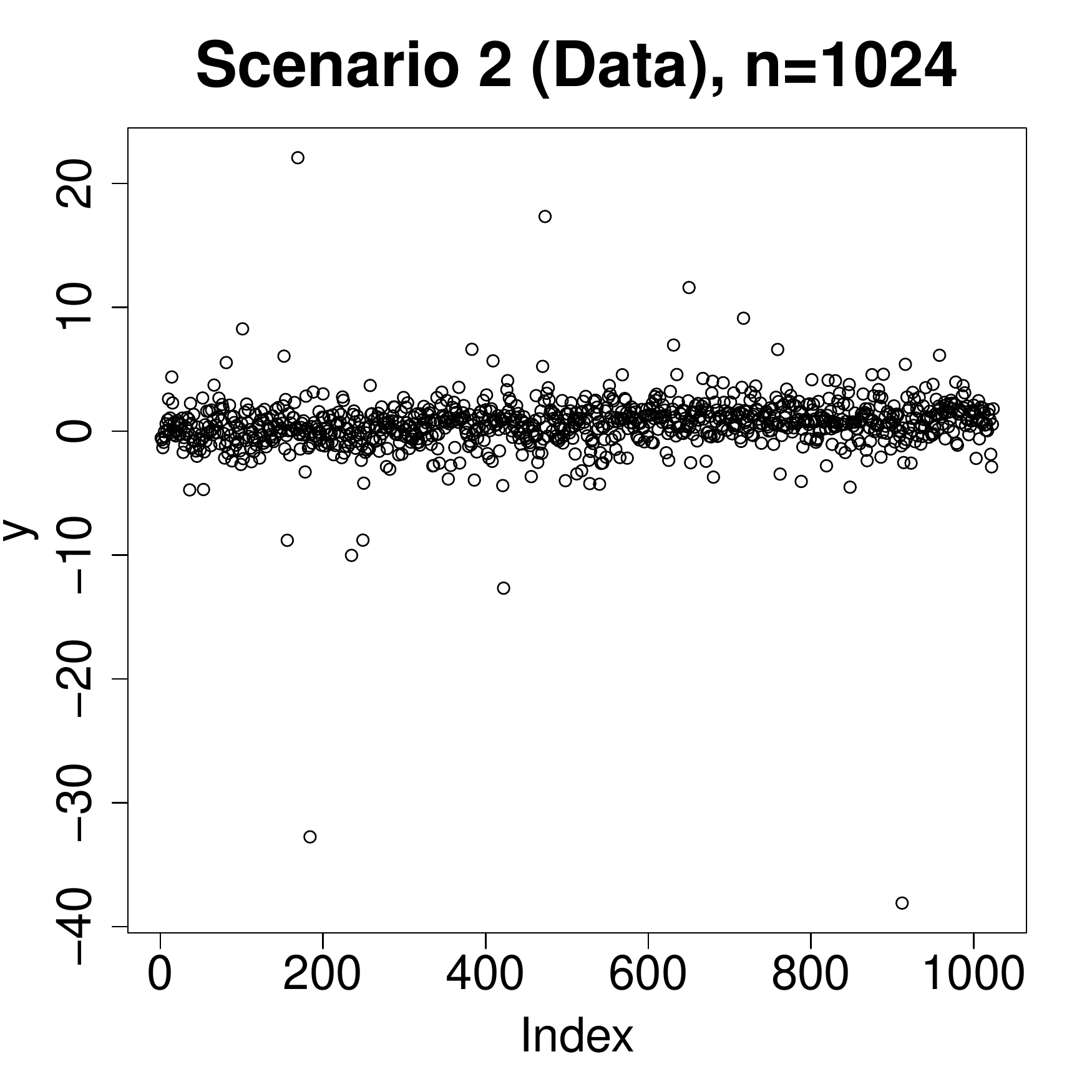} 
		\includegraphics[width=2.32in,height=2.2in]{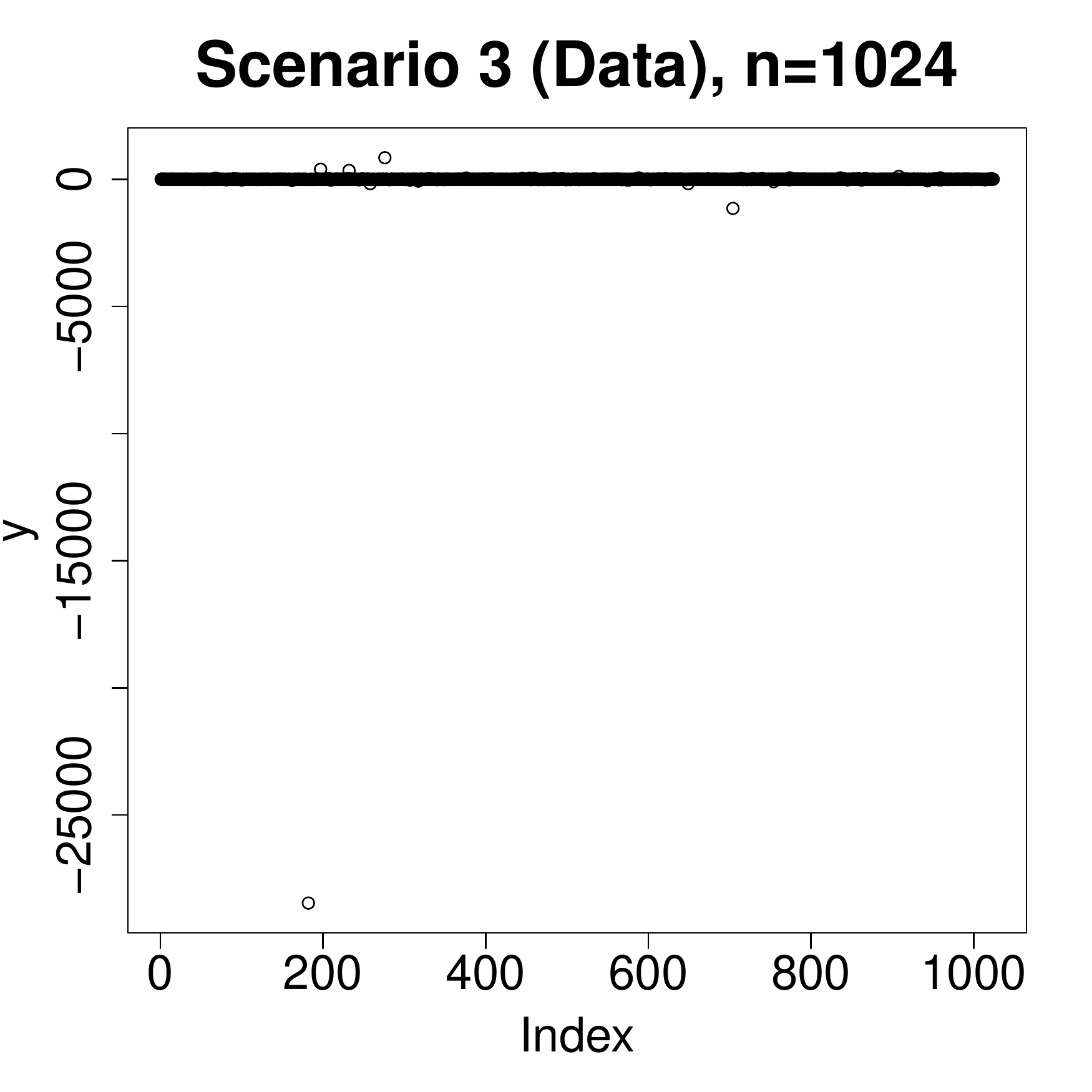} 
		\includegraphics[width=2.32in,height=2.2in]{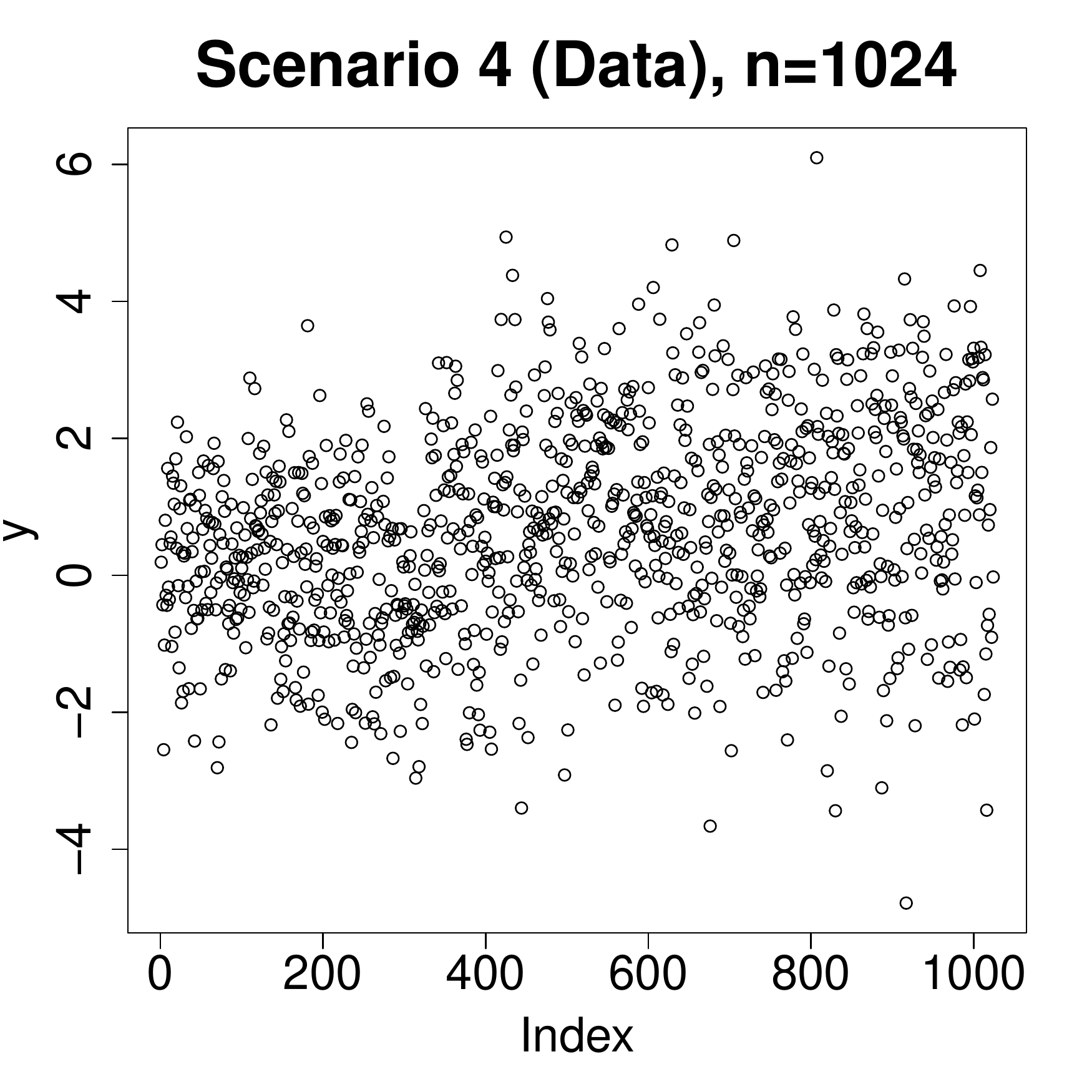}  
		\caption{ 		\label{fig1} The top two panels show the true signal (median) $\theta^*$ for Scenarios 1 and 2. The rest of panels show an example of data generated under each of the scenarios. }
	\end{center}
\end{figure}

\textbf{Scenario 1. (Large  Segments).} In this case $\theta^* \in \mathbb{R}^{L_{1,n}}$ satisfies
\[
\theta_i^*\,=\,\begin{cases}
	1 &\text{if}\,\,\, i \in [\floor{n/5}+1, 2\floor{n/5}]\cup [3\floor{n/5}+1, n]\\
	0 &\text{otherwise,}\,\,\,\\
\end{cases}
\]
and we generate $\epsilon_i \overset{\mathrm{ind}}{\sim} t(2.5)$ where $t(2.5)$ denotes the  t-distribution with $2.5$ degrees of freedom.\\

\begin{table}[t!]
	\centering
	\caption{\label{tab1}  Average mean squared error  $\frac{1}{n} \sum_{i=1}^n( \theta_i^*-\hat{\theta}_i )^2 $,  averaging over 100 Monte carlo simulations for the different methods considered. Captions are described in the text.  }
	\medskip
	\setlength{\tabcolsep}{8.1pt}
	\begin{small}
		\begin{tabular}{ |rrrrr|rrrrr|}
			\hline
			$n$ & Scenario & PQFL   & QDCART    &DCART   & $n$ & Scenario & PQFL   & QDCART&DCART \\
			\hline	
			512 & 1  &  0.124&\textbf{0.094}&  3.08&   512&   3 & \textbf{0.177} &  0.252&249054.2\\
			1024 & 1  &\textbf{0.047}&0.066&   2.52&  1024&   3 & \textbf{0.118} &0.249  &104763.3\\		
			\hline
			512 & 2  & 0.084 &\textbf{0.063} &3.17  &   512&   4 & 0.090& \textbf{0.070}  &0.114\\
			1024 & 2  &0.064  &\textbf{0.047}&  2.99&  1024&   4 & 0.077&  \textbf{0.054} &0.106\\		     
			\hline  
		\end{tabular}
	\end{small}
\end{table}

\textbf{Scenario 2.  (Large and Small Segments).} We generate $\epsilon$ as in Scenario 1 and set 
\[
\theta_i^*\,=\,\begin{cases}
	1 &\text{if}\,\,\, i \in [\floor{n/3}+1, \floor{n/3}+\floor{n/32}]\cup  [\floor{n/3}+2\floor{n/32}+1,\floor{n/3}+3\floor{n/32}]\cup\\ &\,\,\,\,\,\,\,[\floor{n/3}+4\floor{n/32}+1,n],\\
	0 &\text{otherwise.}\,\,\,\\
\end{cases}
\]\\

\textbf{Scenario 3. (Large Segments and Cauchy Errors).} We take  $\theta^* \in \mathbb{R}^{L_{1,n}}$ as in Scenario 1 and generate $\epsilon_i \overset{\mathrm{ind}}{\sim} \mathrm{Cauchy}(0,1)$. \\

\textbf{Scenario 4. (Large and Small Segments, and Heteroscedastic   Errors).} The vector $\theta^*$ is the same as in Scenario $2$ and $\epsilon$ satisfies for all $i$ that
\[
\epsilon_i\,=\,\nu_i \times \sqrt{\frac{2i}{n}+1},
\]
where $\nu_i \overset{\mathrm{ind}}{\sim}  \mathrm{N}(0,1)$.\\


A visualization of data generated under each scenario is given in Figure \ref{fig1}. The results of our comparisons are given in Table \ref{tab1}. Overall, we can see that the QDCART estimator is competitive against QFL. In Scenario $2$, where some of the constant pieces of the true signal are very small, we see that the QDCART estimator performs better. This is in agreement with Corollary~\ref{cor3} where no minimum length assumption is needed for the QDCART to attain near parametric rates. Observe that the mean regression DCART estimator performs poorly under heavy tailed scenarios.

\subsection{Comparisons in 2d}

We now proceed to evaluate the performance of QDCART for 2d grid graphs and use DCART and QTVD as benchmarks. For our experiments in this subsection the tuning parameter $\lambda$ for QDCART and  DCART  is taken such that $\log_{10}(\lambda) $ is in the set $\{-1+ \frac{6.5j}{59}\,:\,j \in \{0,1,\ldots,59\}  \}$. As for the tuning parameter $\lambda$ for QTVD we take it such that  $\log_{2}(\lambda) $ is in the set $\{-1+ \frac{7j}{19}\,:\,j \in \{0,1,\ldots,19\}  \}$. As before,  for each method and choice of tuning parameter we calculate the average mean squared error, averaging over 100 data sets generated from different scenarios. We set 
$d=2$ and  $n\in \{64,128\}$. For each method and each scenario we then report the optimal MSE.  Next we describe the different generative models, where in each case the data are generated as
\[
y_{i,j} = \theta_{i,j}^* +\epsilon_{i,j}
\]
where $\epsilon_{i,j}  \overset{\mathrm{ind}}{\sim} t(2.5)$ 
for $i,j \in \{1,\ldots,n\}$  and with  $\theta^*\in \mathbb{R}^{n \times n}$.\\

\textbf{Scenario 5.}  We set 
\[
\theta_{i,j}^*\,=\,\begin{cases}
	1 &\text{if}\,\,\, n/5< i < 3n/5 \,\text{ and }\, n/5< j < 3n/5,\\
	0 & \text{otherwise.}
\end{cases}
\]

\begin{figure}[bp!]
	\begin{center}
		\includegraphics[width=2.48in,height=2.08in]{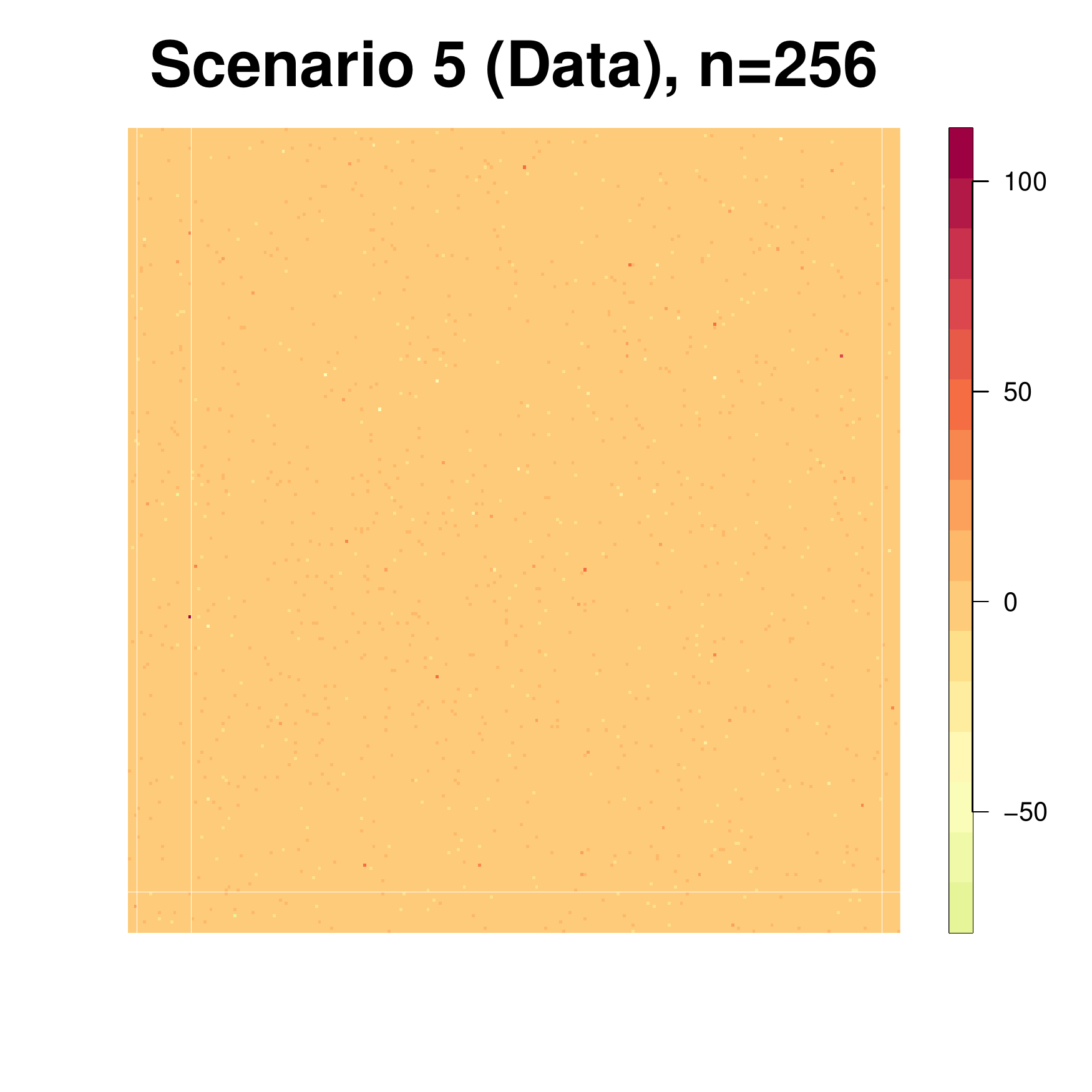}
		\includegraphics[width=2.48in,height=2.08in]{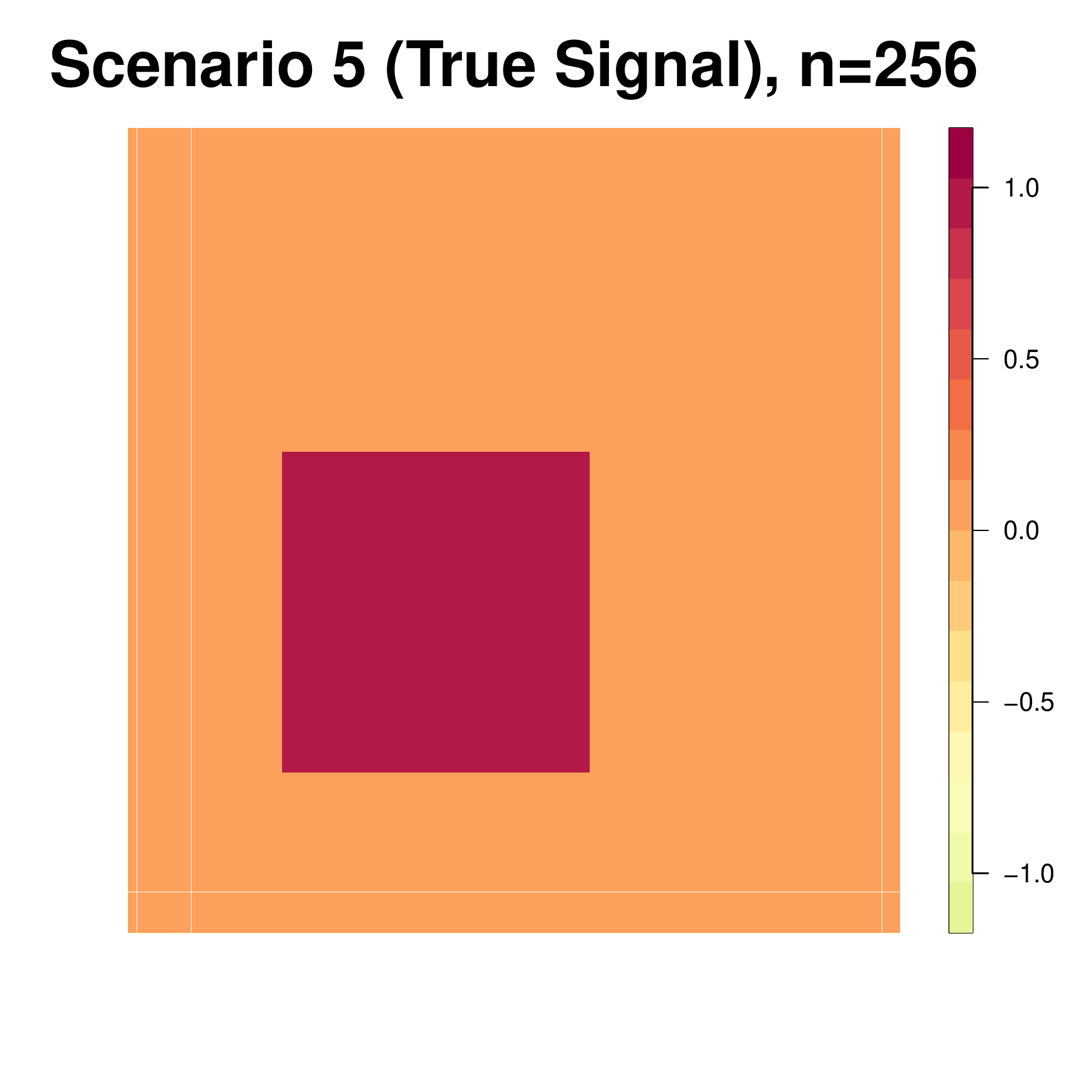}
		\includegraphics[width=2.58in,height=2.08in]{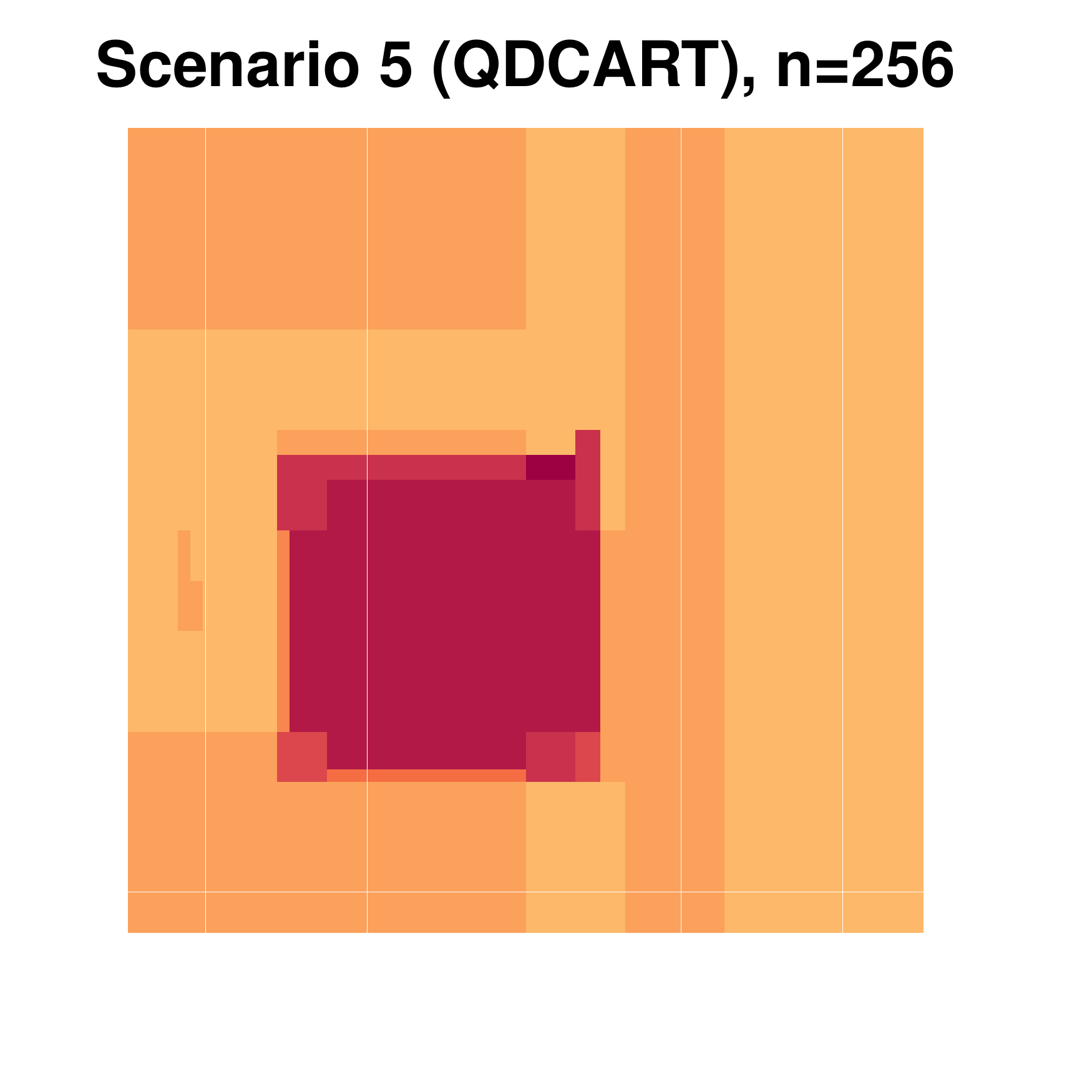}
		\includegraphics[width=2.48in,height=2.08in]{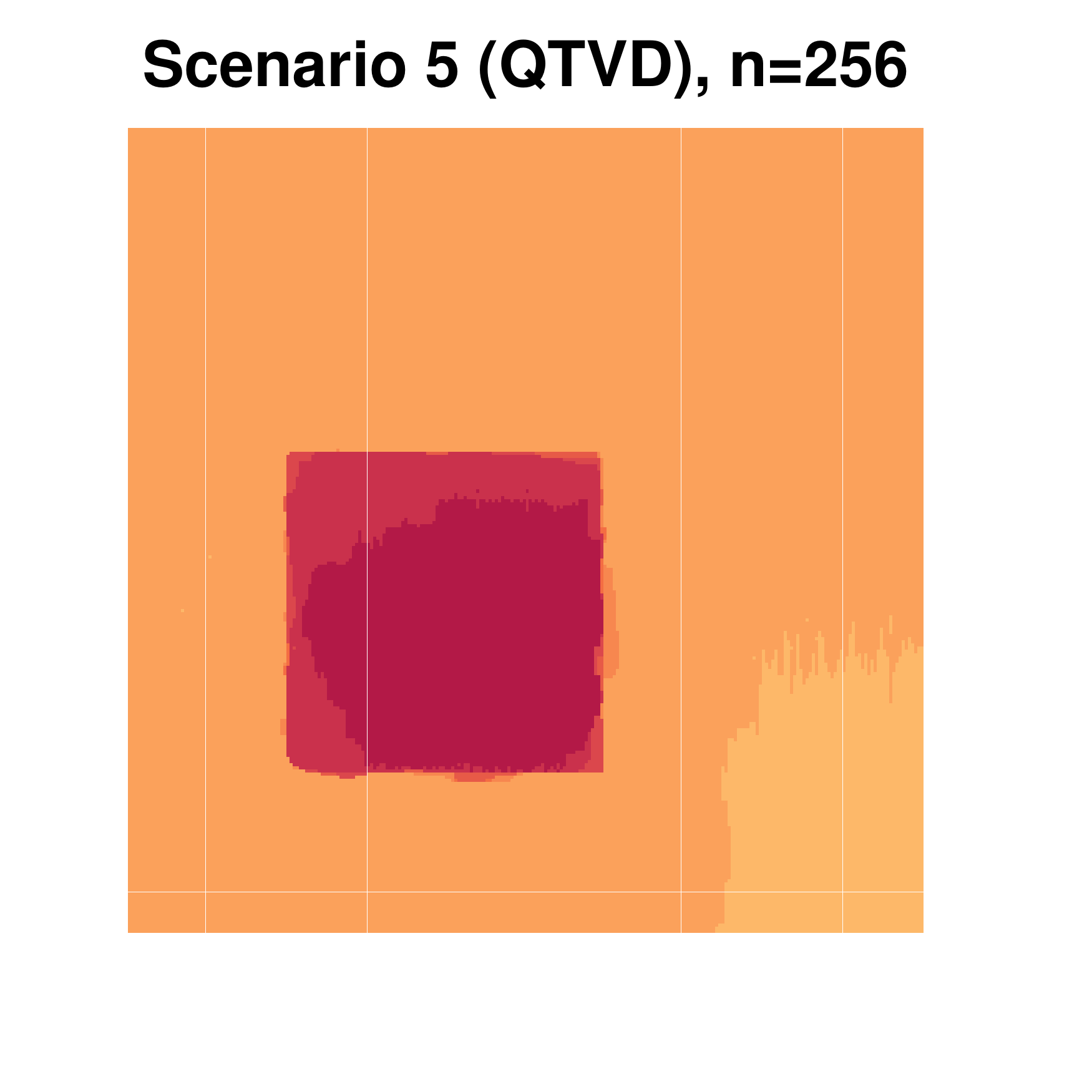}
		\includegraphics[width=2.48in,height=2.08in]{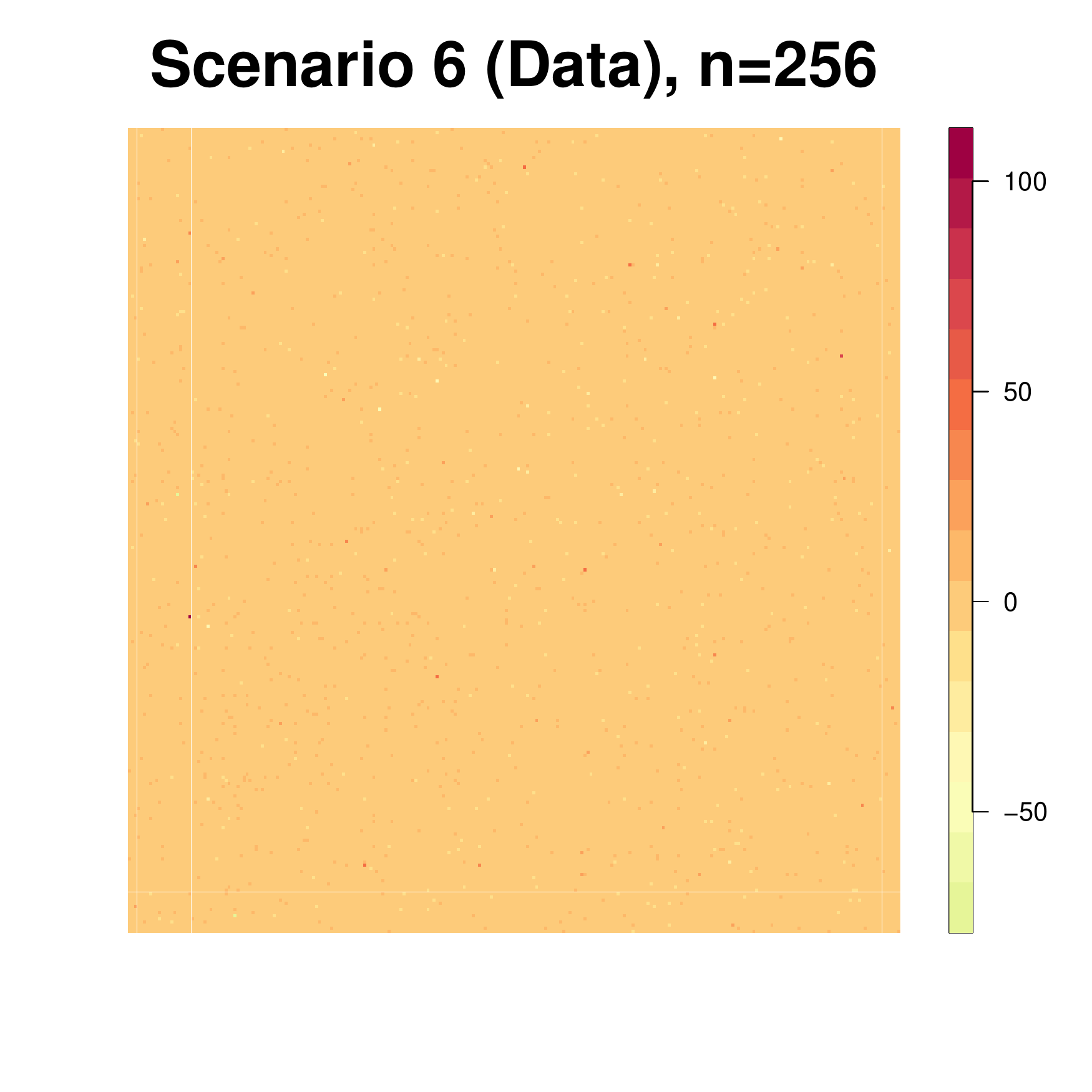}
		\includegraphics[width=2.48in,height=2.08in]{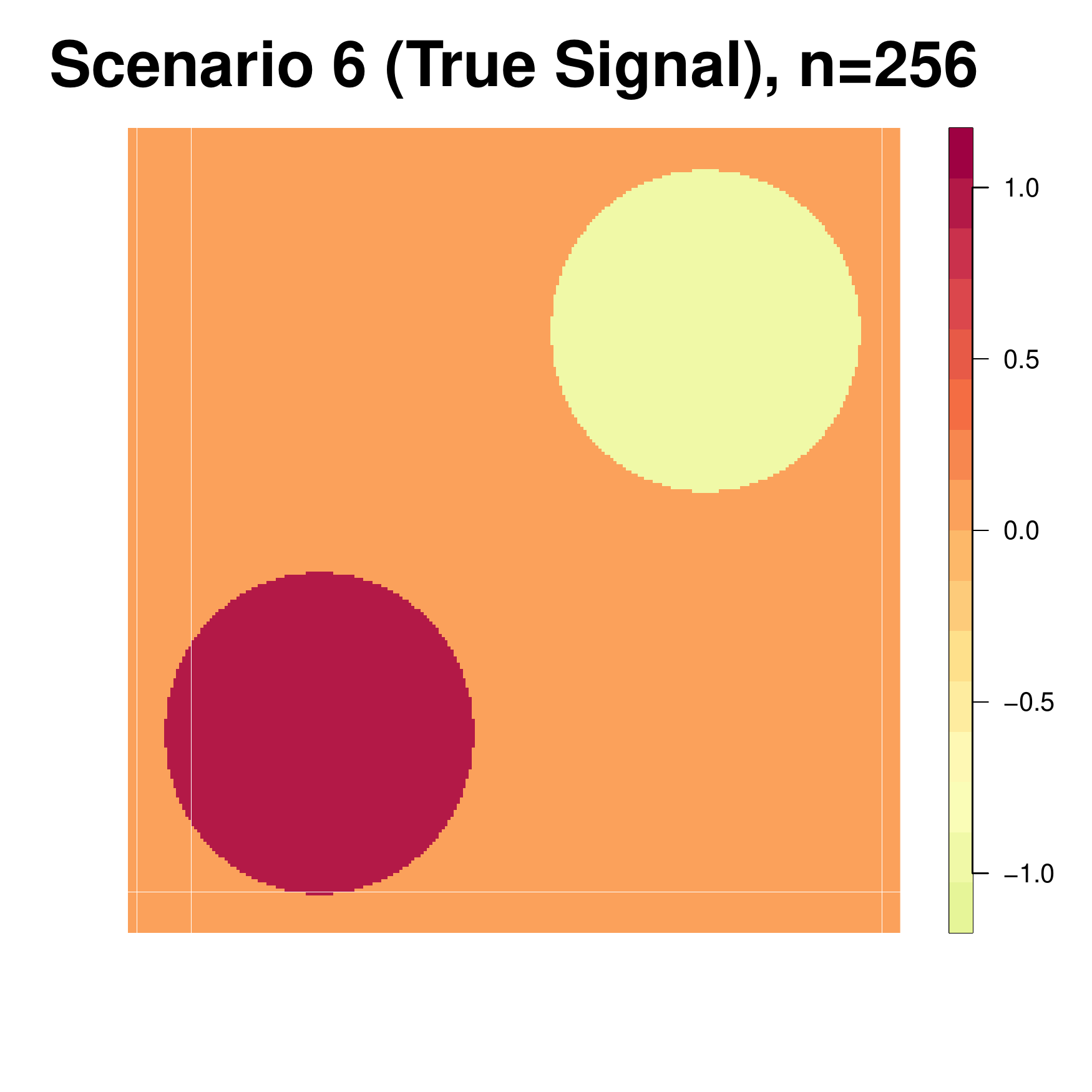}
		\includegraphics[width=2.48in,height=2.08in]{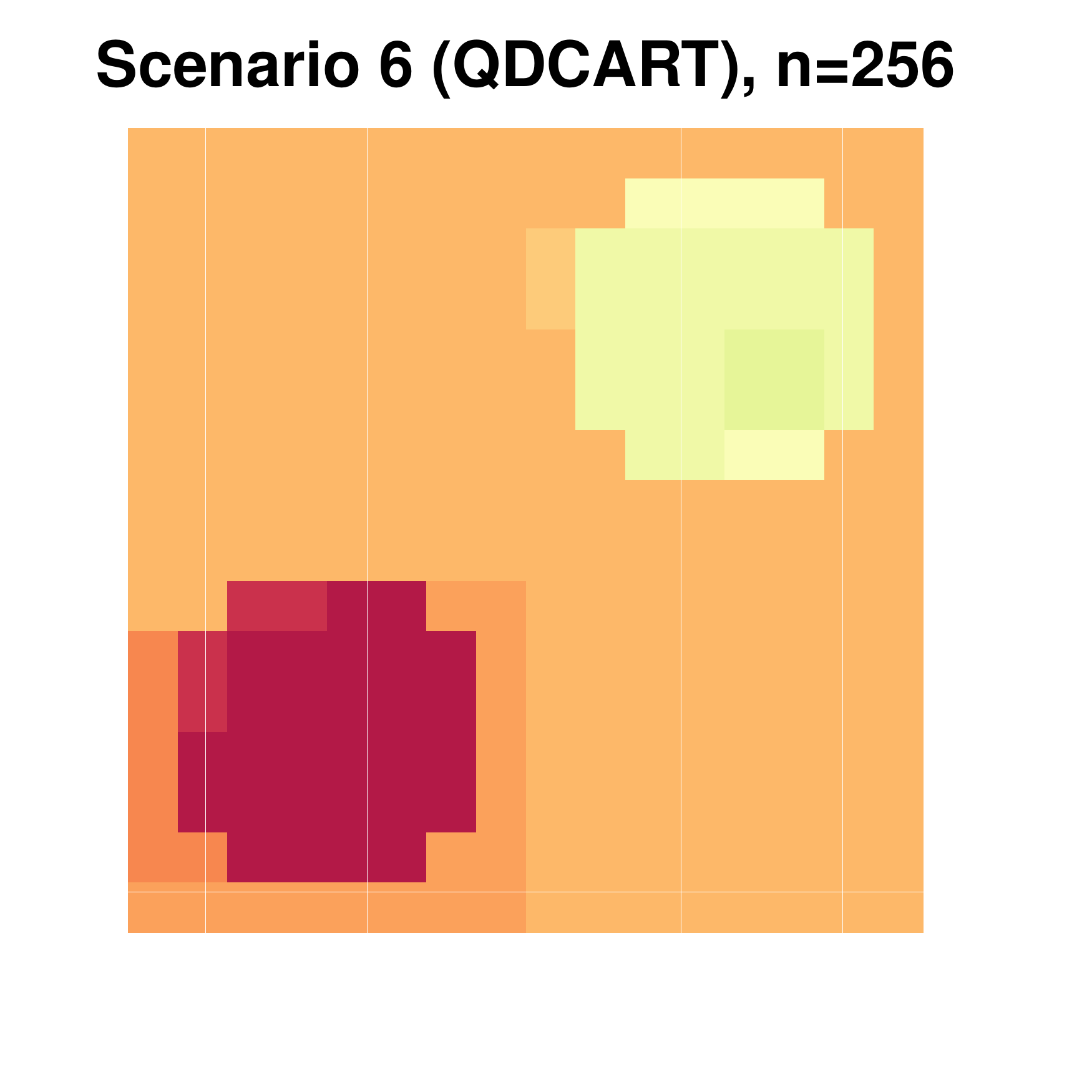}
		\includegraphics[width=2.48in,height=2.08in]{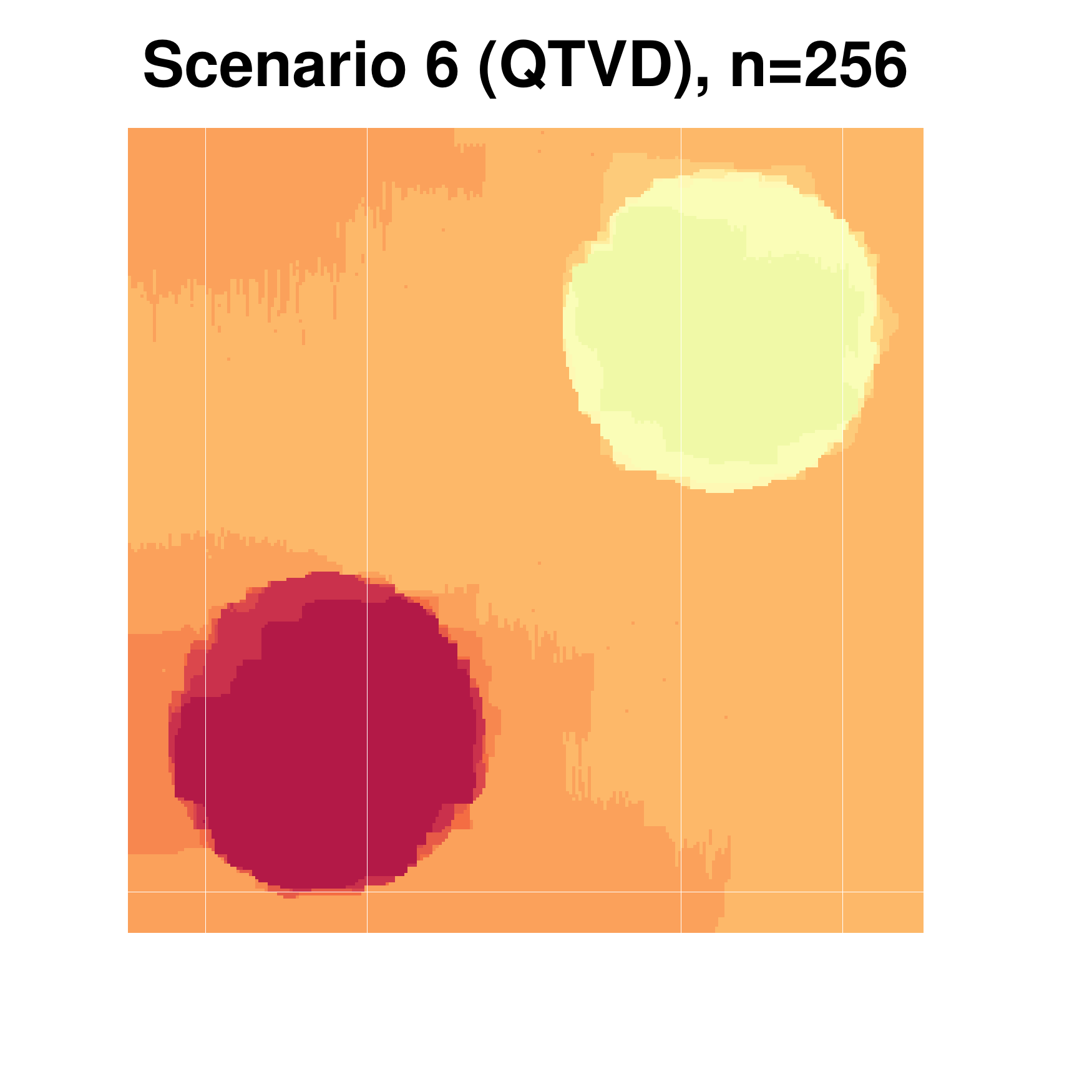}	
		\caption{ 		\label{fig2} The first row of panels shows from left to right an instance of data generated under Scenario 5 and the true median signal $\theta^*$. The  second row shows  the estimate provided by QDCART and the one by QTVD. The third and fourth  rows of panels show the corresponding plots for Scenario 6. }
	\end{center}
\end{figure}

\begin{figure}[h!]
	\begin{center}
		\includegraphics[width=2.48in,height=2.08in]{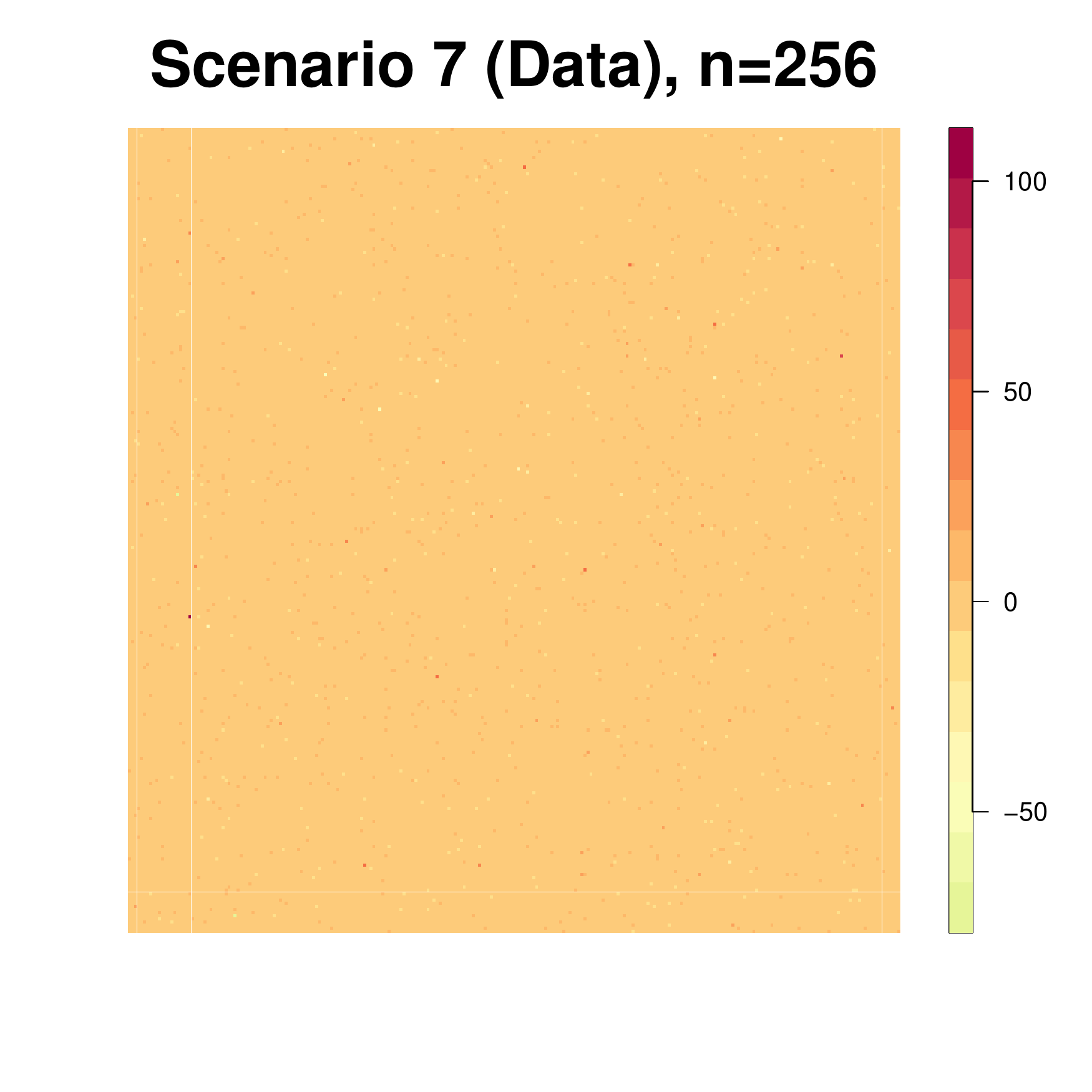}
		\includegraphics[width=2.48in,height=2.08in]{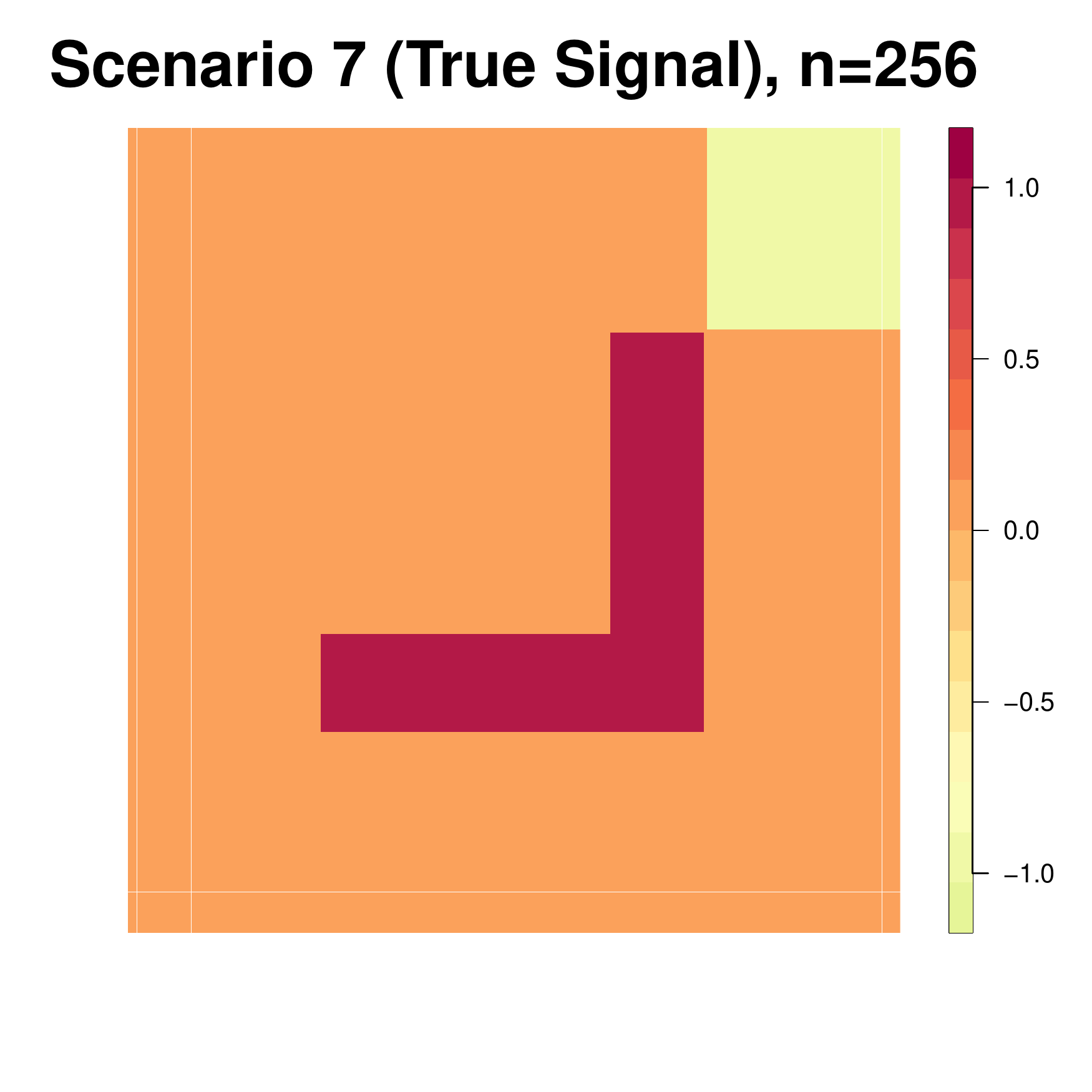}
		\includegraphics[width=2.48in,height=2.08in]{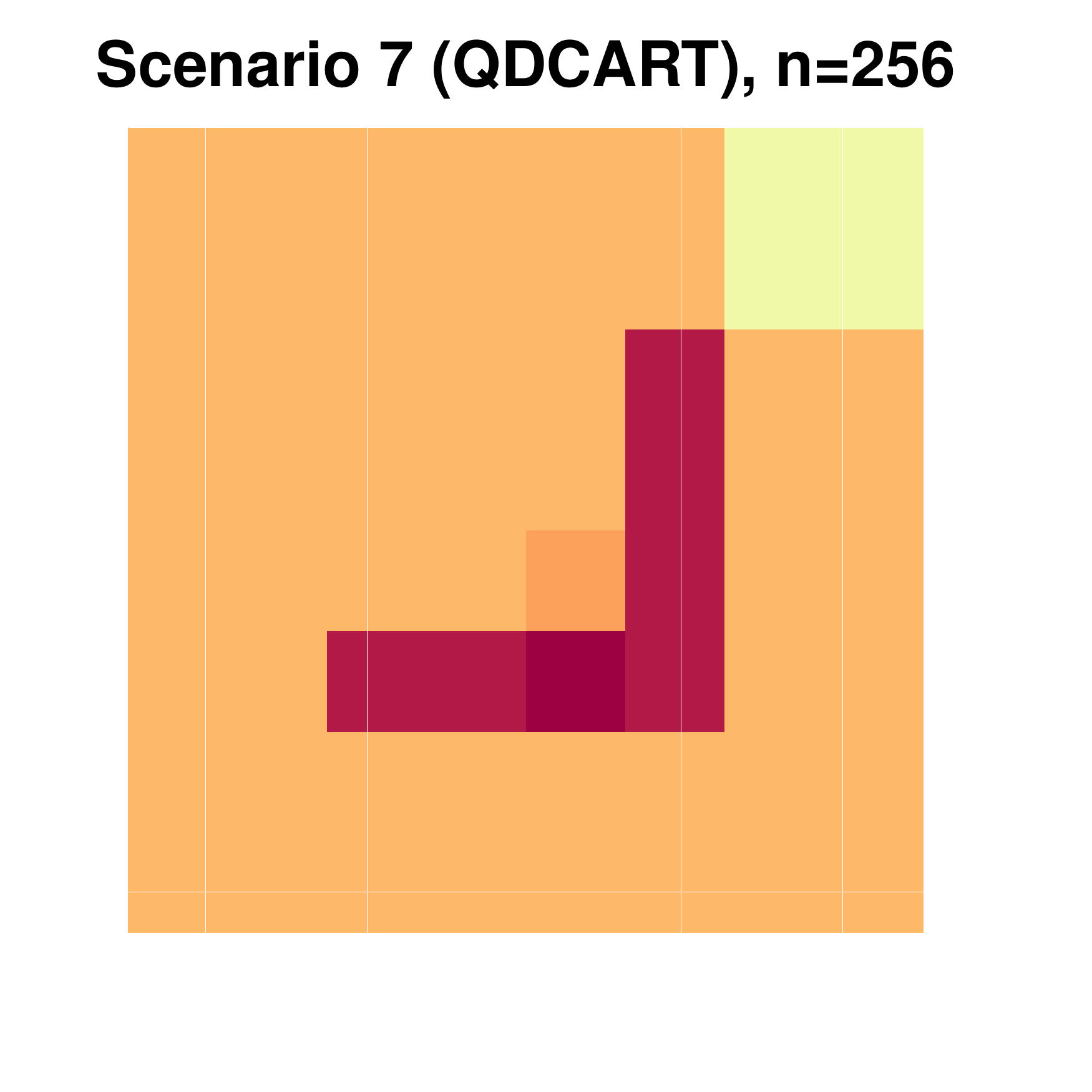}
		\includegraphics[width=2.48in,height=2.08in]{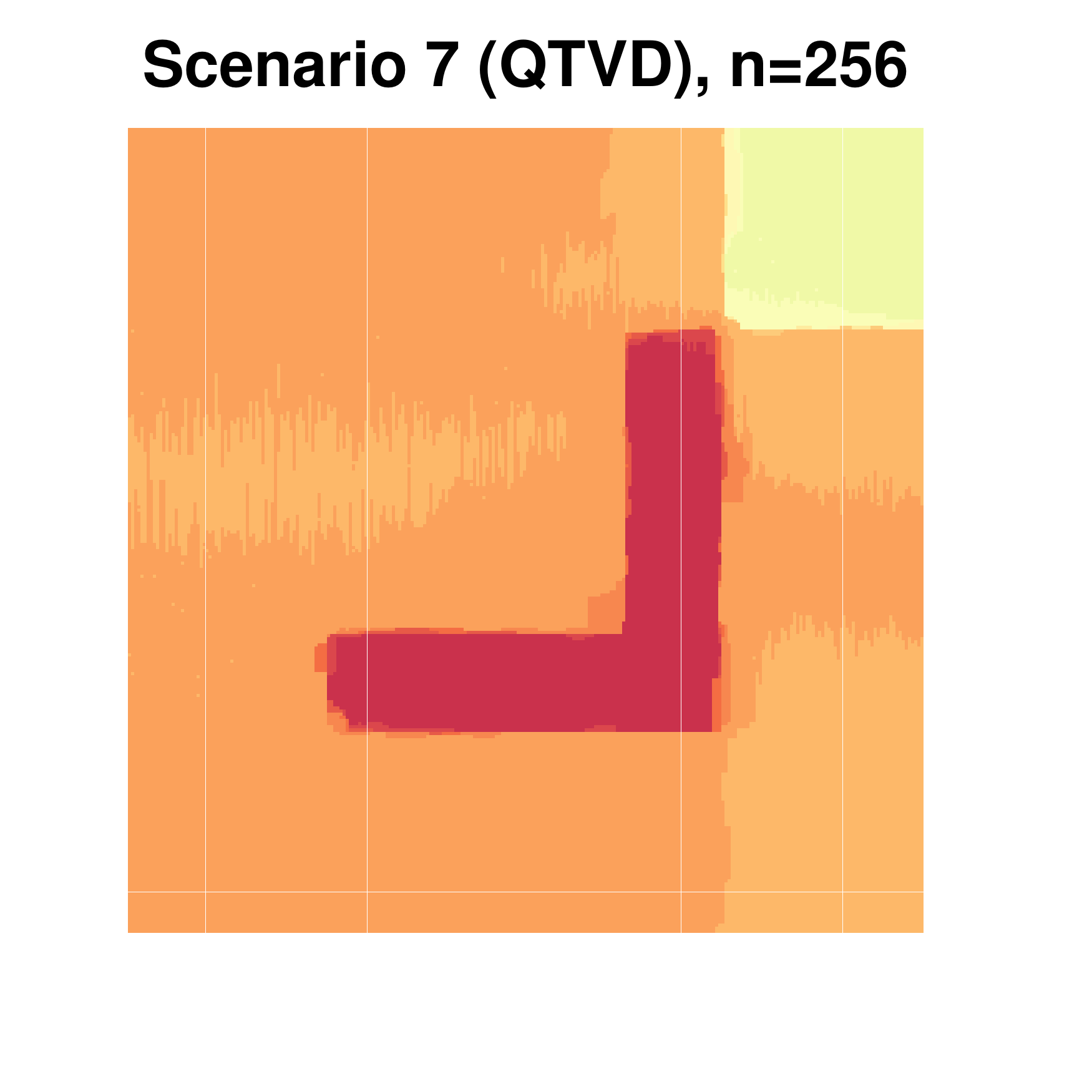}
		\caption{ 		\label{fig3} The first row of panels shows from left to right an instance of data generated under Scenario 7 and  the true median signal $\theta^*$. The second row shows  the estimate provided by QDCART and the one by QTVD. }
	\end{center}
\end{figure}

\textbf{Scenario 6.} Now we take $\theta^*$ satisfying 
\[
\theta_{i,j}^*\,=\,\begin{cases}
	1 &\text{if}\,\,\, (i-n/4)^2+(j-n/4)^2< (n/5)^2, \\
	-1 &\text{if}\,\,\, (i-3n/4)^2+(j-3n/4)^2< (n/5)^2,  \\
	0 & \text{otherwise.}
\end{cases}
\]

\begin{table}[t!]
	\centering
	\caption{\label{tab2}  Average mean squared error  $\frac{1}{n} \sum_{i=1}^n( \theta_i^*-\hat{\theta}_i )^2 $,  averaging over 100 Monte carlo simulations for the different methods considered and with data generated from Scenarios 5 and 6. Captions are described in the text.  }
	\medskip
	\setlength{\tabcolsep}{10pt}
	\begin{small}
		\begin{tabular}{ |r|rrrr|rrrr|}
			\hline
			$n$ & Scenario & QTVD   & QDCART   &DCART    & Scenario & QTVD     & QDCART &DCART\\
			\hline	
			64 & 5  & \textbf{ 0.030} &    0.048   &      0.139  & 6&  \textbf{0.057}& 0.096& 0.250\\
			128& 5  &  \textbf{ 0.013}&   0.021   &    0.134    &  6&\textbf{0.023}&0.035&0.252\\
			256& 5  &    \textbf{0.004}&  0.009         &     0.133     &  6  &\textbf{0.011}&0.026&0.251\\
			\hline
		\end{tabular}
	\end{small}
\end{table}

\begin{table}[t!]
	\centering
	\caption{\label{tab3}  Average mean squared error  $\frac{1}{n} \sum_{i=1}^n( \theta_i^*-\hat{\theta}_i )^2 $,  averaging over 100 Monte carlo simulations for the different methods considered and with data generated from Scenarios 7. Captions are described in the text.  }
	\medskip
	\setlength{\tabcolsep}{12pt}
	\begin{small}
		\begin{tabular}{ |r|rrrr|}
			\hline
			$n$ & Scenario &QTVD     & QDCART   &DCART   \\
			\hline	
			64 & 7  &0.060  &  \textbf{0.033}   &  0.157 \\
			128& 7  &0.022&    \textbf{0.009} &    0.163\\
			256& 7  & 0.009&    \textbf{0.004} &0.166 \\
			\hline
		\end{tabular}
	\end{small}
\end{table}

\begin{figure}[h!]
	\begin{center}
		\includegraphics[width=2.68in,height=2.08in]{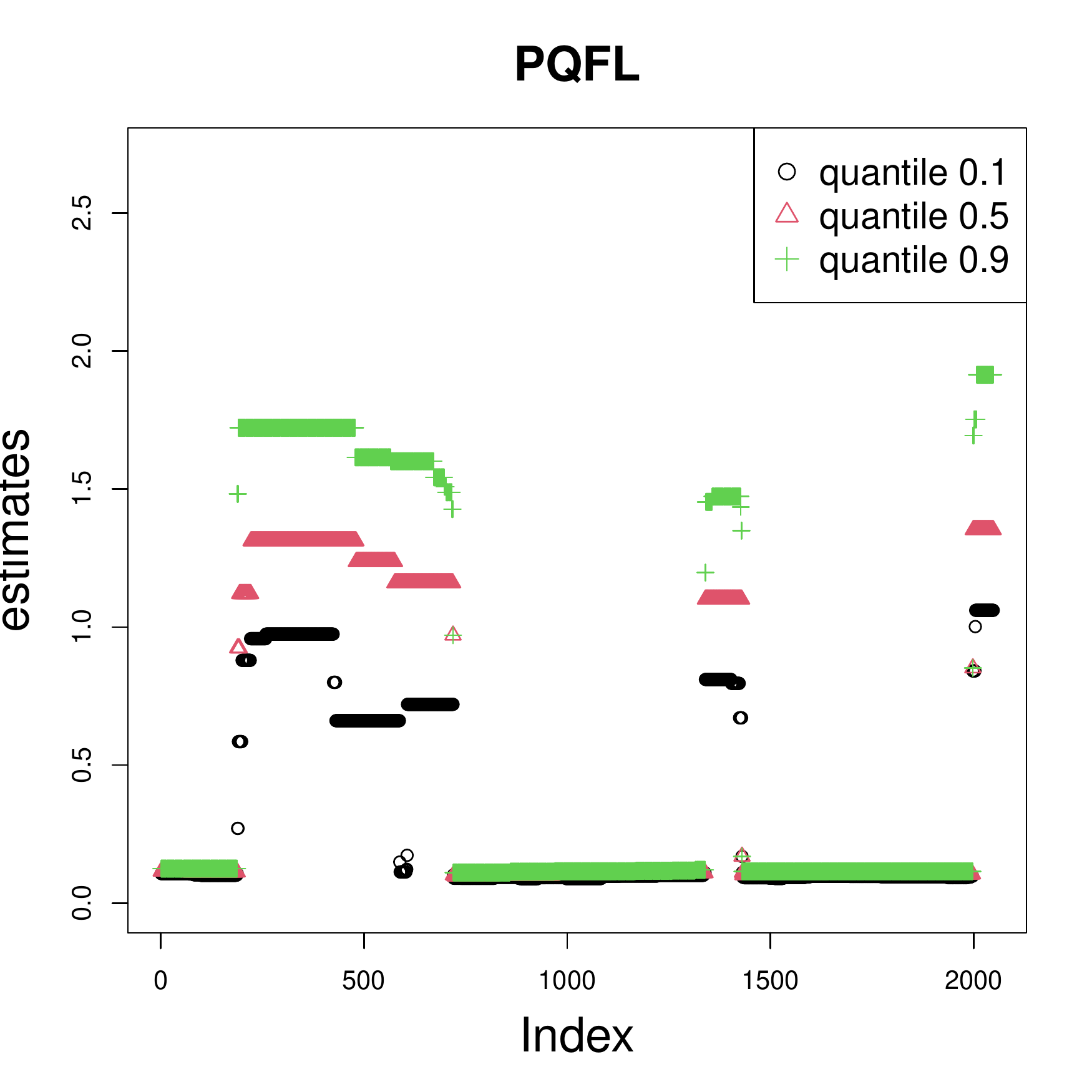}
		\includegraphics[width=2.68in,height=2.08in]{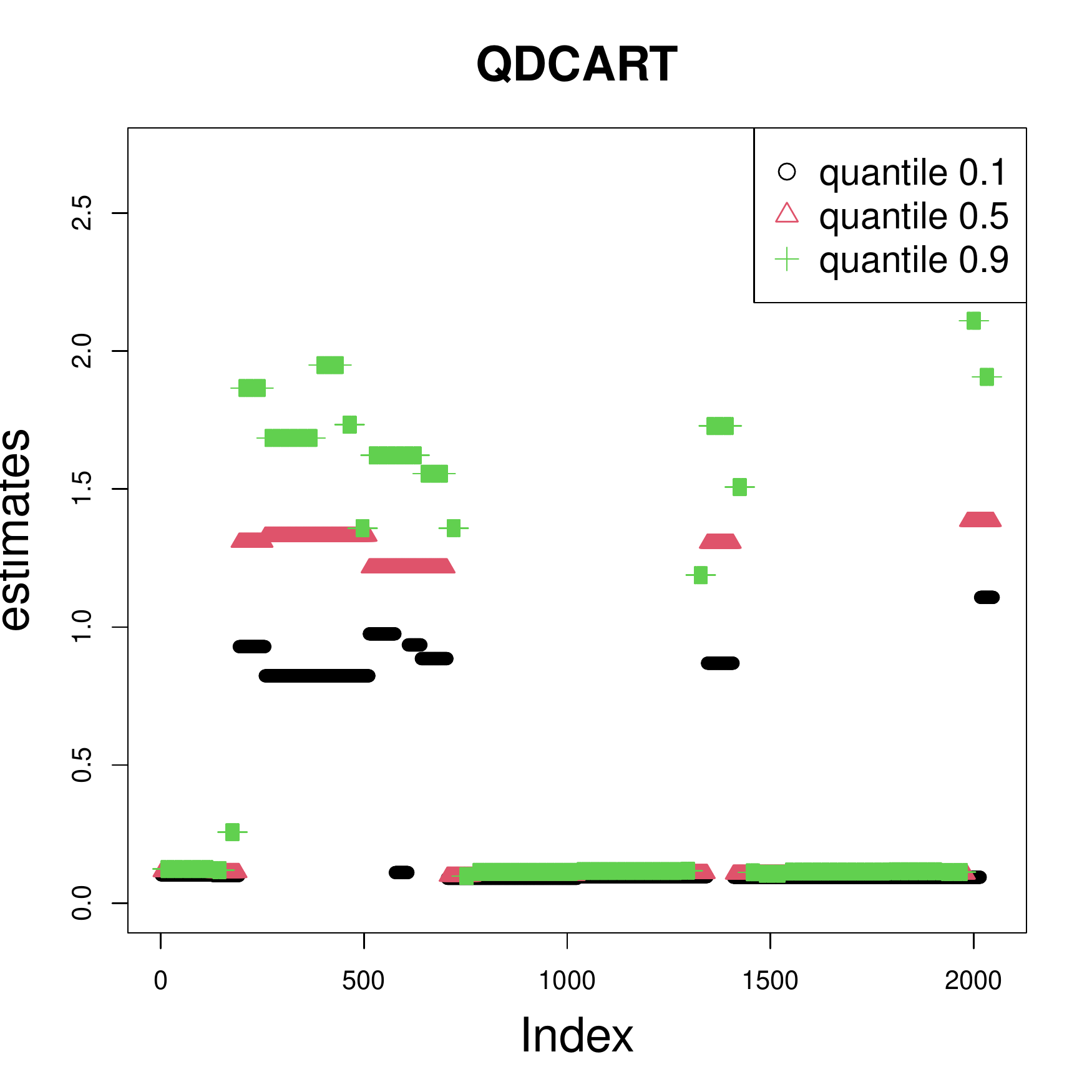}
		\includegraphics[width=2.58in,height=2.08in]{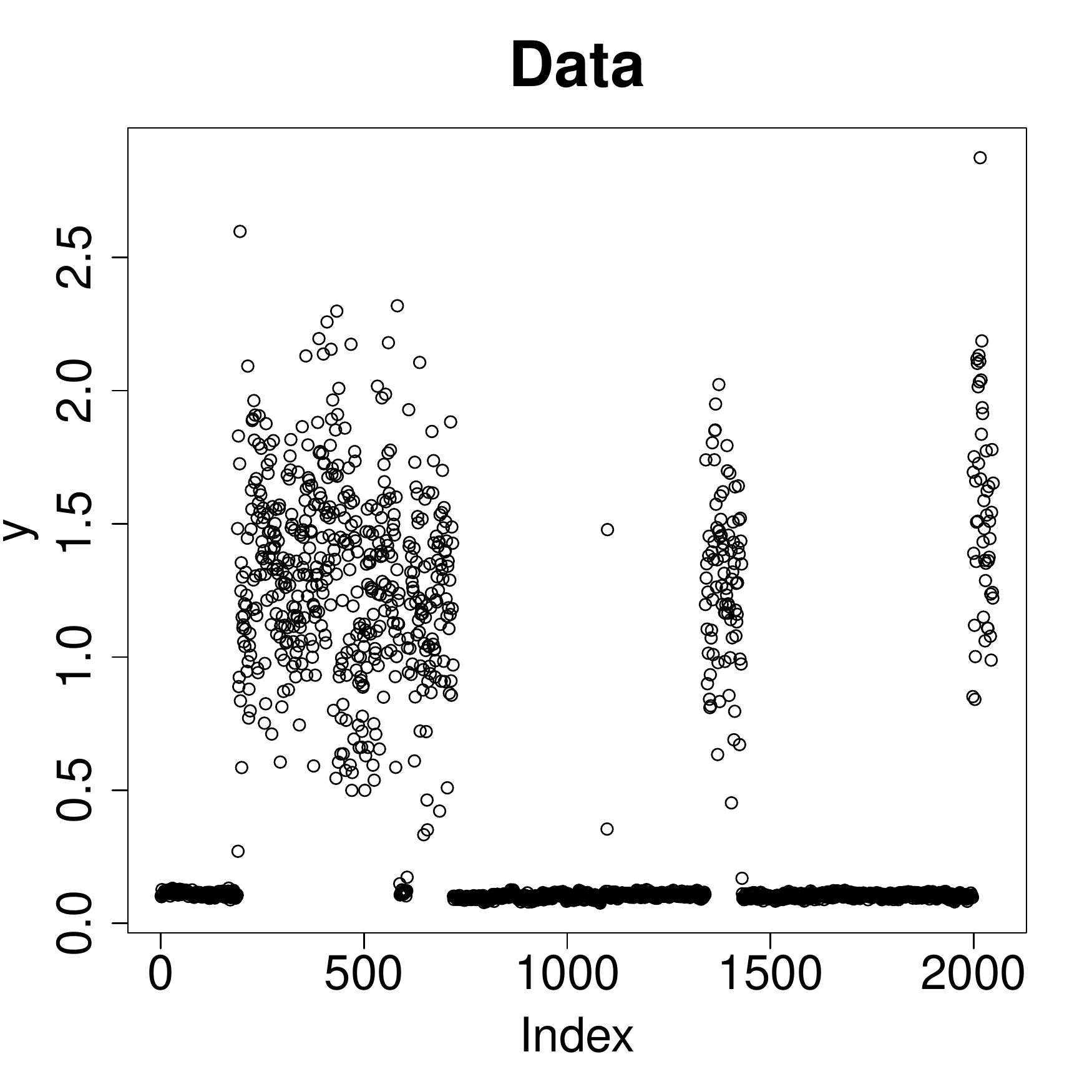}
		\caption{ 		\label{fig4} The top two panels show the estimated PQFL and QDCART for $\tau \in \{0.1,0.5,0.9\}$ when using the ion data. The bottom panel then shows the ion raw data. }
	\end{center}
\end{figure}

\textbf{Scenario 7.} For this model we let 
\[
\theta_{i,j}^*\,=\,\begin{cases}
	1 &\text{if}\,\,\,n/4 <i<3n/4 \,\text{ and } \, n/4< j< n/4+n/8,  \\
	1 & \text{if}\,\,\, n/2+n/8<i< 3n/4\,\text{ and } \, n/4+n/8 \leq j < 3n/4,\\
	-1 &\text{if}\,\,\, i > 6n/8 \,\text{ and } \,  j>6n/8,\\
	0 & \text{otherwise.}
\end{cases}
\]

%
%

A visualization of the data generated under Scenarios 5--7 is provided in Figures \ref{fig2}--\ref{fig3}. There we can see that QDCART can be competitive against QTVD.

A more comprehensive evaluation of performance comparisons is provided in Tables \ref{tab2}--\ref{tab3}. In Scenario $6$ where the level sets are non-rectangular, QTVD seems to do better than QDCART. In Scenario $7$,  however, QDCART performs slightly better. We believe this is because the level set can be well represented by a dyadic partition. We reiterate here that a potential practical advantage of QDCART over QTVD in an image denoising setting is the fact that the QDCART estimator can be computed in near linear time.


\subsection{Ion channels data}

We conclude our experiments section with a real data example involving ion channels data. Ion channels are a class of proteins expressed by all cells that create pathways for ions  to
pass through the  cell membrane. As explained in \cite{jula2021multiscale}, over time, the ion channel changes its gating behavior by
closing and reopening its pore which leads to a piecewise constant current flow structure. 
The original data that we use was produced by the Steinem Lab (Institute of Organic and Biomolecular Chemistry,
University of Gottingen),  and it was recently analyzed by \cite{jula2021multiscale}. It consists of a single ion channel of
the bacterial porin PorB, a bacterium that plays a role in the pathogenicity of Neisseria gonorrhoeae. The original data is 600000 time instances. For our comparisons we focus on a portion of length 32511 and construct a subsampled vector $y \in \mathbb{R}^{2048}$. Thus, our resulting signal is similar to that in \cite{cappello2021scalable}. A depiction of the data is shown in Figure \ref{fig4}.

Given the signal $y\in \mathbb{R}^{2048}$, we fit both PQFL and QDCART with values $\tau \in \{0.1,0.5,0.9\}$. For PQFL we consider values of the penalty parameter $\lambda$ such that $\log \lambda$ is in the set $\{ 1+\frac{j6.5}{99}\,:\,j \in \{0,\ldots,99\} \}$. As for QDCART we take $\lambda$ such that $\log_2 \lambda$ is in $\{-2+ \frac{ 7j}{25}\,:\,j =0,\ldots,25 \}$. Then for both PQFL and QDCART  we choose the tuning parameter that minimizes the BIC for quantile
regression  criteria from \cite{yu2001bayesian} given as
\[
\mathrm{BIC}\,:=\,\frac{2}{\sigma}\sum_{i=1}^n \rho_{\tau}(y_i -\hat{\theta}_i) \,+\, v \log n,
\]
where as in \cite{brantley2019baseline} and \cite{ye2021non} we take 
$\sigma = \frac{1- \vert 1- 2\tau\vert  }{2}$,
and where $v$ is the estimated degrees of freedom. Motivated by \cite{tibshirani2012degrees}, we let
\[
v \,=\, \left\vert\left\{j\,:\, \vert\hat{\theta}_j -\hat{\theta}_{j+1}\vert > 10^{-3}\right\}\right\vert. 
\]

With the above choice of tuning parameter for both PQFL and QDCART, we compute the estimates which are displayed in Figure \ref{fig4}. There, we can see that  the estimators are roughly similar, validating our theoretical findings that QDCRT and PQFL have similar statistical properties.

Finally, in order to provide a clearer quantitative comparison,  we proceed as follows. We randomly choose 50\% of the entries of the signal described above and we use this as training. Then we use the remaining  50\% of the data as testing, and for each coordinate in the test set we make a  prediction based on the closest coordinate in the training set. With this in hand, for each competing method,  we compute $\text{prop}_{0.5}$, the proportion of the test samples that are below its predicted median. We also compute $\text{cov}_{80\%}$, the proportion of  samples in the test set that  are between their predicted $0.1$ and $0.9$ quantiles. The quantities  $\text{prop}_{0.5}$ and $\text{cov}_{80\%}$ are then averaged over 100 repetitions and reported for PQFL and QDCART.
For QDCART  we obtain the values $\text{prop}_{0.5} = 0.502$ and $  \text{cov}_{80\%} = 0.781$, whereas for PQFL  we obtain   $\text{prop}_{0.5} =  0.502$ and $  \text{cov}_{80\%} = 0.772$. These results suggest that QDCART provides ever so slightly better prediction intervals than PQFL.

\section{Proofs}\label{sec:proofs}


\subsection{General estimator}

Let $\mathcal{S}$ be a collection of linear subspaces of $\mathbb{R}^N$. For any subspace $S \in \mathcal{S}$ we denote its dimension with $\mathrm{Dim}(S)$ and define a penalty function $k_{ \mathcal{S} }: \mathbb{R}^{L_{d,n}} \rightarrow \mathbb{Z}_{+}$ induced by $\mathcal{S}$ as
\begin{equation}
\label{eqn:general_penalty2}
k_{ \mathcal{S} }(\theta)\,:=\,\min\{\mathrm{Dim}(S) \,:\,   \theta \in S, \,\,S\in \mathcal{S}   \},
\end{equation}
with the convention that the minimum of the empty set is $\infty$.  We are interested in subspaces of arrays which are piecewise constant on a rectangular partition of $L_{d,n}.$ 
We denote by $\Pi_{S}$ the rectangular partition of $L_{d,n}$ so that $S$ is the subspace of arrays which are constant on every rectangle of $\Pi_{S}.$ We will denote a generic rectangle of $L_{d,n}$ by $R.$ When we say $R \in \Pi_{S}$ we are referring to a rectangle $R$ of the partition $\Pi_{S}.$

The collection of partitions $\mathcal{P}_{rdp}$ or $\mathcal{P}_{ \text{hier} }$ will give rise to corresponding collections of linear subspaces $\mathcal{S}$ and the associated complexity measures $k_{rdp}(\theta)$ and $k_{tree}(\theta)$ respectively. For a collection of subspaces $\mathcal{S}$ corresponding to a collection of rectangular partitions, we now define an additional function $s_{\mathcal{S}}: \R^{L_{d,n}} \rightarrow \Z^{+}$ as follows:
\begin{definition}
	\begin{equation*}
	s_{ \mathcal{S} }(\theta)\,:=\, \underset{ R \,:\, R \in \Pi_S,  \,\,\theta \in S, S\in \mathcal{S},   \,\,\text{and}\,\, k_{\mathcal{S}}(\theta) = \vert \Pi_S\vert  \ }{\min}\,\vert  R \vert .
	\end{equation*}
\end{definition}
In words, $	s_{ \mathcal{S} }(\theta)$ is the size of the minimal rectangle of the minimal rectangular partition $\Pi_{S}$ within $S \in \mathcal{S}$ such that $\theta$ is constant on every rectangle of $\Pi_{S}.$


For a given collection of subspaces $\mathcal{S}$ corresponding to a collection of rectangular partitions of $L_{d,n}$ and a constant $c_1> 0$, we will consider the general $0 < \tau < 1$ quantile estimator
\begin{equation}
\label{eqn:general_estimator2}
\hat{\theta}^{(\tau)}_{\mathcal{S},c_1} \,= \,   \underset{\theta    \in  \Theta    }{\arg \min}\,\,\left\{   \sum_{i \in  L_{d,n}}\rho_{\tau}(y_i-\theta_i)  +    \lambda  k_{\mathcal{S}}(\theta)  \right\},
\end{equation}
where  $\lambda >0 $ is a tuning parameter and
\begin{equation}
\label{eqn:Theta2}
\Theta = \left(\cup_{S \in \mathcal{S}    } S  \right)\cap \left\{\theta \in \mathbb{R}^{L_{d,n}} \,:\,  s_{\mathcal{S}}(\theta) \geq c_1 \log N\right\},
\end{equation}



\subsection{Preliminary lemmas}


Throughout  the proof, without loss of generality we assume that  $\tau \leq 0.5$. The case $\tau>0.5$ can be handled similarly. Throughout  we will suppose that Assumption \ref{as1} holds. We will also drop the subscript in $\hat{\theta}_{\mathcal{S},c_1}$ and just write $\hat{\theta}$ to avoid notational clutter. 

\begin{definition}
	\label{def3}
	Let  $\theta^{\prime},\tilde{\theta} \in \mathbb{R}^{L_{d,n}}$  be  vectors  of $\tau/2$-quantiles and  $(1-\tau)/2$-quantiles  of the data vector $y$ respectively,  so that 
	\[
	\theta_i^{\prime}   =  \arg \min_{a \in  \mathbb{R}} \mathbb{E}\:\{\rho_{\tau/2}(y_i - a) \},
	\]
	and 
	\[
	\theta_i^{\prime \prime}  =  \arg \min_{a \in   \mathbb{R} } \mathbb{E}\:\{\rho_{(1-\tau)/2}(y_i - a) \}.
	\]
	Then we denote
	\[
	U\,:=\,\max\{ \| \theta^{\prime}\|_{\infty},  \|\theta^{\prime \prime}\|_{\infty}    \}.
	\]
\end{definition}


\begin{remark}
	The sequence $U$ is clearly a function of set of marginal distributions of $y$. For realistic distributions the sequence $U = O(1).$ For instance, if we assume that the error distribution is i.i.d then $U$ is a constant sequence. 
\end{remark}

\begin{lemma}\label{cor1}
	For any $\alpha > 0$, if we set $c_1 \geq \frac{2 \alpha}{\tau^2}$ then we have 
	\begin{equation*}
	\mathbb{P}(\|\hat{\theta}\|_{\infty} \leq U) \geq 1 - 2 N^{-\alpha}.
	\end{equation*}
\end{lemma}

\begin{proof}
	Let $\hat{\Pi}$ be the optimal partition on which $\hat{\theta}$ is piecewise constant. Fix a $v \in L_{d,n}$ and denote by $R$ the rectangle of the partition $\hat{\Pi}$ containing $v.$ We know that $\hat{\theta}_{v}$ is a sample $\tau$ quantile of the observations in $y_{R}$. Therefore, we can assert that 
	
	$$\left\{ \exists v   \in L_{d,n} \,:\,\hat{\theta}_{v} < \underset{u \in L_{d,n}}{\min}  \theta^{\prime}_u\right\} \subset \underset{R    \subset   L_{d,n}\,\,\,\text{rectangle}\,\,\,\,:   \vert  R \vert  \geq c_1 \log N   }{\bigcup}\left\{ \bigg\vert \bigg \{ i  \in R \,:\,    y_i   \leq    \underset{u \in L_{d,n}}{\min}  \theta^{\prime}_u     \bigg\}\bigg\vert  > \tau \vert  R\vert   \right\}$$. 
	
	Now for any rectangle $R$ with $\vert R \vert \geq  c_1 \log N$, 
	\[
	\begin{array}{lll}
	\displaystyle \mathbb{P}\left(   \bigg\vert \bigg \{ i  \in R \,:\,    y_i   \leq    \underset{j=1,\ldots,n}{\min}  \theta^{\prime}_j     \bigg\}\bigg\vert   >  \tau \vert  R\vert \right) &\leq  &	   \displaystyle 	\mathbb{P}\left( \sum_{u \in R} \mathrm{1}(y_u \leq \theta^{'}_{u}) > \tau \vert  R\vert \right) \\
	& \leq& \displaystyle \exp\left(  -\frac{\tau^2   \vert R \vert}{2}    \right) \leq \exp(-\frac{\tau^2   c_1 \log N}{2}),
	\end{array}
	\] 
	where the second inequality follows by  Hoeffding's inequality. By choosing $c_1 = 2 \alpha/\tau^2$ and by using  the last two displays imply that
	\begin{equation*}
	\mathbb{P}(  \exists v\,:\,\hat{\theta}_{v} < \underset{u \in L_{d,n}}{\min}  \theta^{\prime}_u) \leq N^{- \alpha}.
	\end{equation*}
	The same bound can readily be shown for $\mathbb{P}(\exists v\,:\,\hat{\theta}_{v} > \underset{u \in L_{d,n}}{\max}  \theta^{''}_u)$ by a similar argument. Both these assertions with a further union bound finishes the proof.

\end{proof}

Our next result is a modified version of  Lemma $9.1$ in \cite{chatterjee2019adaptive} where we prove the corresponding result for rademacher random variables instead of gaussian random variables.

\begin{lemma}[Lemma $9.1$  in~\cite{chatterjee2019adaptive}]
	\label{lem4}
	Let   $\tilde{\theta}\in \mathbb{R}^{L_{d,n}}$  be a fixed array and   $\xi \in \mathbb{R}^{L_{d,n}}$  be an array of  independent Rademacher random variables.  Then there exists $C > 0$ such that if $\lambda \geq  C \log N$  then it follows that 
	\[
	\mathbb{E}\bigg(    \underset{\theta \in \Theta  }{\sup}\left\{   \left(\xi^{\top}\frac{(\theta  - \tilde{\theta}) }{\|\theta  - \tilde{\theta}\|} \right)^2   - \lambda k_{\mathcal{S}}(\theta)    \right\}     \bigg) \leq 16.
	\]
\end{lemma}

\begin{proof}
	Let  $S \in \mathcal{S}$ and  $O_S$ be the orthogonal projection matrix onto $S$.  Then, following the arguments in the proof of  Lemma 9.1 in \cite{chatterjee2019adaptive} it follows that
	\begin{equation}
	\label{eqn:basic}
	\underset{v \neq0,  v \in S }{\sup}\,\frac{\xi^{\top} (v -\tilde{\theta} )}{\|v -\tilde{\theta}\| } \,\leq \,    \frac{ \xi^{\top}  (I-O_{S}) \tilde{\theta} }{ \|(I-O_{S}) \tilde{\theta}\|  }	\,+\, \underset{ v \in S,  \|v\|\leq 1 }{\sup}\,\xi^{\top}v  .
	\end{equation}
	Let  $v_{S,1},\ldots,v_{S,m_S}$ an orthonormal basis of $S$ with  $m_S = \text{dim}(S)$. Then
	\[
	\begin{array}{lll}
	\underset{ v \in S,  \|v\|\leq 1 }{\sup}\,(\xi^{\top}v )^2   & = & 	\displaystyle \underset{ v \in S,  \|v\|\leq 1 }{\sup}\,(\xi^{\top} O_S v)^2\\
	& \leq& \displaystyle \underset{ v \in S,  \|v\|\leq 1 }{\sup}\,\| O_S \xi \|^2  \cdot \|v\|^2\\
	& =& \displaystyle \| O_S \xi \|^2\\
	& =&  \displaystyle  \left\|    \sum_{j=1}^{ m_S  }    (v_{S,j}^{\top} \xi)   v_j   \right\|^2\\
	& = &\displaystyle  \sum_{j=1}^{ m_s  } \vert  v_{S,j}^{\top} \xi\vert^2\\
	& \leq&   m_S   \,\,\underset{j = 1,\ldots, m_S}{\max}   \vert  Z_{j}^{(S)}\vert^2
	\end{array}
	\] 
	where  $Z_{j}^{(S)}:=    v_{S,j}^{\top} \xi  $ is Sub-Gaussian($1$).  Therefore, 
	\begin{equation}
	\label{eqn:s1}
	\begin{array}{lll}
	\mathbb{E}\left(\underset{ v \in S,  \|v\|\leq 1 }{\sup}\,\xi^{\top}v \right)&\leq &\sqrt{ \text{dim}(S) } \cdot\mathbb{E}\left(\underset{j = 1,\ldots, m_S}{\max}   \vert  Z_{j}^{(S)}\vert \right)\\
	& \leq &\sqrt{ 2 \,\text{dim}(S) \log N },
	\end{array}
	\end{equation}
	where the second inequality holds by the usual expectation of maximum of sub-Gaussian random variables inequality.   Then from an application of  McDiarmid’s inequality  inequality as in Page 62 of \cite{van2014probability}, we obtain that  for any $t>0$,
	\begin{equation*}
	\begin{array}{lll}
	\displaystyle\mathbb{P}\left( 	\underset{ v \in S,  \|v\|\leq 1 }{\sup}\,(\xi^{\top}v )^2     -  4\,\text{dim}(S) \log N     \geq  2t  \right)  & \leq&\displaystyle\mathbb{P}\left( 	\underset{ v \in S,  \|v\|\leq 1 }{\sup}\,(\xi^{\top}v )   -   \sqrt{2\,\text{dim}(S) \log N  }    \geq  \sqrt{t}  \right)  \\
	& \leq&\displaystyle  \exp\left(  -\frac{t}{4} \right). 
	\end{array}
	\end{equation*}
	Hence, by union bound and the fact that  $\vert  \{S \in \mathcal{S}, \text{dim}(S) =k \}  \vert \leq  N^{2k}$, we have that 
	\begin{equation*}
	\begin{array}{lll}
	\displaystyle	 \mathbb{P}\left(  \underset{ \,\, S \in \mathcal{S}, \text{dim}(S) =k  }{\max}\, \left \{	\underset{ v \in S,  \|v\|\leq 1 }{\sup}\,(\xi^{\top}v )^2     -  20\,\text{dim}(S) \log N   \right\} \geq  2t \right)	 & \leq&\displaystyle  \exp\left(  -\frac{t}{4} \right). 
	\end{array}
	\end{equation*}
	And so again, by union bound we obtain that 
	\begin{equation}
	\label{eqn:p1}
	\begin{array}{lll}
	\displaystyle	\mathbb{P}\left(   \underset{k=1,\ldots,  N }{\max}\,\underset{ \,\, S \in \mathcal{S}, \text{dim}(S) =k  }{\max}\, \left \{	\underset{ v \in S,  \|v\|\leq 1 }{\sup}\,(\xi^{\top}v )^2     -  28\,\text{dim}(S) \log N   \right\} \geq  2t \right)	 & \leq&\displaystyle  \exp\left(  -\frac{t}{4} \right). 
	\end{array}
	\end{equation}
	Similarly, 
	\begin{equation}
	\label{eqn:p2}
	\begin{array}{lll}
	\displaystyle	\mathbb{P}\left(   \underset{k=1,\ldots,  N }{\max}\,\underset{ \,\, S \in \mathcal{S}, \text{dim}(S) =k  }{\max}\, \left \{	\left(  \frac{ \xi^{\top}  (I-O_{S}) \tilde{\theta} }{ \|(I-O_{S}) \tilde{\theta}\|  }\right)^2    -  28\,\log N   \right\} \geq  2t \right)	 & \leq&\displaystyle  \exp\left(  -\frac{t}{4} \right). 
	\end{array}
	\end{equation}
	The claim follows  from  (\ref{eqn:p1}) and (\ref{eqn:p2}) by simple integration.
	
\end{proof}

Next, we recall some notations  from Section \ref{sec:proof}.

\begin{definition}
	For $u \in L_{d,n}$, define the random function $\hat{M}_u: \mathbb{R} \rightarrow \mathbb{R}$ as follows:
	\begin{equation*}
	\hat{M}_u(\theta_u)   :=     \left\{\rho_{\tau}(y_u -   \theta_u) -   \rho_{\tau}(y_u -   \theta_u^* ) \right\},
	\end{equation*}
	Now define the random function $\hat{M}: \mathbb{R} \rightarrow \mathbb{R}$
	\[
	\hat{M}(\theta)   :=    \sum_{u \in  L_{d,n}}  	\hat{M}_u(\theta_u)  ,
	\]
	and the deterministic function $M: \mathbb{R} \rightarrow \mathbb{R}$ as
	\[
	M(\theta)   := \mathbb{E}\left(	\hat{M}(\theta) \right).
	\] 
	Furthermore, let us denote
	\[
	\Delta^2(\theta)  =    \sum_{i \in  L_{d,n}}   \min\{  \vert \theta_i \vert,   \vert \theta_i\vert^2  \}
	\]
	and  $\Delta^2_N(\theta) =   \Delta^2(\theta)/N$.
\end{definition}


Our analysis relies on viewing the estimator defined in~\eqref{eqn:general_estimator2} as a penalized M estimator or a penalized empirical risk minimization estimator. Hence the natural loss function for us is the population quantile loss $M$ function given above. However, we would like to give risk bounds for the square loss function. For this purpose, the $\Delta$ function defined above plays an important role in converting bounds in the $M$ loss function to bounds for squared error loss.

 We now proceed to state some results (Lemmas $8$--$11$) involving involving the functions $M(\cdot) $ and $\Delta^2(\cdot)$. These are results that also appeared in~\cite{padilla2020risk}, the only difference with the results in \cite{padilla2020risk} is that we now use the penalty function $k_{\mathcal{S}}(\theta)$ instead of the $\text{TV}(\theta)$ function. We omit writing the proofs of these results since the proofs are very similar to what is already given in~\cite{padilla2020risk}.

\begin{lemma}
	\label{lem28}
	(Symmetrization, Lemma  28 in  \cite{padilla2020risk}). For any set  $K$, $\tilde{\theta} \in \mathbb{R}^{L_{d,N}}$, and  $\lambda >0$ it holds that
	\[
	\begin{array}{l}
	\displaystyle E\left[  \underset{\theta \in K}{\sup}\,\,\left\{   M(  \theta  )    -   M(\tilde{\theta})    +   \hat{M}(\tilde{\theta}) - \hat{M}(\theta) +  \lambda(   k_{\mathcal{S}}(\tilde{\theta}) -  k_{\mathcal{S}}(\theta) )      \right\}   \right]\\
	\leq   	\displaystyle2 \,E\left[   \underset{\theta  \in K}{\sup}\,\,\left\{  \sum_{i   \in L_{d,n} } \xi_i (\hat{M}_{i}(\theta_i)  - \hat{M}_{i}(\tilde{\theta}_i)  )   +    \frac{\lambda }{2} k_{\mathcal{S}}(\tilde{\theta}) -  \frac{\lambda }{2}k_{\mathcal{S}}(\theta)   \right\}    \right], 
	\end{array}
	\]
	where  $\xi_1,\ldots,\xi_n$ are   independent  Rademacher variables  independent  of  $\{y_i\}_{i=1}^n$.
\end{lemma}

\begin{lemma}
	\label{lem29}
	(Contraction principle, Lemma 29 in  \cite{padilla2020risk}). Let  $h_1,\ldots,h_n \,:\,  \mathbb{R} \rightarrow \mathbb{R}$   $\eta$-Lipschitz functions  for some $\eta>0$.  Then for any   $\tilde{\theta} \in \mathbb{R}^{L_{d,n}}$, any  compact set $K$, and   $\xi_1,\ldots,\xi_n$ independent Rademacher variables we have that 
	\[
	\begin{array}{l}
	\displaystyle   E\left[  \underset{\theta  \in K}{\sup}\,\,\left\{ \sum_{i   \in L_{d,n} } \xi_i h_i(\theta_i) \  +    \frac{\lambda}{2}k_{ \mathcal{S} }(\tilde{\theta})-  \frac{\lambda}{2}k_{ \mathcal{S} }(\theta)    \right\}\right]\\
	\displaystyle \leq \, 	  E\left[  \underset{\theta \in K}{\sup}\,\,\left\{    \eta \sum_{i   \in L_{d,n} } \xi_i \theta_i \  +    \frac{\lambda}{2}k_{ \mathcal{S} }(\tilde{\theta}) -  \frac{\lambda}{2} k_{ \mathcal{S} }(\theta)  \right\}\right]\\
	\end{array}
	\]
	for any $\lambda>0$.
\end{lemma}

\begin{lemma}
	\label{lem7}
	(Lemma 13 in  \cite{padilla2020risk}).	Suppose  that Assumption  \ref{as1} holds. Then there exists  a constant   $c_0$ such that  for all  $\delta \in R^n$, we have
	\[
	\displaystyle  M(\theta^*+\delta) \geq   c_0 \Delta^2(\delta) .
	\]
\end{lemma}

\begin{lemma}
	\label{lem8}
	(Lemma 17 in  \cite{padilla2020risk}).	Let  $\delta   \in   R^n $. Then 
	\begin{equation}
	\label{eqn:ine}
	\|  \delta\|^2   \,\leq \,  \max\{ \| \delta\|_{\infty},1  \} \Delta^2( \delta).
	\end{equation}
\end{lemma}


Our next lemma is key and controls the expected suprema of a penalized empirical process. 
\begin{lemma}
	\label{lem3}
	Let  $\tilde{\theta} \in \mathbb{R}^n$ and $t>0$. Then there exist a constant  $C>0$ such that 
	  for any  $a>0$  if  
	$ \lambda \,\geq \,  C  a\log N,$ we have that 
	\[
	\begin{array}{l}
	\mathbb{E}\left(    \underset{\theta \in \Theta  \,:\,  \|\theta-\theta^*\|^2  \leq  t^2  }{\sup}\left\{  M(  \theta  )    -   M(\tilde{\theta})    +   \hat{M}(\tilde{\theta}) - \hat{M}(\theta) +  \lambda(   k_{\mathcal{S}}(\tilde{\theta}) -  k_{\mathcal{S}}(\theta) )    \right\}     \right)   \\
	\displaystyle 	 \,\leq \,  C_2 a   +  \frac{2t^2 }{a}+  \frac{2\|  \theta^* -  \tilde{\theta}\|^2 }{a}   
	+\lambda   k_{\mathcal{S}}(\tilde{\theta}), 
	\end{array}
	\]
	for a positive constant  $C_2>0$.
\end{lemma}



\begin{proof}
	Notice that  if  $\xi \in \mathbb{R}^{L_{d,n}}$  consists of independent Rademacher random variables then 
	\[
	\begin{array}{l}
	\displaystyle 	\mathbb{E}\left(    \underset{\theta \in \Theta  \,:\,  \|\theta-\theta^*\|  \leq  t   }{\sup}\left\{   M(  \theta  )    -   M(\tilde{\theta})    +   \hat{M}(\tilde{\theta}) - \hat{M}(\theta) +  \lambda(   k_{\mathcal{S}}(\tilde{\theta}) -  k_{\mathcal{S}}(\theta) )      \right\}     \right) \\
	\displaystyle    	 \leq   \,2\, \mathbb{E}\left(    \underset{\theta \in \Theta  \,:\,  \|\theta-\theta^*\|  \leq  t   }{\sup}\left\{  \xi^{\top}(\theta  - \tilde{\theta}) +\frac{  \lambda}{2}(   k_{\mathcal{S}}(\tilde{\theta}) -  k_{\mathcal{S}}(\theta) )    \right\}     \right) \\
	\displaystyle   	   \leq\, 2 \,\mathbb{E}\left(    \underset{\theta \in \Theta  \,:\,  \|\theta-\theta^*\|  \leq  t   }{\sup}\left\{  \xi^{\top}(\theta  - \tilde{\theta})  - \frac{\lambda}{2} k_{\mathcal{S}}(\theta)    \right\}     \right) +  \lambda   k_{\mathcal{S}}(\tilde{\theta})  \\
	\displaystyle   	    \leq \, 2 \,\mathbb{E}\bigg(    \underset{\theta \in \Theta  \,:\,  \|\theta-\theta^*\|  \leq  t   }{\sup}\left\{  \frac{a}{2} \left(\xi^{\top}\frac{(\theta  - \tilde{\theta}) }{\|\theta  - \tilde{\theta}\|} \right)^2   -\frac{ \lambda}{2} k_{\mathcal{S}}(\theta)     \right\}     \bigg) +   \frac{2t^2 +  2 \|  \theta^* -  \tilde{\theta}\|^2 }{a}  +\lambda   k_{\mathcal{S}}(\tilde{\theta})  \\
	\displaystyle 	     \leq  \,a\, \mathbb{E}\bigg(    \underset{\theta \in \Theta  \,:\,  \|\theta-\theta^*\|  \leq  t   }{\sup}\left\{   \left(\xi^{\top}\frac{(\theta  - \tilde{\theta}) }{\|\theta  - \tilde{\theta}\|} \right)^2   - C k_{\mathcal{S}}(\theta) \log N     \right\}     \bigg) +   \frac{2t^2 }{a} +  \frac{2\|  \theta^* -  \tilde{\theta}\|^2 }{a} +  \lambda   k_{\mathcal{S}}(\tilde{\theta})  \\
	\displaystyle      \leq\,  C_2 a   +  \frac{2t^2 }{a}+  \frac{2\|  \theta^* -  \tilde{\theta}\|^2 }{a} +  \lambda   k_{\mathcal{S}}(\tilde{\theta}) 
	\end{array}
	\]
	where the first inequality follows as in Lemmas  \ref{lem28} and \ref{lem29}, the third by Cauchy Schwarz inequality, and the last by Lemma \ref{lem4}.
\end{proof}

\begin{lemma}
	\label{lem5}
	For  $\tilde{\theta} \in \mathbb{R}^N$ we have that
	\[
	M(\tilde{\theta})    \leq 	\frac{\overline{f}}{2}\|\tilde{\theta} - \theta^*\|^2.
	\]
\end{lemma}

\begin{proof}
	Let  $\delta :=\tilde{\theta} - \theta^*$.	We start by recalling   by Equation (19)  in 	\cite{padilla2020risk}  which states that
	\[
	M_i( \tilde{\theta}_i ) =   \int_{0}^{\delta_i}  (  F_{y_i}(\theta_i^*+z) - F_{y_i}(\theta_i^*) )dz.
	\]
	Hence,
	\[
	\begin{array}{lll}
	M(\tilde{\theta})         &=&  \displaystyle  \sum_{i \in  L_{d,n}}  \int_{0}^{\delta_i}    \left( F_{y_i}(\theta^*+z) -    F_{y_i}(\theta^*) \right)dz\\
	& \leq& \displaystyle  \sum_{i \in  L_{d,n}}  \int_{0}^{\delta_i}     \overline{f} z  dz\\ 
	&= & \displaystyle   \frac{\overline{f}}{2}\| \delta\|^2
	\end{array} 
	\]
	where the inequality follows from Assumption \ref{as1}.
\end{proof}

\subsection{General upper bound}

\begin{theorem}
	\label{thm:general}
Suppose that   Assumption~\ref{as1} holds. There exists universal constants $c_1,C_1, C_2,C_3>0$ such that for any $0 < \epsilon < 1$, if we set $\gamma = c_1 \log N $ and
	\[
	\lambda \,=\,   C_1 \frac{\max\{ 1,U   \} \log(N)\log(N U)}{\epsilon},
	\]
	implies that  with probability at least  $1-C_2\epsilon$,
	\begin{equation}
	\label{eqn:main2}
	\frac{\| \hat{\theta}-\theta^* \|^2}{N}  \leq  \frac{C_3Q(\theta^*)}{\epsilon^2},
	\end{equation}
	where
	\[
	Q(\theta^*)\,:=\, \underset{\theta \in \Theta}{\inf}\left\{\frac{ k_{\mathcal{S}}(\theta) \max\{1,U^2\}  \log^2 \left(\max\{N,U\}  \right) }{N}\ +  \frac{\overline{f}  \|  \theta^* -\theta   \|^2 }{N}  \right\}.
	\]
\end{theorem}


\begin{proof}
	Let  $t \in  (0,2N\:U)$ and notice that  for  $\hat{\delta } := \hat{\theta} - \theta^*$ we have that  for  $U$ as in  Definition \ref{as2} it holds that 
	\begin{equation}
	\label{eqn:first}
	\mathbb{P}(  \Delta^2(  \hat{\delta } )>t^2  )   \,\leq  \,\mathbb{P}(  \Delta^2(  \hat{\delta } )>t^2 ,    \|\hat{\theta} \|_{\infty}  \leq  U )  +  \mathbb{P}(    \|\hat{\theta} \|_{\infty}  > U ), 
	\end{equation}
	with $\Delta(\cdot)$ as in Definition \ref{def3}.
	Hence,  from Lemma~\ref{cor1} it is enough to bound $\mathbb{P}(  \Delta^2(  \hat{\delta } )>t^2 ,    \|\hat{\theta} \|_{\infty}  \leq  U) $. Towards that end we notice that $ \|\hat{\theta} \|_{\infty} \leq  U$ implies 
	\begin{equation}
		\label{eqn:inf}
		\Delta^2(  \hat{\delta } ) \,\leq  \,\|\hat{\delta } \|_1 \,\leq\, \|\hat{\theta}\|_1  +   \|\theta^*\|_1  \leq  4N \:U
	\end{equation}
	 and hence
	\[
	\begin{array}{lll}
	\mathbb{P}(  \Delta^2(  \hat{\delta } )>t^2 ,    \|\hat{\theta} \|_{\infty} \leq  U)  & \leq &  	 \mathbb{P}(  \Delta^2(  \hat{\delta } )>t^2,  \Delta^2(  \hat{\delta } ) \leq  4N \:U ,    \|\hat{\theta} \|_{\infty} \leq  U) . 
	\end{array}
	\]

	Now, we will undertake the so called peeling step. Letting 
	\begin{equation}
	\label{eqn:defx}
	x \,=\,\ceil{  \log(4N \:U /t^2)  / \log  2 },
	\end{equation}
	we obtain that for  any $\tilde{\theta} \in \Theta$ 
	\begin{equation*}
	\begin{array}{lll}
	\mathbb{P}(  \Delta^2(  \hat{\delta } )>t^2 ,    \|\hat{\theta} \|_{\infty} \leq  U )  	 		 & \leq&\displaystyle \sum_{j=1}^{ x } \mathbb{P}\left( \Delta^2(  \hat{\delta } )> 2^{j-1} t^2,  \Delta^2(  \hat{\delta } ) \leq 2^j t^2 ,      \|\hat{\theta} \|_{\infty} \leq  U  \right)\\
	& \leq&\displaystyle \sum_{j=1}^{ x } \mathbb{P}\left(  M(  \hat{\theta} )> c _02^{j-1} t^2,  \Delta^2(  \hat{\delta } ) \leq 2^j t^2 ,      \|\hat{\theta} \|_{\infty} \leq  U\right)\\
	& = & \displaystyle\sum_{j=1}^{ x } \mathbb{P}(  M(  \hat{\theta } )    -   M(\tilde{\theta})   + M(\tilde{\theta})  > c _02^{j-1} t^2,  \Delta^2(  \hat{\delta } ) \leq  2^j t^2 ,    \|\hat{\theta} \|_{\infty} \leq  U ) \\
	& \leq  & \displaystyle\sum_{j=1}^{ x } \mathbb{P}(  M(  \hat{\theta } )    -   M(\tilde{\theta})    +   \left\{\hat{M}(\tilde{\theta}) - \hat{M}(\hat{\theta}) +   \lambda(k_{\mathcal{S}}(\tilde{\theta})  - k_{\mathcal{S}}(\hat{\theta})) \right\}  +  M(\tilde{\theta})  \\
	&& \,\,\,\,\,\,\,\,\,\,\,\,\,  > c _02^{j-1} t^2,\Delta^2(  \hat{\delta } ) \leq 2^j t^2 ,    \|\hat{\theta} \|_{\infty} \leq  U ) \\		
	\end{array}
	\end{equation*}
	where the second inequality follows from Lemma \ref{lem7}, and the last by the optimality of  $\hat{\theta}$, since
	\[
	\hat{M}(\tilde{\theta}) - \hat{M}(\hat{\theta}) +   k_{\mathcal{S}}(\tilde{\theta})  - k_{\mathcal{S}}(\hat{\theta})   \geq 0.
	\]
	We can now continue to write from the previous display,
	\begin{equation}
	\label{eqn:main}
	\begin{array}{lll}
	\mathbb{P}(  \Delta^2(  \hat{\delta } )>t^2 ,    \|\hat{\theta} \|_{\infty} \leq  U )   & \leq &\displaystyle \sum_{j=1}^{ x } \mathbb{P}\bigg( \underset{ \theta \in \Theta \,:\,\Delta^2(  \theta-\theta^* ) \leq 2^j t^2 ,      \|\theta \|_{\infty} \leq  U }{\sup}\bigg\{ M(  \theta )    -   M(\tilde{\theta})    +   \hat{M}(\tilde{\theta}) - \hat{M}(\theta) +\\
	& &\,\,\,\,  \,\,\,\,\, \,\,\,\,\,\, \,\,\,    \lambda (k_{\mathcal{S}}(\tilde{\theta}) - k_{\mathcal{S}}(\theta)) + M(\tilde{\theta})  \bigg\}\geq  c_0 2^{j-1} t^2    \bigg)\\
	& \leq&  \displaystyle \sum_{j=1}^{ x } \mathbb{P}\bigg( \underset{ \theta \in \Theta \,:\,\| \theta-\theta^* \|^2 \leq 2^{j} \max\{1,U\} t^2 }{\sup}\bigg\{  M(  \theta  )    -   M(\tilde{\theta})    +   \hat{M}(\tilde{\theta}) - \hat{M}(\theta)  +   \\
	& &\,\,\,\,  \,\,\,\,\, \,\,\,\,\,\, \,\,\,  \lambda (k_{\mathcal{S}}(\tilde{\theta}) - k_{\mathcal{S}}(\theta))  +M(\tilde{\theta})\bigg\}\geq   c_0 2^{j-1} t^2    \bigg)\\
	& \leq&\displaystyle \sum_{j=1}^{ x }    \frac{1}{c_0 2^{j-1} t^2 }\mathbb{E}\bigg(  \underset{ \theta \in \Theta \,:\,\| \theta-\theta^* \|^2 \leq 2^{j} \max\{1,U\}t^2 }{\sup}\bigg\{  M(  \theta )    -   M(\tilde{\theta})    +  \\
	& &\displaystyle\,\,\,\,  \,\,\,\,\, \,\,\,\,\,\, \,\,\, \,\,\,\,  \,\,\,\,\, \,\,\,\,\,\, \,\,\,\,\, \,\,\, \,\, \,\,\,   \hat{M}(\tilde{\theta}) - \hat{M}(\theta)  +  \lambda (k_{\mathcal{S}}(\tilde{\theta}) - k_{\mathcal{S}}(\theta))   \bigg\}\bigg)\,+\, \frac{2M(\tilde{\theta}) }{c_0t^2}\\
	\end{array}
	\end{equation}
	where the second inequality follows from Lemma \ref{lem8}, and the third  inequality follows from Markov's inequality and summing up the geometric series.

	Let  $\epsilon \in (0,1)$ be fixed.  We notice that  (\ref{eqn:main}) and Lemma \ref{lem5} imply that
	\[
	\begin{array}{lll}
	\mathbb{P}(  \Delta^2(  \hat{\delta } )>t^2 ,    \|\hat{\theta} \|_{\infty} \leq  U )   
	& \leq&\displaystyle \sum_{j=1}^{ x }    \frac{1}{c_0 2^{j-1} t^2 }\mathbb{E}\bigg(  \underset{ \theta \in \Theta \,:\,\| \theta-\theta^* \|^2 \leq 2^{j}  \max\{1,U\} t^2 }{\sup}\bigg\{  M( \theta )    -   M(\tilde{\theta})    +    \\
	& &\displaystyle \,\,\,\,  \,\,\,\,\, \,\,\,\,\,\, \,\,\, \,\,\,\,  \,\,\,\,\, \,\,\,\,\,\, \,\,\,\,\, \,\,\, \,\, \,\,\, \hat{M}(\tilde{\theta}) - \hat{M}(\theta)  +  \lambda (k_{\mathcal{S}}(\tilde{\theta}) - k_{\mathcal{S}}(\theta))   \bigg\}\bigg)\,+\,  \frac{ 2\overline{f} }{c_0 t^2}\|\tilde{\theta} -\theta^*\|^2.\\
	\end{array}
	\]
	Next, for some $a>0$ to be chosen later, we set $\lambda =   C a \log N$. Hence,
	from Lemma  \ref{lem3},   we have that 
	\begin{equation}
	\label{eqn:main3}
	\begin{array}{lll}
	\mathbb{P}(  \Delta^2(  \hat{\delta } )>t^2 ,    \|\hat{\theta} \|_{\infty} \leq  U )   
	& \leq&\displaystyle  \sum_{j=1}^{ x }    \frac{1}{c_0 2^{j-1} t^2 }\left[   C_2 a   +  \frac{2^{j+1}  \max\{1,U\}  t^2+2\|\tilde{\theta} -\theta^*\|^2     }{a}  
	+\lambda   k_{\mathcal{S}}(\tilde{\theta})    \right]\\
	& &\displaystyle \,+\,  \frac{2 \overline{f} }{c_0t^2}\|\tilde{\theta} -\theta^*\|^2\\
	& \leq& \displaystyle  \frac{2C_2  a }{c_0 t^2}  + \frac{4 \max\{1,U\}   x   }{ac_0 } +  \frac{2 \lambda   k_{\mathcal{S}}(\tilde{\theta}) }{c_0 t^2}   +  \frac{1}{t^2}\left( \frac{4}{ac_0}  +  \frac{2 \overline{f} }{c_0}\right)\|\tilde{\theta} -\theta^*\|^2. 
	\end{array}
	\end{equation}
	Therefore,  setting
	\[
	a \,:=\, \frac{\max\{ 1,U   \} \log(N U)}{\epsilon},
	\]  
	\[
	t^2 \,:=\,  \frac{ C \max\{1,U,C_2\} k_{\mathcal{S}}(\tilde{\theta})  \log^2 (N U)}{\epsilon^2 c_0} +  \frac{ \overline{f}\|\tilde{\theta} -\theta^*\|^2}{2\epsilon}, 
	\] 
	and letting 
	\[
	\tilde{\theta}  \in \underset{\theta     }{\arg \min } \left\{ \frac{ k_{\mathcal{S}}(\theta) \max\{1,U^2\}  \log^2 \left(\max\{N,U\}  \right) }{N}\ +  \frac{\overline{f}  \|  \theta^* -\theta   \|^2 }{N} \right\}
	\]
	we obtain the conclusion in (\ref{eqn:main2}) by combining  (\ref{eqn:first}) and (\ref{eqn:main3}).

\end{proof}

\subsection{Proof of Theorem \ref{thm0}}

In the rest of the proofs we denote  $\hat{\theta}_{rdp} $  simply as $\hat{\theta}$. 

\begin{proof}
	From Theorem \ref{thm:general} we obtain that
	
	\begin{equation}
	\label{eqn:upper6}
	\frac{\| \hat{\theta}-\theta^* \|^2}{N}\,=\, O_{\mathbb{P}}\left(   \underset{\tilde{\theta}  \in \Theta }{\inf}  \left\{\frac{\overline{f}  \| \tilde{\theta}   - \theta^* \|^2 }{N}\,+\, \frac{ k_{rdp}(\tilde{\theta})\max\{1,U^2\}\log^2 \left(\max\{N,U\}\right)  }{N}  \right\} \right).
	\end{equation}
	
	However, for any $\theta \in \mathbb{R}^N$   by Lemma \ref{lem6} there exists  $A(\theta ) \in  \Theta$  such that  $k_{rdp}(A(\theta )) \leq k_{rdp}(\theta)$ and   
	\[
	\|A(\theta )-\theta\|^2 \leq 4c_1\| \theta\|_{\infty}^2    k_{rdp}(\theta)\log N.
	\] 
	It follows that 
	\begin{equation}
	\label{eqn:upper7}
	\begin{array}{l}
	\displaystyle \underset{\tilde{\theta}  \in \Theta }{\inf}  \left\{\frac{\overline{f}  \| \tilde{\theta}   - \theta^* \|^2 }{N}\,+\, \frac{ k_{rdp}(\tilde{\theta})\max\{1,U^2\}\log^2 \left(\max\{N,U\}\right)  }{N}  \right\}   \\
	\leq \displaystyle \underset{\theta  \in \mathbb{R}^N }{\inf}  \left\{\frac{\overline{f}  \| A(\theta)   - \theta ^*\|^2 }{N}\,+\,\frac{ k_{rdp}(A(\theta))\max\{1,U^2\}\log^2 \left(\max\{N,U\}\right)  }{N}  \right\} \\
	\leq \displaystyle \underset{\theta  \in \mathbb{R}^N }{\inf}  \left\{\frac{2\overline{f}  \| A(\theta)   - \theta \|^2 }{N}\,+\, \frac{2\overline{f}  \| \theta  - \theta^* \|^2 }{N} +\frac{ k_{rdp}(A(\theta))\max\{1,U^2\}\log^2 \left(\max\{N,U\}\right)  }{N}  \right\} \\
	\leq  \displaystyle \underset{\theta  \in \mathbb{R}^N }{\inf}  \left\{\frac{8\overline{f}  c_1\| \theta\|_{\infty}^2    k_{rdp}(\theta)\log N }{N}\,+\, \frac{2\overline{f}  \| \theta  - \theta^* \|^2 }{N} +\frac{ k_{rdp}(\theta)\max\{1,U^2\}\log^2 \left(\max\{N,U\}\right)  }{N}  \right\} \\
	\end{array}
	\end{equation}
	where the second inequality follows from the Cauchy–Schwarz inequality, and the third  by the construction of $A(\cdot)$. The claim follows combining (\ref{eqn:upper6}) with (\ref{eqn:upper7}).
	
\end{proof}

\subsection{Other lemmas}

\begin{lemma}
	\label{lem6}
	Let  $\theta \in \mathbb{R}^{L_{d,n}}$. Given $c_1 >0$ there exists a $\tilde{\theta} \in \mathbb{R}^{L_{d,n}}$ such that the follwing holds:
	\begin{itemize}
		\item  $k_{rdp}(\tilde{\theta})\leq k_{rdp}(\theta)$.
		\item 
		\[
		\|\tilde{\theta}-\theta\|^2 \leq 4c_1\| \theta\|_{\infty}^2    k_{rdp}(\theta)\log N.
		\]
		\item$s(\tilde{\theta}) \geq c_1 \log N$, where  $s(\cdot)$ is the $s_{\mathcal{S}}(\cdot)$      corresponding to Dyadic partitions.
	\end{itemize}
	
\end{lemma}
\begin{proof}
	Let $\Pi$ a  minimal dyadic partition induced by $\theta$. Then 
	consider $\tilde{\Pi}$ the  dyadic partition obtained by performing the same splits as in the construction of $\Pi$ but only when each split produces rectangles of size at least $c_1\log N$. Then let $\tilde{\theta}$ be constructed by  averaging the  values of $\theta$ on each rectangle of $\tilde{\Pi}$.  Notice that  by construction the first and third claims of the lemma hold. To see why the second claim holds, we observe that  $\Pi$ and $\tilde{\Pi}$ differ in at most $k_{rdp}(\tilde{\theta})$ rectangles each of which is of size at most $2 c_1\log N$. The claim then follows.
	
\end{proof}

\subsection{Proof of Corollary \ref{cor2} }
\label{sec:thm1}

\begin{proof}
	\textbf{Case $d=1$.}

	We proceed in two cases. First, if $V =\mathrm{TV}(\theta^*)=0$ then $k_{rdp}(\theta^*)=1$ and
	Theorem 
	\ref{thm0} implies that 
	\begin{equation}
	\label{eqn:part1}
	\frac{1}{N}\sum_{i \in L_{d,N}}  (\hat{\theta}_i - \theta_i^*)^2 \,=\, O_{\mathbb{P}}\left(  \frac{ \max\{1,U^2\} \log^2 \{N,U\}  }{N}\right).
	\end{equation}
	Suppose now that  $V >0$. Then by Proposition 8.9  in \cite{chatterjee2019adaptive}, for any $\eta>0$ there exits $\theta$ such that for some positive constant $C$ it holds that $k_{rdp}(\theta) \leq C \eta^{-1} $ and 
	\[
	\| \theta - \theta^*\|_{\infty} \,\leq \,  V \eta. 
	\]
	Then  notice that $\|\theta\|_{\infty}^2 \leq   2V^2 \eta^2 +2 U^2$  and  $\| \theta-\theta^*\|^2 \leq  V^2 \eta^2 N$.

		Next,  we  set 
	\[
	\eta \,:=\,\frac{1}{V^{2/3}}\left(\frac{\log N}{N}\right)^{1/3}
	\]
	and notice that 
	
	\begin{itemize}
		\item  
		\begin{equation}
			\label{eqn:c1}
			   \begin{array}{lll}
				\displaystyle 	    \frac{\|\theta\|_{\infty}^2  k_{rdp}(\theta) \log N }{N}&\leq &         		\displaystyle 	    2 \left( V^2 \eta^2 + U^2\right) C  \eta^{-1} \log N\\ 
				&\leq  &        		\displaystyle 	   \frac{2C   V^2 \eta \log N}{N}   +     \frac{2C U^2 \eta^{-1}\log N  }{N}\\
				&  =& \displaystyle O\left(   \frac{U^2  V^{2/3}  \log^{2/3} N}{N^{2/3}}    \right).
			\end{array}
		\end{equation}
	\item 
	\[
\frac{	\| \theta - \theta^*\|^2}{N}\,=\,  V^2 \eta^2 \,=\,   V^2\left(\frac{1}{V^{2/3}}\left(\frac{\log N}{N}\right)^{1/3} \right)^2  \, =\, \frac{ V^{2/3}  \log^{2/3} N  }{N^{2/3} }.
	\]
	\item   \[
	\begin{array}{lll}
					\displaystyle 	  	\frac{ k_{rdp}(\theta) \max\{1,U^2\}\log^2 \left(\max\{N,U\}\right)  }{N}  &\leq& 				\displaystyle 	   \frac{C   V^{2/3}  N^{1/3}( \log^{-1/3} N )  \max\{1,U^2\}  \log^2\max\{N,U\}   }{N}\\ 
		&=&				\displaystyle 	  \frac{C V^{2/3}  \max\{1,U^2\} \log^{5/3} \max\{N,U\}}{N^{2/3}}.
	\end{array}
	\] 
	\end{itemize}
	Combining the cases above we obtain the claim for  $d=1$.

			\textbf{Case $d>1$.} If case $V=0$ we proceed as we did in the previous case. Suppose now that $V>0$.  Then, by the proof of Theorem 4.2 in \cite{chatterjee2019adaptive}, for any $\eta>0$ there exists a $\theta$  such that  
			\[
			k_{rdp}(\theta) \leq \frac{C \,\text{TV}(\theta^*)\,  \log N }{\eta},
			\]
			\[
			 \|\theta \|_{\infty}\leq \|\theta^*\|_{\infty},
			\]
			and
			\[
			\| \theta - \theta^*\|^2 \leq  C \eta\,  \text{TV}(\theta^*) \, \log N
			\]
			for some positive constant $C$. 
			
			Next, let  $\eta = \log N$  and notice that
			\begin{itemize}
				\item 
				\begin{equation}
					\label{eqn:c4}
				\begin{array}{lll}
							\displaystyle 	    \frac{\|\theta\|_{\infty}^2  k_{rdp}(\theta) \log N }{N}&\leq &         		\displaystyle 	   \frac{C  V   \log N  }{\eta}  \cdot\frac{U^2  \log N}{N} \\
							 & =&  	\displaystyle 	 \frac{C V U^2\log N}{N}.
						\end{array}
				\end{equation}
			
			\item 
			\begin{equation}
				\label{eqn:c5}
				\begin{array}{lll}
					\displaystyle  \frac{	\| \theta - \theta^*\|^2}{N}  & =&	\displaystyle \frac{C V\log^2 N}{N}.
				\end{array}
			\end{equation}
				
				\item 
				\begin{equation}
				\label{eqn:c6}
					\begin{array}{lll}
					\displaystyle 	  	\frac{ k_{rdp}(\theta) \max\{1,U^2\}\log^2 \left(\max\{N,U\}\right)  }{N}  &\leq& 				\displaystyle 	   \frac{CV \max\{1,U^2\}\log^2\left(\max\{N,U\}\right)  }{N}.
				\end{array}
				\end{equation}
			\end{itemize}
		
		Therefore, combining (\ref{eqn:c4})--(\ref{eqn:c6}) the claim follows.

\end{proof}

	\newpage




\bibliographystyle{plainnat}
\bibliography{references}	
\end{document}